\documentclass[a4paper,12pt,notitlepage]{article}
\usepackage{amsfonts}

\usepackage{amsmath}
\usepackage{graphicx,lscape,color}
\usepackage{epstopdf}

\usepackage{amsfonts}
\usepackage{amsthm,bbm}
\usepackage{multirow}

\usepackage[normalem]{ulem}

\usepackage[breaklinks,hypertexnames=false]{hyperref}
\usepackage[dvipsnames]{xcolor}

\allowdisplaybreaks

\hypersetup{
	colorlinks=true,
	citecolor=Blue,
	linkcolor=Blue,
	urlcolor=Blue
}

\usepackage{rotating}
\usepackage[english]{babel}
\usepackage[latin1]{inputenc}
\usepackage{latexsym}
\usepackage{amssymb}
\usepackage{enumerate}

\newtheorem{theorem}{Theorem}

\newtheorem{assumption}{Assumption}

\newtheorem{corollary}{Corollary}

\newtheorem{lemma}{Lemma}

\evensidemargin\oddsidemargin
\input{tcilatex}

\newcommand{\argmin}{\operatorname*{argmin}}

\begin{document}

\title{Minimizing Sensitivity to Model Misspecification\thanks{%
We thank Kirill Evdokimov, Josh Angrist, Tim Armstrong, Gary Chamberlain, Tim Christensen, Ben Connault, Max Farrell, Jin Hahn, Chris Hansen, Lars Hansen, Kei Hirano, Max Kasy, Roger Koenker, Thibaut Lamadon, Esfandiar Maasoumi, Magne Mogstad, Roger Moon, Whitney Newey, Tai Otsu, Franco Peracchi, Jack Porter, Andres Santos, Azeem Shaikh, Jesse Shapiro, Richard Smith, Alex Torgovistky, and Ken Wolpin, as well as the audiences in various seminars and conferences, for comments. Bonhomme acknowledges support from the NSF, Grant SES-1658920. Weidner acknowledges support 
from the Economic and Social Research Council through grants (RES-589-28-0001, RES-589-28-0002 and ES/P008909/1)
 to the ESRC Centre for Microdata Methods and Practice (CeMMAP),
 and from the European Research Council grants  ERC-2014-CoG-646917-ROMIA and
ERC-2018-CoG-819086-PANEDA.}}
\author{St\'ephane
Bonhomme\thanks{%
University of Chicago.} \and Martin Weidner\thanks{University of Oxford, and Institute for Fiscal Studies, London.}}

\date{$\quad $\\ October 2021}
\vskip 3cm\maketitle

\begin{abstract}
\noindent    
We propose a framework for estimation and inference when the model may be misspecified. We rely on a local asymptotic approach where the degree of misspecification is indexed by the sample size. We construct estimators whose mean squared error is minimax in a neighborhood of the reference model, based on one-step adjustments. In addition, we provide confidence intervals that contain the true parameter under local misspecification. As a tool to interpret the degree of misspecification, we map it to the local power of a specification test of the reference model. Our approach allows for systematic sensitivity analysis when the parameter of interest may be partially or irregularly identified. As illustrations, we study three applications: an empirical analysis of the impact of conditional cash transfers in Mexico where misspecification stems from the presence of stigma effects of the program, a cross-sectional binary choice model where the error distribution is misspecified, and a dynamic panel data binary choice model where the number of time periods is small and the distribution of individual effects is misspecified.   

\bigskip


\end{abstract}

\baselineskip21pt

\bigskip

\bigskip

\setcounter{page}{0}\thispagestyle{empty}

\newpage

\section{Introduction\label{Intro_sec}}
Although economic models are intended as plausible approximations to a complex economic reality, econometric inference often relies on the model being an exact description of the population environment. To account for the possibility that their models are misspecified, economists have developed a number of approaches such as specification tests, semi-parametric and nonparametric methods, and more recently bounds approaches. Implementing those approaches typically requires estimating a more general model than the original specification, possibly involving nonparametric and partially identified components.

In this paper, we consider a different approach, which consists in quantifying how model misspecification affects the parameter of interest, and in modifying the estimate in order to minimize the impact of misspecification. The goal of the analysis is twofold. First, we provide simple adjustments, which do not require re-estimating the model, and provide guarantees on performance when the model is misspecified. Second, we construct confidence intervals that account for model misspecification error in addition to sampling uncertainty. 

In our approach, we consider deviations from a {reference specification} of the model{, in a particular class}. The reference model is parametric and fully specified given covariates. It may, for example, correspond to the empirical specification of a structural economic model. We do not assume that the reference model is correctly specified, and allow for \emph{local} deviations from it within a larger class of models. Relative to other approaches, a local analysis presents important advantages in terms of tractability.

We construct {minimax} estimators which minimize worst-case mean squared error (MSE) in a given neighborhood of the reference model. The worst case is influenced by the directions of model misspecification which matter most for the parameter of interest. We focus in particular on two types of neighborhoods, for two leading classes of applications: Euclidean neighborhoods, in settings where the larger class of models containing the reference specification is parametric, and Kullback-Leibler neighborhoods, in semi-parametric mixture models where misspecification of functional forms is measured by the Kullback-Leibler divergence between density functions.

The framework we propose is inspired by Hansen and Sargent's (2001, 2008) work on robust decision making under uncertainty and ambiguity. As in their work, optimal decisions depend on the size of the neighborhood around the reference model. In this paper, we do not attempt to provide a data-driven choice for the neighborhood size. Instead, we take the size as given and derive formulas for optimal estimation in neighborhoods of a given size. We discuss how to interpret the magnitude of the neighborhood size in various parametric and semi-parametric examples. In addition, we show that the neighborhood size can be mapped to the local power --- in certain directions --- of a likelihood-ratio test of correct specification of the reference model. 

Our approach delivers a class of estimators that can be used for systematic sensitivity analysis. In addition, we show how to construct confidence intervals which asymptotically contain the population parameter of interest with pre-specified probability, both under correct specification and local misspecification. We show that acknowledging misspecification leads to easy-to-compute enlargements of conventional confidence intervals. Such confidence intervals are ``honest'', in the sense that they account for the bias of the estimator (e.g., Donoho, 1994, Armstrong and Koles\'ar, 2020).

Our local approach leads to tractable expressions for worst-case bias and MSE, as well as for   minimum-MSE estimators in a given neighborhood of the reference model. A minimum-MSE estimator takes the form of a one-step adjustment of the estimator based on the reference model by a term which reflects the impact of model misspecification, in addition to a more standard term which adjusts the estimate in the direction of the efficient estimator based on the reference model. Implementing the optimal estimator only requires computing the score and Hessian of a larger model, evaluated at the reference model. The large model never needs to be estimated. This feature of our approach is reminiscent of the logic of Lagrange Multiplier (LM) testing.

We illustrate our approach using three examples. We first study the evaluation of the PROGRESA program in Mexico, which provides income transfers to households subject to the condition that the child attends school. Todd and Wolpin (2006) estimate a structural model of education choice on villages that were initially randomized out. They compare the predictions of the structural model with the estimated experimental impact. As emphasized by Todd and Wolpin (2008) and Attanasio \textit{et al.} (2012), the ability to predict the effects of the program based solely on control villages imposes restrictions on the economic model. Within a simple static model of education choice, we assess the sensitivity of counterfactual predictions to a form of misspecification under which program participation may have a direct ``stigma'' effect on the marginal utility of schooling (Wolpin, 2013).

We next study the impact of misspecification of the error distribution in a cross-sectional binary choice model. Our aim is to estimate the outcome probabilities under different values of the covariates. While point-identification can be achieved under independence and sufficiently rich support of covariates (Manski, 1988), the quantities of interest are partially identified in our setting. Relying on a normal (probit) reference model, we show how our estimators and confidence intervals can be used for sensitivity analysis, when the researcher is concerned about misspecification of the normal distribution. 

Our third and last example is a dynamic binary choice model in short panel data. We assume that time-varying errors are \textit{i.i.d.} normal, but leave the distribution of individual heterogeneity given initial conditions unrestricted. In this setting also, common parameters and average effects often fail to be point-identified (Chamberlain, 2010, Honor\'e and Tamer, 2006, Chernozhukov \textit{et al.}, 2013), thus motivating a sensitivity analysis approach. We show that minimizing worst-case MSE in such panel data settings leads to a Tikhonov-regularized estimator, where the penalization reflects the degree of misspecification allowed for. In simulations, we illustrate that our estimator can provide substantial bias and MSE reduction relative to commonly used estimators.

\paragraph{Related work and outline.}

As in the literature on robust statistics (Huber, 1964, Huber and Ronchetti, 2009, Hampel \textit{et al.}, 1986, and especially Rieder, 1994), we rely on a minimax approach and aim to minimize the worst-case impact of misspecification in a neighborhood of a model. A difference with this work is that we focus on misspecification of {specific aspects} of a model, by considering {parametric or semi-parametric} classes of models around the reference specification. By contrast, the robust statistics literature has mostly focused on fully {nonparametric} classes, motivated by data contamination issues. 

A related literature studies orthogonalization and locally robust moment functions; see Neyman (1959), Newey (1994), Chernozhukov \textit{et al.} (2018), Chernozhukov \textit{et al.} (2020), and also Fraser (1964). Here we account for both bias and variance, weighting them by the size of the neighborhood around the reference model. In addition, our approach does not require the larger model to be point-identified. Our analysis also connects to Bayesian robustness (e.g., Berger and Berliner, 1986, Gustafson, 2000, Vidakovic, 2000, Mueller, 2012), although our minimum-MSE estimators and confidence intervals have a frequentist interpretation. 

Also related are the literatures on statistical decision theory (e.g., Wald, 1950, Chamberlain 2000, Watson and Holmes, 2016, Hansen and Marinacci, 2016, and especially Hansen and Sargent, 2008) and the literature on sensitivity analysis in statistics and economics (e.g., Rosenbaum and Rubin, 1983, Leamer, 1985, Imbens, 2003, Altonji \textit{et al.}, 2005, Nevo and Rosen, 2012, Oster, 2019, Masten and Poirier, 2020, 2021). Our analysis of minimum-MSE estimation and sensitivity in the OLS/IV example is related to Hahn and Hausman (2005) and Angrist \textit{et al.} (2017). Our approach based on local misspecification has a number of precedents, such as Newey (1985), Conley \textit{et al.} (2012), Guggenberger (2012), Bugni \textit{et al.} (2012), Kitamura \textit{et al.} (2013), and Bugni and Ura (2019). Also related is Claeskens and Hjort's (2003) work on the focused information criterion. 

Recent papers rely on a local approach to misspecification to provide tools for sensitivity analysis. Andrews \textit{et al.} (2017) propose a measure of sensitivity of parameter estimates to the moments used in estimation. Andrews \textit{et al.} (2020) introduce a measure of informativeness of descriptive statistics in the estimation of structural models; see also Mukhin (2018). Our goal is different, in that we aim to provide a framework for estimation and inference in the presence of misspecification. Armstrong and Koles\'ar (2021) study models defined by over-identified systems of moment conditions that are approximately satisfied at true values, up to an additive term that vanishes asymptotically, and derive results for optimal estimation and inference. In this paper, we seek to ensure robustness to misspecification of a reference model within a larger class of models. 

Our focus on {specific} forms of model misspecification also relates to recent approaches to estimate partially identified models (Chen \textit{et al.}, 2011, Norets and Tang, 2014, Schennach, 2013, Giacomini and Kitagawa, 2021). Christensen and Connault (2019) consider structural models defined by equilibrium conditions, and develop inference methods on the identified set of counterfactual predictions subject to restrictions on the distance between the true model and a reference specification. Our local approach is complementary to these methods. It allows tractability in complex models, such as structural economic models, since implementation does not require estimating a larger model. In our framework, we view the parametric reference model as a useful benchmark, although its predictions need to be modified in order to minimize the impact of misspecification. This aspect relates our paper to shrinkage methods (e.g., Hansen, 2016, 2017, Fessler and Kasy, 2019, Maasoumi, 1978), with the difference that here we are interested in a single parameter.

The plan of the paper is as follows. In Section \ref{Sec_outline}, we describe our framework and derive the main results. In Section \ref{Sec_param}, we apply our framework to parametric and semi-parametric mixture models. In Section \ref{Sec_epsilon}, we discuss how to use our approach for sensitivity analysis, with a focus on the interpretation of neighborhood size. In Sections \ref{App_TW} and \ref{Sec_numeric}, we present our illustrations. Finally, we conclude in Section \ref{Sec_conclu}. The supplementary material available \href{https://sites.google.com/site/stephanebonhommeresearch/}{online} contains an appendix and codes for replication.

\section{Framework of analysis\label{Sec_outline}}

In this section, we describe the main elements of our approach in a general setting. In the next section, we will specialize the analysis to a locally quadratic setting, which includes both parametric misspecification and semi-parametric misspecification of distributional functional forms.   

\subsection{Setup}

We observe a random sample $( Y_i \, : \, i=1,\ldots,n)$ from a density $f_{\beta,\pi}(y)$ (with respect to a continuous or discrete measure), where $\beta \in {\cal B}$ is a finite-dimensional parameter,
and $\pi \in \Pi$ is a finite- or infinite-dimensional parameter. 
 Throughout the paper, the parameter of interest is $\delta_{\beta,\pi} $, a scalar function or functional of $\beta$ and $\pi$. We assume that $\delta_{\beta,\pi}$ and $f_{\beta,\pi}$ are known, smooth functions of $\beta$ and $\pi$. Examples of functionals of interest in economic applications include counterfactual policy effects in structural models, and average effects in panel data settings. The true parameter values $\beta_0$ and $\pi_0$ that generate the observed data $Y_1,\ldots,Y_n$ are unknown to the researcher.
 Our goal is to estimate $\delta_{\beta_0,\pi_0}$ and construct confidence intervals for it. We abstract from covariates to simplify the presentation, but it is straightforward to extend our results to  conditional models of the form $f_{\beta_0,\pi_0}(y\,|\, x)$; see Subsection \ref{sec:Covariates}.
 
Our starting point is that the researcher has chosen a reference model $\pi(\gamma)$, which parameterizes the unknown $\pi \in \Pi$
in terms of a finite-dimensional parameter $\gamma \in {\mathcal G} $. We say that the reference model is {correctly specified} if there exists a  value $\gamma \in {\mathcal G}$ such that  $\pi_0 = \pi(\gamma)$. Otherwise, we say that the model is {misspecified}.
{To measure misspecification, we rely on a distance measure $d$ on $\Pi$, and we denote the maximal amount of misspecification as $\epsilon\geq 0$.} 

In our theory, we consider an asymptotic sequence where $\epsilon = \epsilon_n $ tends to zero as $n $ tends to infinity, so the maximal amount of misspecification gets smaller as the sample size increases. The reason for focusing on $\epsilon$ tending to zero is tractability, as a small-$\epsilon$ analysis allows us to rely on linearization techniques and obtain simple, explicit expressions. Moreover, when estimating $\delta_{\beta_0,\pi_0}$, the estimation bias due to misspecification (of order $\epsilon^{1/2}$) and the
standard deviation (of order $n^{-1/2}$) are asymptotically comparable, so both play a role in the mean squared error. This local asymptotic approach has a number of precedents in the literature, notably Rieder (1994). Along the sequence, the true parameter $\pi_0 = \pi_{0,n}$ depends on $n$,
and we assume that, for a fixed parameter $\gamma_* $, $d( \pi_{0,n} , \pi(\gamma_*) )\leq \epsilon_n$ for all $n$. This implies that $\lim_{n \rightarrow \infty} d( \pi_{0,n}, \pi(\gamma_*) ) = 0$; that is, $\pi_{0,n}$ converges to $\pi(\gamma_*)$ as $n $ tends to infinity. Hereafter we drop the indices $n$ and do not make the sample size dependence of $\epsilon$ and $\pi_0$ explicit. For example,
we simply write $d( \pi_0 , \pi(\gamma_*) )\leq \epsilon$. 

Given the distance measure $d$, and some $\epsilon>0$, we define an $\epsilon$-{neighborhood} around $\pi(\gamma_*)$ as
\begin{align*}
\Gamma_{\epsilon}(\gamma_*) &= \left\{\pi_0 \in \Pi \,:\,d( \pi_0, \pi(\gamma_*) )\leq \epsilon\right\}.
\end{align*}
We assume that the true $\pi_0$ that generates the data satisfies $\pi_0 \in \Gamma_{\epsilon}(\gamma_*) $. Later we will assume that $\gamma_*$ can be estimated
consistently by some preliminary estimator $\widehat \gamma  $. The distance measure $d$, the misspecification bound $\epsilon$, and the preliminary estimator $\widehat \gamma$ are chosen by the researcher.\footnote{Instead of defining the $\epsilon$-neighborhood of misspecified models around a fixed point $\pi(\gamma_*)$, one could
alternatively consider all $\pi_0$ in the set $\Gamma_{\epsilon} = \cup_{\gamma \in {\cal G}} \, \Gamma_{\epsilon}(\gamma)$, which is the
$\epsilon$-neighborhood around the manifold of reference models $\pi(\gamma)$,  $\gamma \in {\cal G}$. This  alternative definition would avoid having to define $\gamma_*$, and
we employed it in the first version of this paper to justify the same 
local approximation to the worst-case MSE optimal estimator derived below
 (see Subsection 2.3 in Bonhomme and Weidner, 2018). In the current presentation, we only introduce the $\epsilon$-neighborhood $\Gamma_{\epsilon}(\gamma_*)$ around a fixed $\gamma_*$, however this has no effect on the minimax misspecification adjustments  that we derive. By fixing $\gamma_*$, this presentation also aligns with the way in which local minimax results are typically discussed in statistics (see, e.g., Theorem 8.11 in Van der Vaart, 2007).}

\paragraph{Examples.}
As a first example, consider a parametric model defined by Euclidean parameters $\beta$ and $\pi$, where $\pi=0$ under the reference model. For example, $\pi$ can represent the effect of an omitted control variable in a regression, or the degree of endogeneity of a regressor as in the example we analyze in Subsection \ref{subsec_par}. Suppose that the researcher is interested in the parameter $\delta_{\beta_0,\pi_0}=c'\beta_0$ for a known vector $c$, such as one component of $\beta_0$. In this case, we will take the weighted Euclidean (squared) distance $d(\pi_0,\pi)=\|\pi_0-\pi\|_{\Omega}^2=(\pi_0-\pi)'\Omega (\pi_0-\pi)$, for a positive-definite matrix $\Omega$.

As a second example, consider a semi-parametric mixture model whose likelihood depends on a finite-dimensional parameter vector $\beta$ and a nonparametric density $\pi$ of unobservables $A\in{\cal{A}}$, abstracting from conditioning covariates for simplicity. The joint density of $(Y,A)$ is $g_{\beta_0}(y\,|\, a)\pi_0(a)$, for some known function $g$. Suppose that the researcher's goal is to estimate an average effect $\delta_{\beta_0,\pi_0}=\mathbb{E}_{\pi_0}\Delta(A,\beta_0)$, for a known function $\Delta$. It is common to estimate the model by parameterizing the unknown density as $\pi(\gamma)$, where $\gamma$ is finite-dimensional. We focus on situations where, although the researcher thinks of $\pi(\gamma)$ as a plausible approximation to the population distribution $\pi_0$, she is not willing to rule out that it may be misspecified. In this case we use the Kullback-Leibler divergence to define semi-parametric neighborhoods, and we take $d(\pi_0,\pi)=2\int_{\cal{A}} \log \left(\frac{\pi_0(a)}{\pi(a)}\right)\pi_0(a)da$. 
\hfill $\square$

\vskip .3cm

We focus on asymptotically linear estimators $\widehat{\delta}=\widehat{\delta}(Y_1,\ldots,Y_n)$ that admit a stochastic expansion of the form
\begin{equation}
\widehat{\delta}=\delta_{\beta_0,\pi(\gamma_*)}+\frac{1}{n}\sum_{i=1}^n h(Y_i,\beta_0,\gamma_*)+ o_{P_0}(n^{-\frac{1}{2}} + \epsilon^{1/2}),
\label{est_delta_hat}
\end{equation}
where this expansion holds uniformly for all $P_0=P_{\beta_0,\pi_0}$ such that $\pi_0\in\Gamma_{\epsilon}(\gamma_*)$, in a sense that we will discuss below and make precise in Theorem \ref{theo1}. Along the sequence we consider, the product $\epsilon n$ tends to a positive constant, so the remainder in (\ref{est_delta_hat}) is $o_{P_0}(n^{-\frac{1}{2}} )$. Although asymptotic linearity is satisfied by many econometric estimators, it can fail in certain semi-parametric problems (e.g., Cattaneo \textit{et al.}, 2014) and in problems involving model selection or shrinkage (e.g., Liao, 2013, Cheng and Liao, 2015), for example.

Equation (\ref{est_delta_hat}) is a form of local regularity of the estimator $\widehat{\delta}$. Consider first the correctly specified case, where $\epsilon=0$ and $\pi_0 = \pi(\gamma_*)$. In this case $h(\cdot,\beta_0,\gamma_*)$ is the influence function of $\widehat{\delta}$. We assume that the following conditions are satisfied,
\begin{align}
\mathbb{E}_{\beta_0,\pi(\gamma_*)} \, h(Y,\beta_0,\gamma_*) = 0 , 
\label{Con:Unbiased}
\end{align} 
and
\begin{align} 
\nabla_{\beta\gamma}\delta_{\beta_0,\pi(\gamma_*)}
+\mathbb{E}_{\beta_0,\pi(\gamma_*)} \, \nabla_{\beta\gamma}h(Y,\beta_0,\gamma_*) &=0,
\label{CharacterizeInfluenceH}   
\end{align}
where $\mathbb{E}_{\beta,\pi}$ denotes the expectation under $f_{\beta,\pi}$, and $\nabla_{\beta\gamma}$ denotes the derivative with respect to the vector $(\beta',\gamma')'$. Both \eqref{Con:Unbiased} and \eqref{CharacterizeInfluenceH} are standard properties of influence functions of regular asymptotically linear estimators.

We will refer to \eqref{Con:Unbiased} as \emph{unbiasedness}, since it guarantees that $ \widehat{\delta}$
is asymptotically unbiased for $\delta_{\beta_0,\pi_0}$ under correct specification of the reference model. We assume that unbiasedness holds at all possible values of $\beta_0$ and $\gamma_*$.
Then, by differentiating \eqref{Con:Unbiased} with respect to $\beta_0$ and $\gamma_*$
and plugging the resulting equations into \eqref{CharacterizeInfluenceH}, we obtain
\begin{align} 
&\mathbb{E}_{\beta_0,\pi(\gamma_*)}  \,  h(Y,  \beta_0,\gamma_*)    \, \nabla_{\beta\gamma} \log f_{\beta_0,\pi(\gamma_*)}(Y)
=  \nabla_{\beta\gamma}\delta_{\beta_0,\pi(\gamma_*)}.
\label{Con:EtaGradient}
\end{align}
Under unbiasedness, \eqref{CharacterizeInfluenceH} and \eqref{Con:EtaGradient} are  equivalent. We will later work with \eqref{Con:EtaGradient}, since it only features $h(Y,  \beta_0,\gamma_*)$ and not its gradient. Under suitable conditions, \eqref{Con:EtaGradient} is necessary and sufficient for the asymptotically linear estimator $\widehat{\delta}$ to be regular; see, e.g., Newey (1990). As an example, for m-estimators, \eqref{Con:EtaGradient} can be interpreted as the generalized information matrix equality. Asymptotic linearity and regularity are commonly imposed in the  semi-parametric efficiency literature (Bickel \textit{et al.}, 1993). These conditions rule out, for example, superefficient estimators such as Hodges' estimator. We will refer to \eqref{CharacterizeInfluenceH}, or alternatively \eqref{Con:EtaGradient}, as \emph{local robustness},
using a terminology introduced by Chernozhukov \textit{et al.} (2020).\footnote{%
	While in Chernozhukov \textit{et al.} (2020) local robustness is imposed as a substantive restriction
	on more general moment functions,  in our setting (\ref{CharacterizeInfluenceH}) and (\ref{Con:EtaGradient}) are regularity conditions given unbiasedness.}

Consider now the misspecified case, where $\epsilon>0$. In this case, we strengthen the condition of asymptotic linearity, and require that $\widehat{\delta}$ be locally asymptotically linear; see, e.g., Klaassen (1987). Formally, under local, small-$\epsilon$ misspecification, we assume the stochastic expansion \eqref{est_delta_hat} continues to hold, but now uniformly for all $\pi_0\in\Gamma_{\epsilon}(\gamma_*)$.\footnote{To see why \eqref{est_delta_hat} is a plausible way of imposing asymptotic linearity here, let $\phi(Y_i,\beta_0,\pi_0)$ be the influence function of $ \widehat{\delta}$. Expanding as $n\rightarrow\infty$ and $\epsilon\rightarrow 0$ we have
	\begin{align*}
	\widehat{\delta} &= \delta_{\beta_0,\pi(\gamma_*)} 
	+\frac{1}{n}\sum_{i=1}^n \underbrace{\phi(Y_i,\beta_0,\pi(\gamma_*))}_{=h(Y_i,\beta_0,\gamma_*)}
	+  [\pi_0 - \pi(\gamma_*)]' \, 
	\underbrace{  \left[\nabla_{\pi}\delta_{\beta_0,\pi(\gamma_*)}
		+\mathbb{E}_{\beta_0,\pi(\gamma_*)}\nabla_{\pi}\phi(Y,\beta_0,\pi(\gamma_*)) \right] }_{=0}
		\\ & \qquad \qquad \qquad \qquad \qquad \qquad \qquad \qquad \qquad \qquad \qquad \qquad \qquad \qquad \qquad \qquad \qquad \qquad \qquad
	+ o_{P_0}(n^{-\frac{1}{2}} + \epsilon^{\frac{1}{2}}) .
	\end{align*}
	In this expansion the term linear in $\pi_0 - \pi(\gamma_*)$ vanishes, whenever
	$\phi(Y,\beta_0,\pi_0)$  satisfies an influence function regularity condition analogous to \eqref{CharacterizeInfluenceH},
	and the term quadratic in $\pi_0 - \pi(\gamma_*)$  gives a contribution $ o_{P_0}(\epsilon^{\frac{1}{2}}) $.
} In the following, we focus on locally asymptotically linear estimators that satisfy \eqref{est_delta_hat},
under the conditions \eqref{Con:Unbiased} and \eqref{Con:EtaGradient}. Notice that, under local misspecification, the influence function $h(Y,\beta_0,\gamma_*)$ has no longer mean zero under $P_0$ in general.

Our goal in this paper is twofold. First, we will construct confidence intervals for the target parameter $\delta_{\beta_0,\pi_0}$ which are uniformly asymptotically valid on $\Gamma_{\epsilon}(\gamma_*)$. Second, an important goal of the analysis is to construct estimators $\widehat{\delta}=\widehat{\delta}(Y_1,\ldots,Y_n)$ that are asymptotically optimal in a minimax sense. For this purpose, we will show how to compute a function $h$ such that the (trimmed) worst-case mean squared error (MSE) $\mathbb{E}_{\beta_0,\pi_0} [ ( \widehat \delta  - \delta_{\beta_0,\pi_0} )^2 ]$ in the $\epsilon$-neighborhood $\Gamma_{\epsilon}(\gamma_*)$ of the reference model, among estimators of the form
\begin{equation}
\label{eq_delta_hat}
\widehat{\delta}_{h,\widehat{\beta},\widehat{\gamma}}=\delta_{\widehat{\beta},\pi(\widehat{\gamma})}+\frac{1}{n} \sum_{i=1}^n h(Y_i, \widehat \beta,\widehat \gamma),
\end{equation}
is minimized under our local asymptotic analysis. In fact, we will show how to compute estimators that minimize (trimmed) worst-case MSE among asymptotically linear estimators; see Theorem \ref{theo1} below for a precise statement. Here $\widehat{\beta}$  and $\widehat{\gamma}$ are
preliminary estimators of $\beta_0$ and $\gamma_*$ that are {root-$n$} consistent under correct specification. For example, $\widehat{\beta}$  and $\widehat{\gamma}$ may be maximum likelihood estimators (MLE) based on the reference model. It follows from \eqref{Con:Unbiased} and \eqref{Con:EtaGradient} that, under regularity conditions on the preliminary estimators, the form of the minimum-MSE $h$ function is not affected by the choice of 
$\widehat{\beta}$ and $\widehat{\gamma}$.

\paragraph{Examples (cont.)}

In our parametric example, a natural estimator is the MLE of $c'\beta_0$ based on the reference specification, such as the OLS estimator under the assumption that $\pi=0$; e.g., that the coefficient of an omitted control variable is zero. In a correctly specified likelihood setting, such an estimator will be consistent and efficient. However, when the reference model is misspecified, it may be dominated in terms of bias or MSE by other regular estimators.  

In our semi-parametric mixture example, a commonly used (``random-effects'') estimator of $\delta_{\beta_0,\pi_0}=\mathbb{E}_{\pi_0}\Delta(A,\beta_0)$ is obtained by replacing the population average by an integral with respect to the parametric distribution $\pi(\widehat{\gamma})$, where $\widehat{\gamma}$ is the MLE of $\gamma$. Another popular (``empirical Bayes'') estimator is obtained by substituting an integral with respect to the posterior distribution of $A$ based on $\pi(\widehat{\gamma})$.  We will compare the finite-sample performance of these estimators to that of our minimum-MSE estimator in our panel data illustration in Section \ref{Sec_numeric}. 
\hfill $\square$

\subsection{Heuristic derivation of the minimum-MSE estimator}

In this subsection, we provide heuristic derivations for the worst-case bias and the minimum-MSE estimator. This will lead to the main expressions in equations (\ref{eqbepsilon}) and (\ref{MMSEproblem}) below. In the next subsection, we will provide regularity conditions under which these derivations are formally justified.

We assume that $\Gamma_\epsilon(\gamma_*)$ is a convex set. For any linear map $u:\Pi \rightarrow \mathbb{R}$, we define\footnote{The definition in (\ref{DefDualNorm}) is stated for finite-dimensional $\pi$. This definition can be generalized to the case where $\pi$ is infinite-dimensional; see Appendix \ref{App_general} and Section \ref{Sec_param}. }
\begin{align}
    \left\|  u  \right\|_{\gamma_*,\epsilon}
    =
     \sup_{ \pi_0 \in  \Gamma_\epsilon(\gamma_*)} \, \epsilon^{-\frac{1}{2}} \,  u'(\pi_0 - \pi(\gamma_*)),\quad   \left\|  u  \right\|_{\gamma_*}=\limfunc{lim}_{\epsilon\rightarrow 0} \,  \left\|  u  \right\|_{\gamma_*,\epsilon} .
   \label{DefDualNorm}
\end{align}
We assume that the distance measure $d$ is chosen such that $ \left\|  \cdot  \right\|_{\gamma_*}$ is unique and well-defined,
and that it is a norm, dual to a local  approximation to $d( \pi_0, \pi(\gamma_*))$ for fixed $\pi(\gamma_*)$. Both our examples of distance measures -- weighted Euclidean distance and Kullback-Leibler divergence -- satisfy these assumptions.

We focus on estimators $\widehat{\delta}$ that satisfy (\ref{est_delta_hat}) for a suitable $h$ function for which \eqref{Con:Unbiased} and \eqref{Con:EtaGradient} hold. Under appropriate regularity conditions, the worst-case bias of $\widehat{\delta}$ in the neighborhood $\Gamma_\epsilon(\gamma_*)$ can be expanded for small $\epsilon$ and large $n$ as
\begin{align} 
&\sup_{\pi_0 \in \Gamma_\epsilon(\gamma_*)}
\left|   \mathbb{E}_{\beta_0,\pi_0} \widehat{\delta}  -\delta_{\beta_0,\pi_0}\right|=b_{\epsilon}(h,\beta_0,\gamma_*)+ o(n^{-\frac{1}{2}}+\epsilon^{\frac{1}{2}} ),\label{eqbias}\end{align} 
where
\begin{align} 
&b_{\epsilon}(h,\beta_0,\gamma_*)=\epsilon^{\frac{1}{2}} \, \left\|  \nabla_\pi \delta_{\beta_0,\pi(\gamma_*)} -  \mathbb{E}_{\beta_0,\pi(\gamma_*)} \, h(Y,\beta_0,\gamma_*) \; \nabla_\pi \log f_{\beta_0,\pi(\gamma_*)}(Y) \right\|_{\gamma_*} ,\label{eqbepsilon}\end{align} 
for $\|\cdot\|_{\gamma_*}$ the dual norm defined in (\ref{DefDualNorm}).\footnote{When $\pi$ is infinite-dimensional $\nabla_{\pi}$ denotes a G\^ateaux derivative.} 

Then, the worst-case MSE in $\Gamma_\epsilon(\gamma_*)$ can be expanded as follows, again under appropriate regularity conditions
(see Lemma~\ref{lemma:MSEapprox} in the appendix),
\begin{align}
\sup_{\pi_0\in \Gamma_\epsilon(\gamma_*)}   \mathbb{E}_{\beta_0,\pi_0} \left[ \left( \widehat \delta - \delta_{\beta_0,\pi_0} \right)^2 \right]
&= 
b_{\epsilon}(h,\beta_0,\gamma_*)^2
+  \frac{{\rm Var}_{\beta_0,\pi(\gamma_*)}(h(Y,\beta_0,\gamma_*))  } {n}
+ o(n^{-1}+\epsilon).\label{eqMSE}
\end{align}
We therefore define the {minimum-MSE} function $h_{\epsilon}^{\rm MMSE}(y,\beta_0,\gamma_*)$ as
\begin{align}
h_{\epsilon}^{\rm MMSE}( \cdot ,\beta_0,\gamma_*) & = \argmin_{h( \cdot ,\beta_0,\gamma_*)} \,
\Bigg\{  \epsilon \, \left\|   \nabla_\pi \delta_{\beta_0,\pi(\gamma_*)} -  \mathbb{E}_{\beta_0,\pi(\gamma_*)} \, h(Y,\beta_0,\gamma_*) \; \nabla_\pi \log f_{\beta_0,\pi(\gamma_*)}(Y) \right\|^2_{\gamma_*}  
\nonumber \\ & \quad \quad \quad \qquad 
+  \frac{{\rm Var}_{\beta_0,\pi(\gamma_*)}(h(Y,\beta_0,\gamma_*))  } {n}   \Bigg\}
  \quad 
\text{subject to \eqref{Con:Unbiased} and \eqref{Con:EtaGradient}}.
\label{MMSEproblem}
\end{align}

Finally, let $\widehat \beta$ and $\widehat \gamma$ be preliminary estimators
that are consistent for $\beta_0$ and $\gamma_*$ under the reference model $f_{\beta_0,\pi(\gamma_*)}$.
Then, the {minimum-MSE} estimator of $\delta_{\beta_0,\pi_0}$ is given by
\begin{align}
\widehat \delta\,^{\rm MMSE}_\epsilon = \delta_{\widehat \beta,\pi(\widehat \gamma)}+\frac 1 n \sum_{i=1}^n h_\epsilon^{\rm MMSE}(Y_i, \widehat \beta,\widehat \gamma).
\label{est_delta_hat_eta_hat}
\end{align}   
This estimator minimizes an asymptotic approximation to the worst-case MSE in $\Gamma_\epsilon(\gamma_*)$. Using a small-$\epsilon$ approximation is crucial for analytic tractability, since the variance term in (\ref{eqMSE}) only needs to be calculated under the reference model, and the optimization problem \eqref{MMSEproblem} is convex.
In practice, (\ref{MMSEproblem}) only needs to be solved at $ \widehat \beta$ 
and $ \widehat{\gamma}$. In addition, as we already pointed out, the form of the minimum-MSE estimator is not affected by the choice of the preliminary estimators $\widehat \beta$ and $\widehat{\gamma}$. 

The constraints on $h_{\epsilon}^{\rm MMSE}( \cdot ,\beta_0,\gamma_*) $ imposed in
\eqref{MMSEproblem} are the unbiasedness condition \eqref{Con:Unbiased}
and the local robustness condition (\ref{Con:EtaGradient}).
As we discussed above, given unbiasedness, local robustness is a regularity condition, and unbiasedness is a substantive condition that implies that our estimator 
$ \widehat \delta\,^{\rm MMSE}_\epsilon $ is only optimal within the class of estimators
that are asymptotically unbiased for $\delta_{\beta_0,\pi_0}$ under the reference model.

\paragraph{Special cases.}

To provide intuition about the minimum-MSE function $h^{\rm MMSE}_{\epsilon}$, let us define two Hessian matrices $H_{\beta\gamma}$, of size $\limfunc{dim}\beta+\limfunc{dim}\gamma$, and ${\cal{H}}_{\beta\pi}$, of size $\limfunc{dim}\beta+\limfunc{dim}\pi$, as\footnote{The definition of ${\cal{H}}_{\beta\pi}$ generalizes to the infinite-dimensional $\pi$ case; see Appendix \ref{App_general} and Section \ref{Sec_param}.
	}
\begin{align*}
 H_{\beta\gamma}&=   \mathbb{E}_{\beta_0,\pi(\gamma_*)}\left[ \nabla_{\beta\gamma} \log f_{\beta_0,\pi(\gamma_*)}(Y) \right]   \left[ \nabla_{\beta\gamma} \log f_{\beta_0,\pi(\gamma_*)}(Y) \right]',\\{\cal{H}}_{\beta\pi}&
= \mathbb{E}_{\beta_0,\pi(\gamma_*)}\left[ \nabla_{\beta\pi} \log f_{\beta_0,\pi(\gamma_*)}(Y) \right]   \left[ \nabla_{\beta\pi} \log f_{\beta_0,\pi(\gamma_*)}(Y) \right]'.
\end{align*}

Throughout our analysis, we assume that $H_{\beta\gamma}$ is invertible. This requires that the Hessian matrix of the parametric reference model be non-singular, thus requiring that $\beta_0$ and $\gamma_*$ be identified under the reference model. When $\epsilon=0$ we find that
\begin{align}
h_0^{\rm MMSE}( y ,\beta_0,\gamma_*) &=        \nabla_{\beta\gamma} \log f_{\beta_0,\pi(\gamma_*)}(y) '     H_{\beta\gamma}^{-1}     
\, \nabla_{\beta\gamma}  \delta_{\beta_0,\pi(\gamma_*)}  .
\label{SolutionEpsilonZero}  
\end{align}
Thus, under the assumption that the parametric reference model is correctly specified, $\widehat \delta^{\rm MMSE}_\epsilon$ is simply the one-step approximation to the MLE for $\delta_{\beta_0,\pi_0}$ that maximizes the likelihood with respect to the ``small'' parameter $(\beta',\gamma')'$. This ``one-step efficient'' adjustment is purely based on efficiency considerations. Such one-step approximations are classical estimators in statistics (e.g., Bickel \textit{et al.}, 1993).   

Another interesting special case of the minimum-MSE $h$ function arises in the limit $\epsilon \rightarrow \infty$, when the matrix or operator ${\cal{H}}_{\beta\pi}$ is invertible. Note that invertibility of ${\cal{H}}_{\beta\pi}$, which may fail when $\pi_0$ is not identified, is not needed in our analysis and we only use it to analyze this limiting case. We then have that
\begin{align}
\lim_{\epsilon \rightarrow \infty} h_\epsilon^{\rm MMSE}( y ,\beta_0,\gamma_*)
=      \left[ \nabla_{\beta\pi}\log f_{\beta_0,\pi(\gamma_*)}(y) \right]'     
\, {\cal{H}}_{\beta\pi}^{-1}          \, \nabla_{\beta\pi}  \delta_{\beta_0,\pi(\gamma_*)}   .
\label{OneStepFullMLE}   
\end{align}
In this limit $\widehat \delta\,^{\rm MMSE}_\epsilon$
is simply the one-step approximation to the MLE
for $\delta_{\beta_0,\pi_0}$ that maximizes the likelihood with respect to the ``large'' parameter $(\beta',\pi')'$. For any $\epsilon$, the estimator $ \widehat \delta\,^{\rm MMSE}_\epsilon $ is a nonlinear interpolation between the one-step MLE approximation to the parametric reference model and the one-step MLE approximation to the large model. We obtain one-step approximations in our approach, since \eqref{MMSEproblem} is only
a {local} approximation to the full MSE-minimization problem. 

However, an estimator based on (\ref{OneStepFullMLE}) may be ill-behaved in non point-identified problems, or in problems where the identification of $\pi_0$ is irregular. By contrast, $\widehat \delta\,^{\rm MMSE}_\epsilon$ is always
well-defined, since the variance of $h(Y,\beta_0,\gamma_*)$ acts as a sample size-dependent regularization. The form of $\widehat \delta\,^{\rm MMSE}_\epsilon$ is thus based on {both} efficiency and robustness. In addition, note that, while neither (\ref{SolutionEpsilonZero}) nor (\ref{OneStepFullMLE}) involve the particular choice of distance measure with respect to which neighborhoods are defined, for given $\epsilon>0$ the minimum-MSE estimator will depend on the chosen distance measure.

Lastly, it is common in applications with covariates to model the conditional distribution of outcomes $Y$ given covariates $X$ as $f_{\beta_0,\pi_0}(y\,|\, x)$, while leaving the marginal distribution of $X$, $f_X(x)$, unspecified. Our approach can easily be adapted to deal with such {conditional} models, as we will describe in Subsection \ref{sec:Covariates} in a locally quadratic setting.

\subsection{Properties of the minimum-MSE estimator}

In this subsection, we provide a formal characterization of the minimum-MSE estimator by showing that it achieves minimum worst-case MSE in a class of regular asymptotically linear estimators, as $n$ tends to infinity and $\epsilon n$ tends to a constant. All sequences  can thus be
equivalently indexed by $\epsilon$ or $n$; for example, $h_{\epsilon}$ in the following theorem could equivalently be indexed by $n$. Moreover, under the stated assumptions, the heuristic derivations of the previous subsection are formally justified. All proofs are in the appendix.

	\begin{theorem}
		\label{theo1}
		Let $n \rightarrow \infty$ and $\epsilon \rightarrow 0$
		such that $\epsilon  n \rightarrow c$, for some constant $c \in (0,\infty)$. Let Assumptions~\ref{ass:Expansion} and~\ref{ass:Expansion2} in Appendix \ref{App_general} hold, 
		let $\beta_0 \in {\cal B}$ and $\gamma_* \in {\cal G}$,
		and
		let $\widehat \delta_{\epsilon}=\widehat \delta_{\epsilon}(Y_1,\ldots,Y_n)$ be a sequence of 
		estimators with an influence function expansion of the form
		\begin{align}
	            \widehat \delta_{\epsilon}  &=  \delta_{\beta_0,\pi(\gamma_*)} 
	            +    \frac 1 n \sum_{i=1}^n  h_{\epsilon}(Y_i, \beta_0, \gamma_*)  
                     +  n^{-1/2} \,   R_n ,      
                     \label{DefRemainder1}
		\end{align}
		where   $R_n$ is a sequence of random  variables
		with 
		\begin{align*}
		  \sup_{\pi_0\in \Gamma_\epsilon(\gamma_*)}
      {\rm P}_{\beta_0,\pi_0} \left(     \left|     R_n \right| >  \log(n)   \right) &= o(1),
	&
		 \sup_{\pi_0 \in \Gamma_{\epsilon}(\gamma_*)}  \mathbb{E}_{\beta_0,\pi_0}
		\left[    R_n^2 \, \mathbbm{1}\left(  |R_n| \leq  2 \, \log(n)   \right)  \right] &= o( 1 ),
	\end{align*}	
		and $h_{\epsilon}(\cdot,\beta_0, \gamma_*)$ is
    a sequence of influence functions  that satisfy the constraints \eqref{Con:Unbiased} and \eqref{Con:EtaGradient},
		as well as $\sup_{\pi_0 \in \Gamma_{\epsilon}(\gamma_*)} \allowbreak
		 \mathbb{E}_{\beta_0,\pi_0}  \left| h_{\epsilon}(Y, \beta_0, \gamma_*)  \right|^{\kappa} = O(1)$, for some 
		$\kappa > 2$.
		We then have, for any sequence $m_n > 0$
                 with $m_n \rightarrow 0$ and $m_n \, n^{1/2}   \,  [\log(n)]^{-1}    \rightarrow \infty$, that
		\begin{align}
		\sup_{\pi_0\in \Gamma_\epsilon(\gamma_*)} 
		& \mathbb{E}_{\beta_0,\pi_0}  \left[ 
		\left( \widehat \delta^{\, \rm MMSE}_{\epsilon}  - \delta_{\beta_0,\pi_0} \right)^2 
		\mathbbm{1}\left( \left| \widehat \delta^{\, \rm MMSE}_{\epsilon}  - \delta_{\beta_0,\pi_0} \right|
		 \leq  m_n
		\right)		
		 \right]
	\nonumber \\ & \quad 	\leq 
		\sup_{\pi_0\in \Gamma_\epsilon(\gamma_*)}  \mathbb{E}_{\beta_0,\pi_0}  \left[ \left( \widehat \delta_{\epsilon}  - \delta_{\beta_0,\pi_0} \right)^2 
		\mathbbm{1}\left( \left| \widehat \delta_{\epsilon}   - \delta_{\beta_0,\pi_0} \right| \leq 
		m_n \right)
		\right]   
		+ o\left(\frac{1}{n}\right)  . \label{eq_star}
		\end{align}
	\end{theorem}

\vskip .3cm

We establish Theorem \ref{theo1} in a joint asymptotic sequence where $\epsilon $ tends to zero as $n$ tends to infinity and $\epsilon n $ tends to a finite positive constant. Under this sequence, the leading term in the worst-case MSE is of order $\epsilon$ (squared bias), or equivalently of order $1/n$ (variance).
The theorem considers a
trimmed MSE to allow for the possibility that the estimators for $ \delta_{\beta_0,\pi_0}$ do not have moments.
The trimming cutoff $m_n$ shrinks to zero at a rate slower than $n^{-1/2}$ (or equivalently $\epsilon^{1/2}$), so that
for estimators without heavy tails the leading-order bias and standard deviation should not be affected by the trimming.	The theorem states that the leading-order worst-case trimmed MSE achieved by our minimum-MSE estimator $\widehat \delta^{\, \rm MMSE}_{\epsilon}$ is at least as good as the one achieved by any other sequence of estimators satisfying our regularity conditions. All the assumptions on $\widehat \delta_{\epsilon}$ and $h_{\epsilon}(\cdot, \beta,\gamma)$ that we require for this result are  listed in the statement of the theorem. Note that in Theorem \ref{theo1} we {assume} that (\ref{DefRemainder1}) holds, subject to \eqref{Con:Unbiased} and \eqref{Con:EtaGradient}. On might conjecture that the result holds absent these conditions. However, our current proof crucially relies on them.\footnote{Condition (\ref{DefRemainder1}) is a form of local regularity of the sequence of estimators $\widehat{\delta}_{\epsilon}$. The additional regularity conditions in Assumptions~\ref{ass:Expansion} and~\ref{ass:Expansion2} are smoothness conditions on $f_{\beta_0,\pi_0}(y)$, $\delta_{\beta_0,\pi_0}$, $\pi(\gamma)$, and $d(\pi_0,\pi(\gamma))$ as functions of $\beta_0$, $\pi_0$, and $\gamma$, 
	and an appropriate  rate condition on the preliminary estimators $\widehat \beta$ and $\widehat \gamma$. In particular, in Assumption~\ref{ass:Expansion2}, we require the preliminary estimators
 $\widehat \beta$ and $\widehat \gamma$ to have moments
of order larger than two. This may require modifying the preliminary estimators to ensure that they have finite moments, as in for example Hausman \textit{et al.} (2011), who focus on GMM estimators.} {We will provide explicit expressions for $h_{\epsilon}^{\rm MMSE}$ in various models in the next section.}

\subsection{Confidence intervals\label{subsec_CI}}

In addition to point estimates, our framework allows us to compute confidence intervals that contain $\delta_{\beta_0,\pi_0}$ with pre-specified probability under our local asymptotic analysis. To see this, let $\widehat{\delta}$ be an estimator satisfying (\ref{est_delta_hat}), (\ref{Con:Unbiased}), and (\ref{Con:EtaGradient}). For a given confidence level $\alpha\in (0,1)$, let us define the following interval
\begin{eqnarray}
CI_{\epsilon}(1-\alpha,\widehat{\delta})= \left[\widehat{\delta}\pm\left(b_{\epsilon}\left(h,\widehat{\beta},\widehat{\gamma}\right)+\frac{\widehat{\sigma}_h}{\sqrt{n}}c_{1-\alpha/2}\right)\right],\label{CI_def}
\end{eqnarray}
where $b_{\epsilon}\left(\cdot\right)$ is given by (\ref{eqbepsilon}), $\widehat{\sigma}_h^2$ is the sample variance of $h(Y_1,\widehat{\beta},\widehat{\gamma}),\ldots , h(Y_n,\widehat{\beta},\widehat{\gamma})$, and $c_{1-\alpha/2}=\Phi^{-1}(1-\alpha/2)$ is the $(1-\alpha/2)$-standard normal quantile. Under suitable regularity conditions, the interval $CI_{\epsilon}(1-\alpha,\widehat{\delta})$ contains $\delta_{\beta_0,\pi_0}$ with probability {at least} $1-\alpha$ as $n$ tends to infinity and $\epsilon n$ tends to a constant, both under correct specification and under local misspecification of the reference model. Formally, we have the following result.

\begin{theorem}\label{theo_CI}
	Let $n \rightarrow \infty$ and $\epsilon \rightarrow 0$
	such that $\epsilon  n \rightarrow c$, for some constant $c \in (0,\infty)$. Let Assumptions \ref{ass:Expansion} and \ref{ass_CI} in Appendix \ref{App_general} hold,
	and also assume that the influence function $h$ of $\widehat{\delta}$ satisfies $ \sup_{\pi_0\in \Gamma_{\epsilon}(\gamma_*)}     \mathbb{E}_{\beta_0,\pi_0}  h^2(Y,\beta_0,\gamma_*) = O(1)$. Then we have
	\begin{equation}{\limfunc{inf}}_{\pi_0\in\Gamma_\epsilon(\gamma_*) }\,{\limfunc{Pr}}_{\beta_0,\pi_0} \left[\delta_{\beta_0,\pi_0}\in CI_{\epsilon}(1-\alpha,\widehat{\delta})\right]\geq 1-\alpha+o(1).\label{res_theo_CI}\end{equation}
	
\end{theorem}

\vskip .3cm

Such ``fixed-length'' confidence intervals, which take into account both misspecification bias and sampling uncertainty, have been studied in different contexts (e.g., Donoho, 1994, Armstrong and Koles\'ar, 2020, 2021).\footnote{A variation suggested by these authors, which reduces the length of the interval, is to compute the interval as $\widehat{\delta}\pm$ $b_{\epsilon}(h,\widehat{\beta},\widehat{\gamma})$ times the $(1-\alpha)$-quantile of $\left|{\cal{N}}\left(1,\widehat{\sigma}_h^2/(nb_{\epsilon}(h,\widehat{\beta},\widehat{\gamma})^2)\right)\right|$.\label{Fig_CIftnoteAK}}

\section{Locally quadratic case: applications to parametric models and semi-parametric mixture models\label{Sec_param}}

In this section, we first derive explicit expressions for minimum-MSE estimators in a class of models that have a locally quadratic structure; that is, where the dual norm $\|\cdot\|_{\gamma_*}$ can be associated with an inner product. We then apply these results to parametric models and semi-parametric mixture models. Our motivation for focusing on these settings is that they allow us to develop practical implementation methods. We have not explored implementation in other models with infinite-dimensional $\pi$ parameters.

Formally, let us define the tangent space $\overline{{\cal{T}}}$ of the parameter space $\Pi$ at $\pi(\gamma^*)$, where for simplicity we ignore the dependence on $\gamma^*$ in the notation.\footnote{The tangent space at $\pi(\gamma^*)$ includes all directions in which one can pass tangentially through $\pi(\gamma^*)$. See Barden and Thomas (2003) for a formal exposition.} Let us then define the cotangent space ${\cal T}$ as the set of linear maps $u:\overline{{\cal{T}}}\rightarrow \mathbb{R}$. Throughout this section, we assume that  ${\cal T}$ is a Hilbert space equipped with the norm $\left\|\cdot\right\|_{\gamma_*}$. This locally quadratic structure characterizes our two leading examples of parametric and semi-parametric mixture models. In such cases the tangent space (at an interior $\pi(\gamma^*)$) is simply $\Pi$, and the cotangent space is the set of linear maps $u:\Pi\rightarrow \mathbb{R}$.

\subsection{Characterization of the minimum-MSE estimator\label{sec:LocallyQuadratic}}

Consider the case where the square of the local dual norm defined in \eqref{DefDualNorm} can be written as $\left\|  u  \right\|^2_{\gamma_*} = u^{\top}   u  $,
where $u^{\top}   w$ represents some  inner product of elements $u$ and $w$ of the cotangent space ${\cal T}$ 
of $\Pi$ at $\pi(\gamma_*)$. For conciseness, from now on we will remove the subscripts $\beta_0$, $\gamma_{*}$, and $\pi(\gamma_*)$ throughout, unless there is a risk of confusion. In particular, unless otherwise noted, all expectations will be evaluated under the reference model. Here, $\pi$ can be finite-dimensional as in parametric models (which we analyze in Subsection \ref{subsec_par}), or infinite-dimensional as in semi-parametric mixture models where $\pi$ is a density (studied in Subsection \ref{subsec_semipar}). 

Let us start by introducing some notation. Let $s_{\beta\gamma}(y)=\nabla_{\beta\gamma} \log f(y) $ and $s_{\pi}(y)=\nabla_{\pi} \log f(y) $ denote the components of the score. We define the Hessian  operators $H_{\pi} : {\cal T} \rightarrow {\cal T}$, $H_{\pi,\beta\gamma}:\mathbb{R}^{\limfunc{dim}\beta+\limfunc{dim}\gamma}\rightarrow{\cal{T}}$, and $H_{\beta\gamma,\pi}:{\cal{T}}\rightarrow\mathbb{R}^{\limfunc{dim}\beta+\limfunc{dim}\gamma}$ by\footnote{
	Formally, $u^{\top}$ is an element of the tangent space of $\Pi$ at $\pi(\gamma_*)$; that is,
	$u \mapsto u^{\top}$ represents a linear mapping from the cotangent space ${\cal T}$ to the tangent space $\overline {\cal T}$. %
	}
$$  H_{\pi} = \mathbb{E}  s_\pi (Y) s_\pi (Y)^{\top},\quad H_{\pi,\beta\gamma}   =   \mathbb{E} s_{\pi} (Y)   s_{\beta\gamma} (Y)' ,\quad H_{\beta\gamma,\pi}  =  \mathbb{E} s_{\beta\gamma} (Y)s_\pi(Y) ^{\top}.$$ In addition, we define the following projected versions of the gradient $\widetilde \nabla_{\pi}  
= \nabla_{\pi}  
-  H_{\pi,\beta\gamma}    H_{\beta\gamma}^{-1}     \nabla_{\beta\gamma}   $, score $\widetilde s_{\pi}(y)=s_{\pi}(y)-H_{\pi,\beta\gamma}    H_{\beta\gamma}^{-1}     s_{\beta\gamma}(y) $, and Hessian 
$ \widetilde H_{{\pi}}   
=   H_{\pi}   -  H_{\pi,\beta\gamma}   H_{\beta\gamma}^{-1}H_{\beta\gamma,\pi}$. 

The next lemma characterizes the minimum-MSE $h$ function in the locally quadratic case.

\begin{lemma}\label{lem_locquad}
	For the locally quadratic case of this section, the three following equivalent characterizations of $h_{\epsilon}^{\rm MMSE}$ defined in \eqref{MMSEproblem} hold:
\begin{align}
h_{\epsilon}^{\rm MMSE}(y)
= &   s_{\beta\gamma}(y) '     H_{\beta\gamma}^{-1}     
\, \nabla_{\beta\gamma}  \delta +   (\epsilon n)  \,  \widetilde s_{\pi}(y) ^\top
   \left(  \widetilde \nabla_{\pi} \delta -\mathbb{E}\left[h_{\epsilon}^{\rm MMSE}(Y) \widetilde s_{\pi} (Y)  
\right]\right)  
\label{SolutionMMSE_linsys}     \\
=&     s_{\beta\gamma} (y) '     H_{\beta\gamma}^{-1}     
\, \nabla_{\beta\gamma}  \delta  +   (\epsilon n) \,  \widetilde s_{\pi} (y)  ^\top
  \left(   \nabla_{\pi} \delta -\mathbb{E}\left[h_{\epsilon}^{\rm MMSE}(Y)  s_{\pi} (Y)  
\right]\right)    
\label{SolutionMMSE_rewrite}     \\
	= &    s_{\beta\gamma} (y) '    H_{\beta\gamma} ^{-1}     
	\, \nabla_{\beta\gamma}  \delta+     \widetilde s_{\pi} (y)  ^\top
	\,  \left[  \widetilde H_{{\pi}}  + (\epsilon n)^{-1}  \mathbb{I} \right]^{-1}
	\, \widetilde \nabla_{\pi} \delta   ,
	\label{SolutionMMSE}     
	\end{align}
	where $ \mathbb{I} $ denotes the identity map on ${\cal T}$.
	
\end{lemma}

\subsection{Covariates}
\label{sec:Covariates}

So far in our presentation, we have abstracted from covariates.
We now consider the case where in addition to the outcomes $Y_i$ we observe a vector of covariates $X_i$. We assume that $(Y_i,X_i)$ are randomly drawn from a conditional distribution of $Y_i$ given $X_i$ given by the model
$f_{\beta_0,\pi_0}(y\,|\,x)$, and an unrestricted marginal distribution $f_X$ of $X_i$. Our parameter of interest is $\delta_{\beta_0,\pi_0,f_X} =  \mathbb{E}_{f_X} \delta_{\beta_0,\pi_0}(X)$, where $\mathbb{E}_{f_X}$ denotes
an expectation over $f_X$. We consider estimators of the form 
$$\widehat{\delta}_h=
\frac 1 n \sum_{i=1}^n  \, \delta_{\widehat \beta,\pi(\widehat \gamma)}(X_i)+\frac 1 n \sum_{i=1}^n  \,h (Y_i,X_i,\widehat \beta,\widehat \gamma,\widehat f_X ) ,$$
where $\widehat \beta$ and $\widehat \gamma$ are preliminary estimates whose probability limits are
$\beta_0$ and $\gamma_*$,
and $ \widehat f_X$ is the empirical distribution of $X_i$ in the sample. While $f_X$ is unknown and infinite-dimensional, it only enters into our object
of interest (and the expression of $h_{\epsilon}^{\rm MMSE}$ below) as an expectation, and the corresponding sample average is still estimated at the $\sqrt{n}$-rate. 
We have the following characterization of the minimum-MSE influence function.
\begin{lemma}\label{lem_locquad_cov}
	Consider the locally quadratic case of this section. Let $H_{\beta\gamma}(x)$ and $H_{\pi,\beta\gamma}(x) $ be conditional counterparts to $H_{\beta\gamma}$ and $H_{\pi,\beta\gamma}$, and likewise let $s_{\beta\gamma}(y\,|\, x)$, $s_{\pi}(y\,|\, x)$ and $\widetilde s_{\pi}(y\,|\, x)=s_{\pi}(y\,|\, x)-
	\left[  \mathbb{E}_{f_X}H_{\pi,\beta\gamma}(X) \right] [\mathbb{E}_{f_X}H_{\beta\gamma}(X)]^{-1}s_{\beta\gamma}(y\,|\, x)$ denote (projected) scores in the conditional model. We have
	\begin{align}
	&h_{\epsilon}^{\rm MMSE}(y,x)
	\,=\,     \delta(x)-\mathbb{E}_{f_X}\delta(X) +  s_{\beta\gamma}(y\,|\, x) '     [\mathbb{E}_{f_X}H_{\beta\gamma}(X)]^{-1}     
	\, \mathbb{E}_{f_X}\nabla_{\beta\gamma}  \delta(X) 
	\nonumber  \\& \quad \quad \quad \quad\quad \quad\quad \quad    +   (\epsilon n)  \widetilde s_{\pi} (y\,|\, x) ^\top
	\left\{  \mathbb{E}_{f_X}\widetilde{\nabla}_{\pi} \delta(X) -\mathbb{E}_{f_X}\mathbb{E}\left[ h_{\epsilon}^{\rm MMSE}(Y,X) \widetilde s_{\pi} (Y\,|\, X)\right]   \right\} ,
	\label{SolutionMMSE_COV}    
	\end{align}
	with analogous counterparts to (\ref{SolutionMMSE_rewrite}) and (\ref{SolutionMMSE}).
\end{lemma}

A first difference between \eqref{SolutionMMSE_linsys}  and  \eqref{SolutionMMSE_COV} 
 is that various expectations over $f_X$ occur here, which we will replace by sample averages 
when calculating the  estimator $\widehat{\delta}^{\rm MMSE}_\epsilon$ in practice.
A second difference comes from the term $ \delta(x)-\mathbb{E}_{f_X}\,\delta(X)$.
However, this term does not contribute to $\widehat{\delta}^{\rm MMSE}_\epsilon$,
since its sample average is zero once we replace $f_X$ by the empirical distribution $\widehat f_X$.\footnote{The term $ \delta(x)-\mathbb{E}_{f_X}\,\delta(X)$ ensures that $h_{\epsilon}^{\rm MMSE}$ is locally robust with respect to $f_X$ in the sense of (\ref{Con:EtaGradient}).}

\subsection{Parametric models\label{subsec_par}}

A simple locally quadratic example is a parametric model where $\pi$ is finite-dimensional, and the distance measure over $\pi$ is based on a weighted Euclidean metric $\|\cdot\|_{\Omega}$ for a positive definite weight matrix $\Omega$. Here we treat $\Omega$ and the neighborhood size $\epsilon$ as known. We will discuss the choice of $\Omega$, and the interpretation of $\epsilon$, in Section \ref{Sec_epsilon}.

 The small-$\epsilon$ approximation to the bias of $\widehat{\delta}$ is given by (\ref{eqbepsilon}), with $\|\cdot\|_{\gamma_*}=\|\cdot\|_{\Omega^{-1}}$, where $\Omega^{-1}$ is the inverse of $\Omega$. 
In this case, for vectors $u,w \in \mathbb{R}^{\dim \pi}$ we have
$  u^\top \,  w  = u'  \, \Omega^{-1}  w $. Let $$\mathbb{H}_\pi= \mathbb{E}\left[ s_{\pi} (Y) s_{\pi} (Y)' \right] ,\quad  \widetilde{\mathbb{H}}_\pi= \mathbb{E}\left[ s_{\pi} (Y) s_{\pi} (Y)' \right]-\mathbb{E}\left[ s_{\pi} (Y) s_{\beta\gamma} (Y)' \right]H_{\beta\gamma}^{-1}\mathbb{E}\left[  s_{\beta\gamma} (Y)s_{\pi} (Y)' \right],$$ be the usual parametric Hessian matrices. We have $H_{\pi} = {\mathbb{H}}_\pi\Omega^{-1}$, $
H_{\pi,\beta\gamma}  =  \mathbb{E}\left[ s_{\pi} (Y) s_{\beta\gamma} (Y)' \right] $, $H_{\beta\gamma,\pi}=H_{\pi,\beta\gamma}'\Omega^{-1}$,
$    \widetilde s_{\pi}  
= s_{\pi}  
-  H_{\pi,\beta\gamma} H_{\beta\gamma}^{-1}     s_{\beta\gamma}   $,
and
$\widetilde H_{{\pi}}   =   \widetilde{\mathbb{H}}_\pi\Omega^{-1} $. From (\ref{SolutionMMSE}) we then obtain the following.

\begin{corollary}{(parametric models)}\label{SolutionMMSE_para}
\begin{align*}
h_{\epsilon}^{\rm MMSE}(y)
&=     s_{\beta\gamma} (y) '    H_{\beta\gamma} ^{-1}     
\, \nabla_{\beta\gamma}  \delta+       \widetilde s_{\pi} (y) '  
\,  \Omega^{-1}\left[  \widetilde H_{{\pi}} + (\epsilon n)^{-1} I \right]^{-1}  \, \widetilde \nabla_{\pi} \delta  ,
\end{align*}where $I$ is the identity matrix of size $\limfunc{dim}\pi$.
\end{corollary} 

In addition to the ``one-step efficient'' adjustment $h_{0}^{\rm MMSE} = s_{\beta\gamma} '    H_{\beta\gamma} ^{-1}     
 \nabla_{\beta\gamma}  \delta$, the minimum-MSE function $h_{\epsilon}^{\rm MMSE}$ in Corollary \ref{SolutionMMSE_para} thus provides a further adjustment that is motivated by robustness concerns. It is easy to generalize this formula to account for conditioning covariates whose distribution is unspecified, as in Lemma \ref{lem_locquad_cov}.

It is interesting to compute the limit of the MSE-minimizing $h$ function as $\epsilon$ tends to infinity in the case where
$H_\pi$ is invertible. This leads to the following expression, which is identical to \eqref{OneStepFullMLE},
\begin{align}
\lim_{\epsilon \rightarrow \infty}  h_{\epsilon}^{\rm MMSE}(y)
&=   s_{\beta\gamma} (y) '    H_{\beta\gamma} ^{-1}     
\, \nabla_{\beta\gamma}  \delta +     \widetilde s_\pi(y)'  
\;   \widetilde{\mathbb{H}}_\pi^\dagger  \; \; \widetilde \nabla_{\pi} \delta  ,
\label{OneStepFullMLEv2}     
\end{align}
where $\widetilde{\mathbb{H}}_\pi^\dagger$ denotes the Moore-Penrose generalized inverse of $\widetilde{\mathbb{H}}_\pi$. Comparing (\ref{OneStepFullMLEv2}) and Corollary \ref{SolutionMMSE_para} shows that the optimal  $\widehat \delta\,^{\rm MMSE}_\epsilon$ 
is a Ridge-regularized version of the one-step full MLE, where $(\epsilon n)^{-1}  I $  regularizes the projected Hessian matrix $\widetilde{ H}_{\pi}=\widetilde{\mathbb{H}}_\pi\Omega^{-1} $. Our ``robust'' adjustment remains well-defined under singularity, and it accounts for small or zero eigenvalues of the Hessian in an MSE-optimal way.

\paragraph{A linear regression example.}

Studying a linear regression model helps to illustrate some of the main features of our approach. Consider the model
\begin{eqnarray*}
	Y=X'\beta+U ,\quad 
	X=C Z+V,
\end{eqnarray*}
where $Y$ is a scalar outcome, and $X$ and $Z$ are random vectors of covariates and instruments, respectively, $\beta$ is a $\limfunc{dim}X$ parameter vector, and $C$ is a $\limfunc{dim}X\times \limfunc{dim}Z$ matrix. We assume that $U=\pi' V+\xi$, where $\xi$ is normal with zero mean and variance $\sigma^2$, independent of {$X,Z,V$}, and $V$ is normal with zero mean and non-singular covariance matrix $\Sigma_V$, independent of $Z$. Let $\Sigma_Z $ be the covariance matrix of $Z$, and let $\Sigma_X=C \Sigma_ZC'+\Sigma_V$. For simplicity we assume that $C$, $\Sigma_V$, $\Sigma_Z$, and $\sigma^2$ are known, and we take $\Omega=I$ to be the identity matrix. The parameters are thus $\beta$ and $\pi$. In the reference model we take $\pi=0$, hence treating $X$ as exogenous whereas the larger model allows for endogeneity. The target parameter is $\delta_{\beta_0,\pi_0}=c'\beta_0$, for a known $\limfunc{dim}\beta\times 1$ vector $c$.

From (\ref{SolutionMMSE_COV}) we have\footnote{In this case, there is no $\gamma$ parameter, $s_\beta (y,x\,|\,z)=\frac{1}{\sigma^2}x(y-x'\beta_0)$, $s_\pi (y,x\,|\,z)=\frac{1}{\sigma^2}(x-C z)(y-x'\beta_0)$, ${\mathbb{E}}_{f_Z}  H_{\beta}(Z)   = \frac{1}{\sigma^2}\Sigma_X$, $\widetilde\nabla_{\pi}   = \nabla_\pi - \Sigma_V\Sigma_X^{-1}\nabla_\beta$, and ${\mathbb{E}}_{f_Z}  \widetilde{H}_{\pi}(Z)   =  \frac{1}{\sigma^2}(\Sigma_V-\Sigma_V\Sigma_X^{-1}\Sigma_V)$.}
\begin{align}\label{eq_ex1_A}
&h_{\epsilon}^{\rm MMSE}(y,x,z)
=        (y-x'\beta_0)x'     \Sigma_X ^{-1}     
\, c \notag\\&\, \quad \,  -     (y-x'\beta_0) \left[  (x-C z)  - \Sigma_V\Sigma_X^{-1}x\right]'    \left[ \left(\Sigma_V-\Sigma_V\Sigma_X^{-1}\Sigma_V\right)   + (\epsilon n)^{-1}  \sigma^2 I \right]^{-1}  \, \Sigma_V\Sigma_X^{-1}c.
\end{align}
Hence, when $\epsilon=0$ the minimum-MSE estimator of $c'\beta_0$ is the ``one-step efficient'' adjustment in the direction of the OLS estimator, with influence function $h_{0}^{\rm MMSE}(y,x,z)
=     (y-x'\beta_0)x'     \Sigma_X ^{-1}     
\, c$. As $\epsilon $ tends to infinity, assuming $C \Sigma_ZC'$ is invertible, it follows from (\ref{eq_ex1_A}) that
\begin{align*}
&\lim_{\epsilon \rightarrow \infty} h_{\epsilon}^{\rm MMSE}(y,x,z)
=       (y-x'\beta_0) \left[  C z\right]'    \left[ C \Sigma_ZC'\right]^{-1}  \, c,
\end{align*}
which is the influence function of the IV estimator.

For given $\epsilon>0$ and $n$, our adjustment remains well-defined even when $C \Sigma_ZC'$ is singular.\footnote{Note that in the absence of an instrument $Z$, the minimum-MSE estimator coincides with the one-step efficient adjustment in the direction of the OLS estimator.\label{ftnote1}} When $c'\beta_0$ is identified (that is, when $c$ belongs to the range of $C$), the minimum-MSE estimator remains well-behaved as $\epsilon n$ tends to infinity, otherwise setting a finite $\epsilon$ value is needed to control the increase in variance. The term $(\epsilon n)^{-1}$ in (\ref{eq_ex1_A}) acts as a form of regularization, akin to Ridge regression. In Appendix \ref{App_GMM}, we show how to extend the parametric setting of this subsection to models defined by moment restrictions, and we revisit this example while dropping the normality assumptions. \hfill $\square$

\subsection{Semi-parametric mixture models\label{subsec_semipar}}

We now consider a class of semi-parametric models, where the distribution of outcomes $Y$ conditional on unobserved latent variables $A\in {\cal A}$
is described  parametrically by $Y \, \big|\, A \sim g_{\beta_0}(\cdot\, |\,A)$, with finite-dimensional unknown parameter $\beta_0 \in {\cal B}$,
while the distribution of $A \sim \pi_0$ is left unrestricted in the ``large'' correctly specified model. Here $\Pi$ is the set of probability distributions over $ {\cal A}$. The distribution of observed outcomes as a function of the unknown parameters $\beta_0 \in {\cal B}$ and $\pi_0 \in \Pi$ is given by
\begin{align}
   f_{\beta_0,\pi_0}(y) = \int_{\cal A} \, g_{\beta_0}(y|a) \, \pi_0(a) \, da.
   \label{ModelSemiparametric}
\end{align}
The parameter of interest is a functional of $\beta_0$ and $\pi_0$, which takes the form of an expectation over $A$; that is,
$$\delta_{\beta_0,\pi_0} =  \mathbb{E}_{\pi_0} \, \Delta_{\beta_0}(A)
= \int_{\cal A}  \, \Delta_{\beta_0}(a)  \, \pi_0(a) \, da ,
$$
where $ \Delta_{\beta_0}(a)$ is a known function of $\beta_0$ and $a$.

In Section \ref{Sec_numeric}, we will illustrate this setup in two binary choice models: a cross-sectional model, and a dynamic panel data model. In the first case, $A$ is an error term independent of covariates, normally distributed under the reference model. In the second case, $A$ is a latent individual effect correlated with initial conditions, specified using a parametric correlated random-effects reference model (Chamberlain, 1984). In both models, we will estimate average effects, which are expectations with respect to the distribution of $A$. Our approach will provide insurance against misspecification of the parametric functional forms.

Let us specify a parametric reference model for the distribution of the latent variables $A$, and denote the reference density by $\pi(\gamma)$, where $\gamma$ is a finite-dimensional parameter. Under the reference model, the distribution of outcomes is given by $f_{\beta_0,\pi(\gamma_*)}(y) = \int_{\cal A} \, g_{\beta_0}(y\,|\,a) \, \pi(a\,|\,\gamma_*) \, da
$. However, this model may be misspecified, and we assume that the true distribution $\pi_0$
belongs to the  neighborhood
$\Gamma_{\epsilon}(\gamma_*) = \{\pi_0 \in \Pi \,:\,d( \pi_0, \pi(\gamma_*) )\leq \epsilon\}$,
which we define here in terms of the  Kullback-Leibler (KL) divergence
$d( \pi_0, \pi(\gamma_*)  ) = 2 \, \mathbb{E}_{\pi_0}  \log[\pi_0(A) / \pi(A\,|\,\gamma_*)]$.

We are going to derive the expression of the minimum-MSE estimator by applying (\ref{SolutionMMSE_rewrite}). In this setting, elements of the cotangent space ${\cal T}$ of $\pi(\gamma)$ at $\gamma_*$
can be represented by functions $u : {\cal A} \mapsto \mathbb{R}$. For example, the gradient $\nabla_{\pi} \delta_{\beta,\pi} $ is a cotangent element, which can be represented by the function $\Delta_\beta(\cdot)$.\footnote{Note that, since $\pi$ integrates to one (and therefore tangent space elements integrate to zero), one can equivalently represent $\nabla_{\pi} \delta_{\beta,\pi} $ as $\Delta_\beta(\cdot)-c$ for any constant $c$. A possible choice is $c=\mathbb{E}_{\pi(\gamma_*)}\Delta_\beta(A)$. } For elements $u,w \in {\cal T}$ we define their inner product by
$ u^\top \, w     =  {\rm Cov}_{\pi(\gamma_*)}\left[ u(A) , w(A) \right] $;
that is, the corresponding squared norm in \eqref{DefDualNorm} is $
\left\|  u  \right\|^2_{\gamma_*} =  {\rm Var}_{\pi(\gamma_*)} \left[  u(A) \right] $. One can show that this norm is indeed the dual to a suitable local approximation to the KL divergence
as defined in \eqref{DefDualNorm}; see Appendix~\ref{App_der_sec3}.

Let us omit again parameter subscripts from the notation for conciseness. From \eqref{ModelSemiparametric}, we see that
 $s_{\pi}(y)=\nabla_{\pi} \log f(y)$ can be represented by the function $g(y\,|\,a) / f(y) $. As a result,
for  $u \in {\cal T}$ we have
\begin{align*}
  s_{\pi} (y) ^\top \, u
     &=  {\rm Cov}\left[ u(A) ,  \frac{g(y\,|\,A)} {f(y) } \right] 
     =  \mathbb{E}\left[ u(A) \,\big|\, Y=y\right]-\mathbb{E} u(A)  ,
\end{align*}
where we have used that $\mathbb{E} g(y\,|\,A) / f(y)  =1$. In addition, we have
\begin{align*}
 \widetilde s_{\pi} (y) ^\top \, u
&= \mathbb{E}\left[ u(A) \,\big|\, Y=y\right]-\mathbb{E} u(A) - s_{\beta\gamma} (y)'H_{\beta\gamma}^{-1}\mathbb{E}\left[  s_{\beta\gamma} (Y)u(A) \right],
\end{align*}
and, for any function $h$, $\mathbb{E}[h(Y)s_{\pi}(Y)]$ can be represented by the function $\mathbb{E}[h(Y)\,|\, A=a]$. 

Rewriting the first-order condition in equation \eqref{SolutionMMSE_rewrite}, we thus obtain the following result, which shows that $h_{\epsilon}^{\rm MMSE}$ is the solution to a linear system.

\begin{corollary}{(semi-parametric mixture models)}\label{SolutionMMSE_semiparam}
\begin{align*}
h_{\epsilon}^{\rm MMSE}(y)
=    &  s_{\beta\gamma} (y) '     H_{\beta\gamma}^{-1}     
\, \nabla_{\beta\gamma}  \delta  +   (\epsilon n) \bigg\{\mathbb{E}\left[\Delta(A)-\delta-\overline{h}_{\epsilon}^{\rm MMSE}(A)\,|\, Y=y\right]\notag\\
&\quad \quad \quad\quad \quad \quad  \quad \quad \quad \quad -s_{\beta\gamma} (y)'H_{\beta\gamma}^{-1}\mathbb{E} \left[s_{\beta\gamma} (Y) \left(\Delta(A)-\overline{h}_{\epsilon}^{\rm MMSE}(A)\right)     \right]\bigg\},     
\end{align*}	
	where $\overline{h}_{\epsilon}^{\rm MMSE}(a):=\mathbb{E}[h_{\epsilon}^{\rm MMSE}(Y)\,|\, A=a]$. 
\end{corollary}

 Corollary \ref{SolutionMMSE_semiparam} can be generalized to account for conditioning covariates $X$. We now apply Lemma \ref{lem_locquad_cov} to provide two generalizations, which we will use in the two examples in Section \ref{Sec_numeric}. In the first one, we assume that $A$ and $X$ are independent under $\pi_0$. This is the case in our cross-sectional illustration in Subsection \ref{subsec_bin}, where $A$ is an error term independent of $X$. In this case, $\pi_0$ is the marginal distribution of $A$. We then have the following characterization.

\begin{corollary}{(semi-parametric mixture models, independent covariates)}\label{MMSE_prob}
\begin{align*}
h_{\epsilon}^{\rm MMSE}(y,x)&=\delta(x)-\mathbb{E}_{f_X}\delta(X)+s_{\beta\gamma}(y\,|\, x)'     [\mathbb{E}_{f_X}H_{\beta\gamma}(X)]^{-1}    
\, \mathbb{E}_{f_X}\nabla_{\beta\gamma}\delta (X)  \notag\\& +   (\epsilon n)  \bigg\{\mathbb{E}\left[\mathbb{E}_{f_X}\left[\Delta(A,X)\right]-\mathbb{E}_{f_X}\delta(X)-\overline{h}_{\epsilon}^{\rm MMSE}(A)\,\big|\, Y=y,X=x\right]\notag \\
&  - s_{\beta\gamma}(y\,|\, x)' [\mathbb{E}_{f_X}H_{\beta\gamma}(X)]^{-1}    \mathbb{E}_{f_X}\mathbb{E}\,\bigg[s_{\beta\gamma}(Y\,|\, X)\left(\mathbb{E}_{f_X}\left[\Delta(A,X)\right]-\overline{h}_{\epsilon}^{\rm MMSE}(A)\right)\bigg]\bigg\},
\end{align*}
where here $\overline{h}_{\epsilon}^{\rm MMSE}(a):=\mathbb{E}_{f_X}[\mathbb{E}(h_{\epsilon}^{\rm MMSE}(Y,X)\,|\,  A=a,X)]$.	
\end{corollary}

In the second generalization, we leave the joint distribution of $(A,X)$ unrestricted under $\pi_0$. This is the case in our panel data illustration in Subsection \ref{subsec_panel},  where $A$ is an individual effect that may be correlated with $X$. In this case, $\pi_0$ is the conditional distribution of $A$ given $X$, and we measure the distance between conditional distributions using $d( \pi_0, \pi(\gamma_*)  ) = 2 \, \mathbb{E}_{f_X}\mathbb{E}_{\pi_0}  \log[\pi_0(A\,|\, X) / \pi(A\,|\,X,\gamma_*)]$. We then have the following characterization.

\begin{corollary}{(semi-parametric mixture models, correlated covariates)}\label{MMSE_prob2}
	\begin{align*}
	h_{\epsilon}^{\rm MMSE}(y,x)&=   \delta(x)-\mathbb{E}_{f_X}\delta(X)+s_{\beta\gamma}(y\,|\, x)'     [\mathbb{E}_{f_X}H_{\beta\gamma}(X)]^{-1} 
	\, \mathbb{E}_{f_X}\nabla_{\beta\gamma}\delta (X)  \notag\\
&+   (\epsilon n)  \bigg\{\mathbb{E}\left[\Delta(A,X)-\mathbb{E}_{f_X}\delta(X)-\overline{h}_{\epsilon}^{\rm MMSE}(A,X)\,\big|\, Y=y,X=x\right]\notag\\
	&- s_{\beta\gamma}(y\,|\, x)' [\mathbb{E}_{f_X}H_{\beta\gamma}(X)]^{-1}  \mathbb{E}_{f_X}  \mathbb{E}\,\bigg[s_{\beta\gamma}(Y\,|\, X)\left(\Delta(A,X)-\overline{h}_{\epsilon}^{\rm MMSE}(A,X)\right)\bigg]\bigg\},
	\end{align*}
	where here $\overline{h}_{\epsilon}^{\rm MMSE}(a,x):=\mathbb{E}\left(h_{\epsilon}^{\rm MMSE}(Y,X)\,|\,  A=a,X=x\right)$.
\end{corollary}

To provide intuition about the form of the solution in semi-parametric mixture models, let us start by considering a setting where $\beta_0$ and $\gamma_*$ are known to the researcher, while abstracting from covariates for simplicity. Let $\mathbb{E}_{{\cal{Y}}\,|\, {\cal{A}}}$ and $\mathbb{E}_{{\cal{A}}\,|\, {\cal{Y}}}$ denote the conditional expectation operators of $Y$ given $A$ and $A$ given $Y$, respectively. Corollary \ref{SolutionMMSE_semiparam} implies that (see Appendix \ref{App_der_sec3} for a derivation)
\begin{align}
h_{\epsilon}^{\rm MMSE}
=    &   \mathbb{E}_{{\cal{A}}\,|\, {\cal{Y}}}\left[\mathbb{E}_{{\cal{Y}}\,|\, {\cal{A}}}\circ\mathbb{E}_{{\cal{A}}\,|\, {\cal{Y}}}+(\epsilon n)^{-1}\mathbb{I}_{\cal{A}}\right]^{-1}(\Delta-\delta), \label{eq_minimum_IP_sol}
\end{align}
where $\circ$ denotes the composition operator, and
 $\mathbb{I}_{\cal{A}}$ denotes the identity operator; that is, $\mathbb{I}_{\cal{A}}\pi=\pi$ for $\pi:{\cal{A}}\rightarrow\mathbb{R}$. In semi-parametric mixture settings such as panel data models, average effects are often only partially identified or not root-$n$ estimable due to ill-posedness.\footnote{See, e.g., Chernozhukov \textit{et al.} (2013), Pakes and Porter (2013), Severini and Tripathi (2012), and Bonhomme and Davezies (2017).} The presence of the Tikhonov penalty $(\epsilon n)^{-1}$ in (\ref{eq_minimum_IP_sol}) bypasses these issues by making the operator $[\mathbb{E}_{{\cal{Y}}\,|\, {\cal{A}}}\circ\mathbb{E}_{{\cal{A}}\,|\, {\cal{Y}}}+(\epsilon n)^{-1}\mathbb{I}_{\cal{A}}]$ non-singular. By focusing on a shrinking neighborhood of the reference distribution, as opposed to entertaining any possible distribution, our approach avoids issues of non-identification and ill-posedness while guaranteeing MSE-optimality within that neighborhood.\footnote{Note that regular estimation is possible when there exists a function $\psi(y)$ such that $\Delta(A)-\delta=\mathbb{E}[\psi(Y)\,|\, A]$. In this case ${\limfunc{lim}}_{\epsilon\rightarrow \infty} \,\widehat{\delta}_{\epsilon}^{\rm MMSE}$ is consistent for $\mathbb{E}_{\pi_0}\, \Delta(A)=\delta+\mathbb{E}_{\pi_0}\, \mathbb{E}[\psi(Y)\,|\, A]$ for all $\pi_0$. }

Next, consider the estimation of $c'\beta_0$, for $c$ a $\limfunc{dim}\beta\times 1$ vector, and assume $\gamma_*$ known for simplicity. Let $\mathbb{I}_{\cal{Y}}h=h$ for $h:{\cal{Y}}\rightarrow\mathbb{R}$, and let \begin{equation*}\mathbb{W}^{\epsilon}=  \mathbb{I}_{\cal{Y}}-\mathbb{E}_{{\cal{A}}\,|\, {\cal{Y}}} \left[\mathbb{E}_{{\cal{Y}}\,|\, {\cal{A}}}\circ\mathbb{E}_{{\cal{A}}\,|\, {\cal{Y}}}+(\epsilon n)^{-1}\mathbb{I}_{\cal{A}}\right]^{-1} \mathbb{E}_{{\cal{Y}}\,|\, {\cal{A}}}.\end{equation*}
It follows from Corollary \ref{SolutionMMSE_semiparam} that
\begin{equation}
h^{\rm MMSE}_{\epsilon}(y)=\mathbb{W}^{\epsilon}\, s_{\beta} (y)'\left\{\mathbb{E}\left[s_{\beta} (Y)\mathbb{W}^{\epsilon}\,s_{\beta} (Y)'\right]\right\}^{-1}c.\label{regu_IF_semi}
\end{equation}
As $\epsilon$ tends to infinity, $\mathbb{W}^{\epsilon}$ approximates the functional differencing projection operator $\mathbb{W}=\mathbb{I}_{\cal{Y}}-\mathbb{E}_{{\cal{A}}\,|\, {\cal{Y}}} \mathbb{E}_{{\cal{A}}\,|\, {\cal{Y}}}^{\dagger}$, where $\mathbb{E}_{{\cal{A}}\,|\, {\cal{Y}}}^{\dagger}$ denotes the Moore-Penrose generalized inverse of $\mathbb{E}_{{\cal{A}}\,|\, {\cal{Y}}}$ (see Bonhomme, 2012). In this limit, the minimum-MSE estimator is the one-step approximation to the semi-parametric efficient estimator of $c'\beta_0$. Yet, the efficient estimator fails to exist when the matrix denominator in (\ref{regu_IF_semi}) is singular.\footnote{In discrete choice panel data models, common parameters are generally not point-identified (Chamberlain, 2010, Honor\'e and Tamer, 2006). In panel data models with continuous outcomes, identification and regularity require high-level ``non-surjectivity'' conditions which may be hard to verify (Bonhomme, 2012).} Here the term $(\epsilon n)^{-1}$ acts as a regularization of the functional differencing projection, which makes $h^{\rm MMSE}_{\epsilon}$ well-defined irrespective of the nature of identification. 

Lastly, consider a model with covariates $X$ that are independent of the latent variables $A$, as in our illustration in Subsection \ref{subsec_bin}. Assuming that $\beta_0$ and $\gamma_*$ are known, and that $\Delta(A)$ does not depend on $X$, and letting $\mathbb{E}_{{\cal{Y}},{\cal{X}}\,|\, {\cal{A}}}$ and $\mathbb{E}_{{\cal{A}}\,|\, {\cal{Y}},{\cal{X}}}$ denote the conditional expectation operators of $(Y,X)$ given $A$ and $A$ given $(Y,X)$, respectively, Corollary \ref{MMSE_prob} implies
\begin{align}
h_{\epsilon}^{\rm MMSE}
=    &   \mathbb{E}_{{\cal{A}}\,|\, {\cal{Y}},{\cal{X}}}\left[\mathbb{E}_{{\cal{Y}},{\cal{X}}\,|\, {\cal{A}}}\circ\mathbb{E}_{{\cal{A}}\,|\, {\cal{Y}},{\cal{X}}}+(\epsilon n)^{-1}\mathbb{I}_{\cal{A}}\right]^{-1}\mathbb{E}_{f_X}(\Delta-\delta). \label{eq_minimum_IP_sol_X}
\end{align} 
The solution is similar to (\ref{eq_minimum_IP_sol}), with the difference that here, due to independence, both $Y$ and $X$ are informative about the latent $A$.

\subsection{Implementation in parametric models and semi-parametric mixture models}

To implement the method in parametric settings, the researcher needs to compute the score and Hessian of the larger model. Since we focus on smooth models, methods based on numerical derivatives or simulation-based approximations can be used. Minimum-MSE estimators are generally not available in closed form in semi-parametric mixture models. {Nevertheless, in these models $h_{\epsilon}^{\rm MMSE}$ can be computed by minimizing the quadratic objective (abstracting from covariates for simplicity)
	\begin{align}
		\mathbb{E} \left(\Delta (A)-\delta-\mathbb{E}[h(Y)\,|\, A]\right)^2 +(\epsilon n)^{-1}\mathbb{E}  h^2(Y)   
 ,
\end{align}
	  with respect to $h$, subject to the linear constraints $ \mathbb{E}  h(Y) =0 $ and
	  $\mathbb{E}  \,  h(Y)    s_{\beta\gamma} (Y) =  \nabla_{\beta\gamma}  \delta$. This is a regularized linear inverse problem (see, e.g., Engl \textit{et al.}, 2000, and Kress, 2014), which is well-posed given the presence of the Tikhonov penalty $(\epsilon n)^{-1}$. Numerous numerical approaches have been developed to solve linear inverse problems. In the illustrations we implement a simulation-based method that relies on matrix operations. In Appendix \ref{App_compute} we describe this method, and also explain how we compute confidence intervals. Note that, given initial estimates $\widehat{\beta}$ and $\widehat{\gamma}$, computing minimum-MSE estimators and confidence intervals does not require nonlinear optimization.}

\section{Using minimum-MSE estimates and confidence intervals for sensitivity analysis \label{Sec_epsilon}}

In this section, we discuss how to apply our approach in practice. So far, we have shown how to compute minimum-MSE estimators and confidence intervals for different values of $\epsilon$. We now describe a strategy to choose an interpretable range for $\epsilon$, and report the results of the estimation on this range, in the spirit of sensitivity analysis.
	
	\subsection{Context-specific interpretations of magnitudes of $\epsilon$\label{sec41}}
	
In order to apply our approach, an important first step is to form intuition about orders of magnitudes of $\epsilon$ in the model under study. The literature on sensitivity analysis that we referred to in the introduction provides intuition in certain models; see for example the analysis of linear IV models in Conley \textit{et al.} (2012). In Sections \ref{App_TW} and \ref{Sec_numeric}, we will discuss the magnitudes of $\epsilon$ in our examples. In our evaluation of a conditional cash transfer program in Section \ref{App_TW}, misspecification stems from omitted stigma effects in households' preferences. In this case, we will show that $\epsilon$ can be mapped to the ratio of the marginal utility of the subsidy (that is, the ``stigma'' effect) to that of consumption. Economic intuition can then suggest whether a given $\epsilon$ is large or small. {In this setting, one can motivate the assumption that $\epsilon$ shrinks as $n$ increases as reflecting that the econometrician's uncertainty about the presence of stigma effects diminishes when the sample gets larger.}  
	
	In the binary choice models we study in Section \ref{Sec_numeric}, $\pi$ is infinite-dimensional, and $\epsilon/2$ is the squared radius of a Kullback-Leibler ball around a normal density. To visualize the implications of assuming that the true $\pi$ belongs to an $\epsilon$-neighborhood of the normal, we will plot worst-case probability bounds, and show how setting $\epsilon$ to a particular value imposes \textit{ex-ante} restrictions on the parameter of interest. Local approximations to bounds on functionals of $\pi$ are easy to compute. Alternatively, one may report estimated worst-case distributions --- i.e., a $\pi_0$ that achieves the supremum in (\ref{eqbias}) --- in the spirit of Christensen and Connault~(2019).

		\subsection{Interpretation based on statistical testing\label{sec42}}

	As a complement to forming context-specific intuition about magnitudes, here we outline a generic interpretation based on statistical testing. Specifically, we show that setting $\epsilon$ is isomorphic to setting a lower bound on the local power of a likelihood-ratio test of the reference model, against alternatives outside the neighborhood $\Gamma_{\epsilon}(\gamma_*)$ in certain directions. The $\epsilon$-neighborhood will thus contain all models that are hard to statistically distinguish from the reference model in those directions. This logic has antecedents in robust statistics (Huber and Ronchetti, 2009), and robust control in economics (Hansen and Sargent, 2008).

	To proceed, let us focus on the parametric case of Subsection \ref{subsec_par} with identity weight matrix $\Omega$. Let $v$ be a unitary vector, and consider a likelihood-ratio test of the null hypothesis $H_0: \, \pi_0=\pi(\gamma_*)$ against the local alternative $H_1:\, \pi_0=\pi(\gamma_*)+\xi v/\sqrt{n}$, for some constant $\xi>0$. Let the size of the test be $\alpha\in(0,1)$. The local power of the test is then $p=\Pr\left(Z[\mu]>\widetilde{c}_{\alpha}\right)$,
	where $\widetilde{c}_{\alpha}$ is the $(1-\alpha)$-quantile of the chi-squared distribution with one degree of freedom, and $Z[\mu]$ follows a non-central chi-squared distribution with one degree of freedom and non-centrality parameter $\mu=\|\widetilde{H}_{\pi}^{1/2} v\|\xi$; see, e.g., Van der Vaart (2007, page 237).\footnote{Here $\widetilde{H}_{\pi}$ is the usual parametric (projected) Hessian matrix, since $\Omega$ is the identity.} For given $\alpha$ and $p$ values, let $\mu(\alpha,p)$ be such that $\Pr\left(Z[\mu(\alpha,p)]>\widetilde{c}_{\alpha}\right)=p$.\footnote{$\mu(\alpha,p)$ is implicitly defined by $\Phi(\mu(\alpha,p)+\Phi^{-1}(\alpha/2))+\Phi(-\mu(\alpha,p)+\Phi^{-1}(\alpha/2))=p$,
		where $\Phi$ is the standard normal cumulative distribution function.} It follows that
	$ \mu(\alpha,p)=\|\widetilde{H}_{\pi}^{1/2} v\|\xi$. Hence, noting that $\mu(\alpha,p)$ is increasing in $p$, and defining$$\epsilon(v)=\frac{\mu(\alpha,p)^2}{nv'\widetilde{H}_{\pi}v},$$ 
	taking $\epsilon\geq \epsilon(v)$ ensures that local power in direction $v$ is at least $p$ whenever $\xi/\sqrt{n}\geq \epsilon^{\frac{1}{2}}$.\footnote{This definition is easy to extend to the general locally quadratic case of Subsection \ref{sec:LocallyQuadratic}. Let $v$ be a unitary direction in the tangent space $\overline{\cal{T}}$ of $\pi(\gamma)$ at $\gamma_*$, and let $\Omega_{\gamma_*} \, : \, \overline {\cal T} \rightarrow  {\cal T}$ be the linear operator defined in the supplementary appendix. In the parametric case $\Omega_{\gamma_*}$ is simply the matrix $\Omega$. In the general setup the non-centrality parameter is
		$\langle v,\widetilde{H}_{\pi}\Omega_{\gamma_*}v\rangle^{1/2} \, \xi$, 
		and $\epsilon(v)=\mu(\alpha,p)^2/(n\langle  v,\widetilde{H}_{\pi}\Omega_{\gamma_*}v\rangle)$, for $\left\langle v,u\right\rangle\in\mathbb{R}$ the scalar product between $v\in\overline{{\cal{T}}}$ and $u\in{\cal{T}}$.} Note that, for fixed $\alpha$ and $p$, the product $\epsilon(v)  n$ is independent of $n$.

	To ensure power larger than $p$ outside the neighborhood in \emph{all} directions $v$, one could compute the \emph{supremum} of $\epsilon(v)$ over all directions.\footnote{It is sufficient to consider directions that are orthogonal to the directions $\nabla_{\gamma_*}\pi'$ of the reference model.} 	Setting $\epsilon$ larger than all $\epsilon(v)$'s is motivated by a desire to calibrate the fear of misspecification of the researcher: when $p$ is large, say $80\%$, all alternatives outside the neighborhood $\Gamma_{\epsilon}(\gamma_*)$ are then easy to statistically distinguish from the reference model based on a sample of $n$ observations. However, for the supremum of $\epsilon(v)$ to be finite, $\widetilde{H}_{\pi}$ needs to be non-singular, which precludes models with partial or irregular identification. As an example, in the linear model of Subsection \ref{subsec_par}, the supremum of $\epsilon(v)$ is infinite whenever ${\Sigma}_X-{\Sigma}_V=C \Sigma_ZC'$ is singular; that is, whenever the IV model is under-identified. In such a case, there thus exist certain directions along which the specification test has no power, no matter how large $\epsilon$ is. Likewise, in semi-parametric models, the eigenvalues of the infinite-dimensional operator $\widetilde{H}_{\pi}$ may not be bounded away from zero due to ill-posedness.

In partially or irregularly identified models,  given some fixed values of $\alpha$ and $p$, a possibility is to report several $\epsilon$ value: a \emph{first} value $\epsilon_1$ that corresponds to the \emph{infimum} of $\epsilon(v)$ over all directions $v$ --- hence to the most favorable direction; a \emph{second} value $\epsilon_2\geq \epsilon_1$, such that power is at least $p$ outside the neighborhood in the most favorable direction in the subspace of directions orthogonal to the most favorable one; a \emph{third} value $\epsilon_3\geq \epsilon_2$ that provides power guarantees along the most favorable direction orthogonal to the previous two ones, and so on. Letting $\lambda_k(B)$ denote the $k$-th largest eigenvalue of $B$, we have
	\begin{equation}\label{seq_eps_gen}
	\epsilon_k=\frac{\mu(\alpha,p)^2}{n\lambda_{\rm k}(\widetilde{H}_{\pi})},\quad \text{for } k=1,2,...
	\end{equation}  
	We will report the first few $\epsilon_k$ values in our illustrations --- taking $\alpha=5\%$ and $p=80\%$ --- as a complement to context-specific interpretations of orders of magnitude.\footnote{In Appendix \ref{App_compute}, we describe how to compute $\epsilon_k$ in semi-parametric mixture models using a simulation-based approach.}

	\subsection{Reporting results}
	
	By providing intuition about $\epsilon$, either through an interpretation of magnitudes in the context under study (see Subsection \ref{sec41}), and/or through a generic approach based on statistical testing (see Subsection \ref{sec42}), the researcher selects a range of possible values for $\epsilon$. We then recommend plotting the minimum-MSE estimator, and its associated 95\% confidence interval, as a function of $\epsilon$ on this range. In our illustrations, we will use this device to report results, and we will indicate particular values of $\epsilon$ on the x-axis to facilitate interpretation.   
	
By reporting those minimum-MSE estimates and bias-adjusted confidence intervals, we learn about the
fragility of the estimation results under the reference model parameterized by $\gamma$, relative to the larger model parameterized by $\pi$.
Exploring the sensitivity to model misspecification in that way is in line with the traditional suggestion of comparing estimation 
results obtained from different model specifications (see, e.g., Leamer, 1983, 1985). However, our local approach does not require the researcher to estimate the larger model,
which is particularly relevant in situations where the latter is partially or irregularly identified, or computationally hard to estimate.

\subsection{Remark: shape of neighborhoods and choice of norm}

Implementing our approach requires choosing a norm on $\Pi$, which governs the shape of $\Gamma_{\epsilon}(\gamma_*)$. In parametric models, the researcher may have a preferred weight matrix $\Omega$, thus putting more weight on certain elements of the vector $\pi$. An automatic weighting scheme is to set $\Omega$ to be equal to the diagonal of the projected Hessian matrix $ \widetilde{\mathbb{H}}_\pi$. This choice can be motivated using a statistical testing logic as in Subsection \ref{sec42}, focusing on component-wise directions in the canonical basis of $\mathbb{R}^{\limfunc{dim}\pi}$. Taking the diagonal, instead of the entire matrix $ \widetilde{\mathbb{H}}_\pi$, as a weight is in line with our aim to cover models where the parameter of interest may not be regularly estimable.\footnote{In applications, other norms may have particular appeal. For example, measuring deviations according to the supremum norm will lead to an $\ell^1$ dual norm in (\ref{MMSEproblem}), in the spirit of Armstrong and Koles\'ar (2021). While our estimators and confidence intervals remain well-defined in this case, that setting is not locally quadratic.} 

In semi-parametric mixture models, where $\Pi$ is a set of densities, we rely on the Kullback-Leibler divergence for computational convenience. KL is locally quadratic, and this choice allows us to obtain the explicit characterizations of Lemma \ref{lem_locquad}, and to compute minimum-MSE estimators by solving linear systems. Note, however, that the KL divergence does not impose shape or smoothness restrictions on the densities inside the neighborhood.

\section{Empirical illustration: conditional cash transfers in Mexico\label{App_TW}}

The goal of this section is to predict program impacts in the context of the PROGRESA conditional cash transfer program, building on the structural evaluation of the program in Todd and Wolpin (2006, TW hereafter) and Attanasio \textit{et al.} (2012, AMS). We estimate a simple model in the spirit of TW, and adjust its predictions against a specific form of misspecification under which the program may have a ``stigma'' effect on preferences.

\subsection{Setup}

Following TW and AMS, we focus on PROGRESA's education component, which consists of cash transfers to families conditional on children attending school. Those represent substantial amounts as a share of total household income. The implementation of the policy was preceded by a village-level randomized evaluation in 1997-1998. As TW and AMS point out, the randomized control trial is silent about the effect that other, related policies could have, such as higher subsidies or unconditional income transfers, which motivates the use of structural methods.

To analyze this question, we consider a simplified version of TW's model described in Wolpin (2013), which is a static, one-child model with no fertility decision. To describe this model, let $U(C,S,\tau,v)$ denote the utility of a unitary household, where $C$ is consumption, $S\in\{0,1\}$ denotes the schooling attendance of the child, $\tau$ is the level of the PROGRESA subsidy, and $v$ are taste shocks. Utility may also depend on characteristics $X$, which we abstract from for conciseness in the presentation.\footnote{Empirically, we include as covariates the age of the child and her parents, distance to the nearest school, eligibility and year indicators, and the highest grade obtained. We perform estimation separately by gender.} Note the direct presence of the subsidy $\tau$ in the utility function, which may reflect a stigma effect. This direct effect plays a key role in the analysis. The budget constraint is: $C=Y+W(1-S)+\tau S$, where $Y$ is household income and $W$ is the child's wage. This is equivalent to: $C=Y+\tau+(W-\tau)(1-S)$. Hence, {in the absence of a direct effect on utility}, the program's impact is equivalent to an increase in income and a decrease in the child's wage. 

Following Wolpin (2013) we parameterize the utility function as
$$U(C,S,\tau,v)=aC+bS+dCS+\lambda \tau S+Sv,$$
where $\lambda$ denotes the direct (stigma) effect of the program. The schooling decision is then
$$S=\boldsymbol{1}\{U(Y+\tau,1,\tau,v)>U(Y+W,0,0,v)\}=\boldsymbol{1}\{v>a(Y+W)-(a+d)(Y+\tau)-\lambda \tau-b\}.$$
Assuming that $v$ is standard normal, independent of wages, income, and program status (that is, of the subsidy $\tau$), we obtain
$$\Pr(S=1\,|\, y,w,\tau)=1-\Phi\left[a(y+w)-(a+d)(y+\tau)-\lambda \tau-b\right],$$
where $\Phi$ is the standard normal cdf.

We use the specification with $\lambda= 0$ as our reference model, and estimate it on control villages. When $\lambda=0$, the average effect of the subsidy on school attendance is 
\begin{align*}
&\mathbb{E}\left[\Pr(S=1\,|\, Y,W,\tau=\tau^{\rm treat})-\Pr(S=1\,|\, Y,W,\tau=0)\right]\\
&=\mathbb{E}\left(\Phi\left[a(Y+W)-(a+d)(Y+\tau^{\rm treat})-b\right]-\Phi\left[a(Y+W)-(a+d)Y-b\right]\right).
\end{align*}
As Wolpin (2013) notes, data under the subsidy regime ($\tau=\tau^{\rm treat}$) is {not} needed to construct an empirical counterpart to this quantity, since treatment status is independent of $Y,W$.\footnote{AMS make a related point (albeit in a different model), and use both control and treated villages to estimate their structural model. AMS also document the presence of general equilibrium effects of the program on wages. We abstract from such effects in our analysis.}

We contrast two strategies to predict the effect of the program and other counterfactual policies, while accounting for misspecification of the reference model due to the presence of stigma effects. The first strategy --- which we refer to as \textit{ex-ante} policy prediction --- is only based on data from control villages, whereas the second strategy --- \textit{ex-post} prediction --- combines both control and treated villages. In both cases, we allow for $\lambda\neq 0$ in the larger model. While in the present simple static context one could easily estimate a version of the larger model, in dynamic structural models such as the one estimated by TW, estimating a different model in order to assess the impact of any given form of misspecification may be computationally prohibitive. This highlights an advantage of our approach, which does not require the researcher to estimate the parameters under a new model.

To cast this setting into our framework, let $\beta=(a,b,d)$, $\pi=\lambda$, and $$\delta_{\beta,\pi}=\mathbb{E}\left(\Phi\left[a(Y+W)-(a+d)(Y+\tau^{\rm treat})-\lambda \tau^{\rm treat}-b\right]-\Phi\left[a(Y+W)-(a+d)Y-b\right]\right).$$ We focus on the effect on eligible (i.e., poorer) households. We add covariates to gender-specific school attendance equations, which include the age of the child and her parents, year indicators, distance to school, and an eligibility indicator. We report estimates of $\delta_{\beta_0,\pi_0}$ as well as confidence intervals.

\subsection{Empirical results}

We use the sample from TW. We drop observations with missing household income, and focus on 1219 boys and 1089 girls aged 12 to 15.\footnote{Children's wages are only observed for those who work. We {impute potential wages} to all children based on a linear regression that in particular exploits province-level variation and variation in distance to the nearest city, similar to AMS.} Descriptive statistics on the sample show that average weekly household income is 242 pesos, the average weekly wage is 132 pesos, and the PROGRESA subsidy ranges between 31 and 59 pesos per week depending on grade and gender. Average school attendance drops from 90\% at age 12 to between 40\% and 50\% at age 15.

\begin{figure}[h!]
	\caption{Effect of the PROGRESA subsidy as a function of neighborhood size $\epsilon$}
	\label{Fig_prog}
	\begin{center}
		\begin{tabular}{cc}
			\multicolumn{2}{c}{\textit{Ex-ante}}\\	Girls & Boys \\
						\includegraphics[width=70mm, height=50mm]{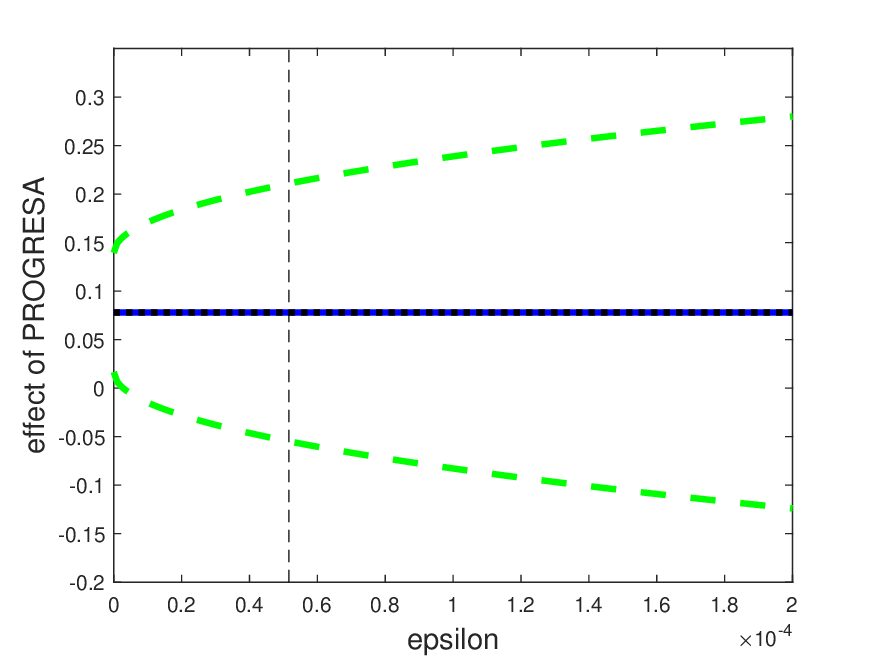}&	\includegraphics[width=70mm, height=50mm]{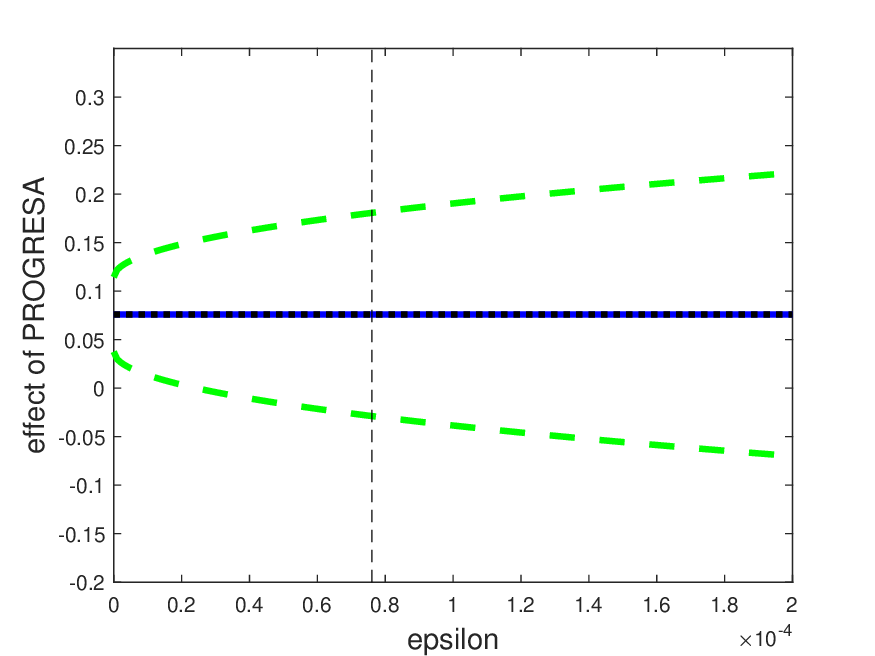}\\
							\multicolumn{2}{c}{}\\
			\multicolumn{2}{c}{\textit{Ex-post}}\\
				Girls & Boys \\
							\includegraphics[width=70mm, height=50mm]{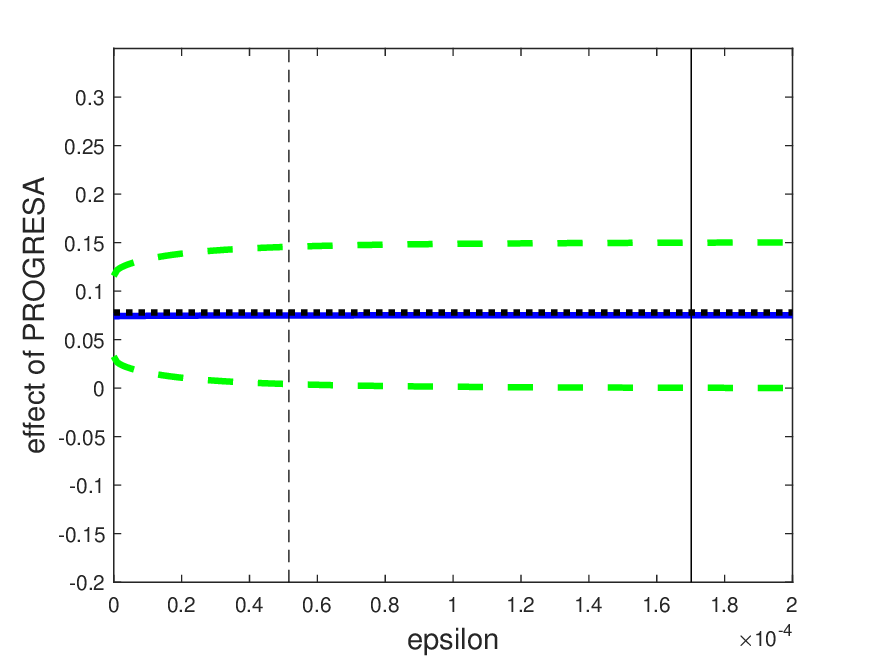}&	\includegraphics[width=70mm, height=50mm]{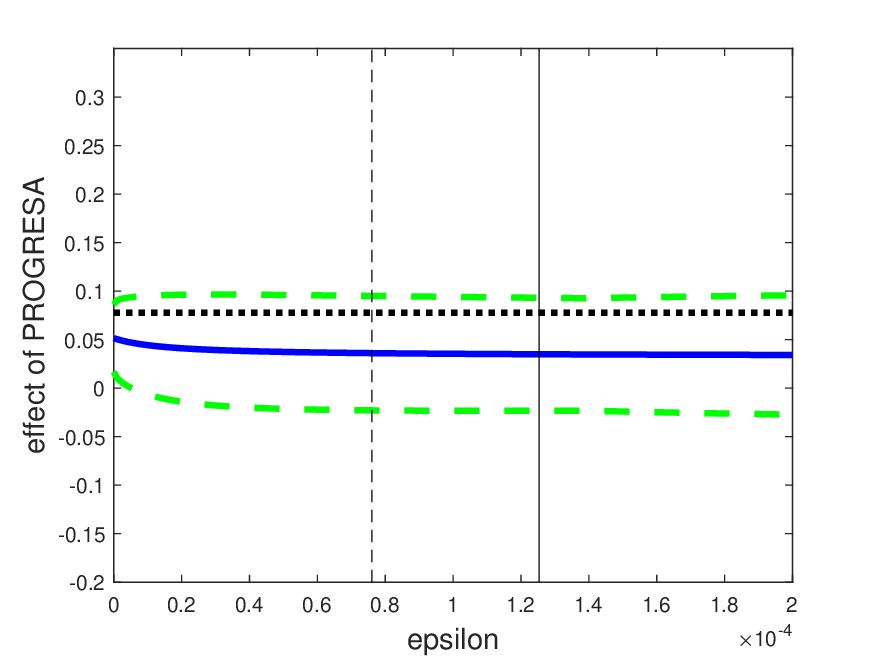}\\\end{tabular}%
	\end{center}
	\par
	\textit{{\small Notes: Sample from Todd and Wolpin (2006). In the top panel, we show estimates based on control villages only, in the bottom panel we show estimates based on both controls and treated. We report $\epsilon$ on the x-axis. The minimum-MSE estimates of the effect of PROGRESA on school attendance are shown in solid. 95\% confidence intervals based on those estimates are in dashed. The dotted line shows the unadjusted prediction based on the reference model. The dashed vertical line indicates $\overline{\epsilon}$ (at which households value consumption and the subsidy equally) and the solid vertical line indicates $\epsilon_1$ (at which a 5\% likelihood ratio specification test has power 80\%). Girls (left) and boys (right).}}
\end{figure}

We start by providing intuition regarding the values of $\epsilon$ in the present context. In the structural model, $\lambda$ is the marginal utility of the subsidy for households sending their child to school. In turn, the marginal utility of consumption is given by $a+d$. Hence, bounding $\lambda^2$ by $\epsilon$ is equivalent to bounding the ratio of marginal utility of the subsidy to marginal utility of consumption by $\sqrt{\epsilon}/(a+d)$. For example, households valuing the subsidy as much as consumption in absolute value --- arguably an upper bound on the stigma effect --- corresponds to $\overline{\epsilon}=(a+d)^2$.

\begin{table}[h!]\caption{Effect of the PROGRESA subsidy and counterfactual reforms\label{Table_TW}}
	\begin{center}
		\begin{tabular}{l||cccc}
			& 	\multicolumn{2}{c}{\textit{Ex-ante}}& 	\multicolumn{2}{c}{\textit{Ex-post}}  \\\hline\hline\\
			& \multicolumn{4}{c}{PROGRESA impacts}\\\hline
			& Girls & Boys & Girls & Boys\\
			Model-based estimate & 0.078 & 0.076 & 0.078 & 0.078 \\
			\hskip .3cm non-robust CI& (0.017,0.140) &(0.038,0.114)&(0.026,0.130)&(0.043,0.112)\\
			\hskip .3cm robust CI& (-0.055,0.211) & (-0.029,0.181) & (0.006,0.150)  & (0.008,0.148)\\
			MMSE estimate & 0.078 & 0.076 & 0.075 & 0.036 \\
			\hskip .3cm robust CI& (-0.055,0.211) & (-0.029,0.181)  & (0.004,0.146)  & (-0.023,0.095)\\
			Experimental &- &- &0.087 &0.050		 \\\hline\\
			& \multicolumn{4}{c}{Counterfactual 1: doubling subsidy}\\\hline
			& Girls & Boys & Girls & Boys\\
			Model-based estimate & 0.142 & 0.133 & 0.142 & 0.136 \\
			\hskip .3cm robust CI& (-0.069,0.353) & (-0.018,0.284) & (0.027,0.258)  & (0.033,0.239)\\
			MMSE estimate & 0.142 & 0.133 & 0.137 & 0.076 \\
			\hskip .3cm robust CI& (-0.069,0.353) & (-0.018,0.284) & (0.022,0.251)  & (-0.017,0.163)\\\hline\\
			& \multicolumn{4}{c}{Counterfactual 2: unconditional transfer}\\\hline	
			& Girls & Boys & Girls & Boys \\
			Model-based estimate & 0.003 & 0.005 & 0.003 & 0.005\\
			\hskip .3cm robust CI& (-0.213,0.218) & (-0.231,0.240) & (-0.208,0.214)  & (-0.239,0.249)\\
			MMSE estimate & 0.003 & 0.005 & 0.009 & -0.071 \\
			\hskip .3cm robust CI& (-0.213,0.218) & (-0.231,0.240) & (-0.232,0.250)  & (-0.279,0.136)\\\hline
		\end{tabular}
	\end{center}
	{\textit{\small Notes: Sample from Todd and Wolpin (2006). In the left two columns we show estimates based on control villages only, in the right two columns we show estimates based on both controls and treated. ``Model-based'' estimates are based on the reference model. $\epsilon =\overline{\epsilon}$, which corresponds to households valuing consumption and the subsidy equally in absolute value. CI are 95\% confidence intervals. The unconditional transfer amounts to 5000 pesos in a year.}}
	
\end{table}

In Figure \ref{Fig_prog}, we show the minimum-MSE estimator of the impact of the PROGRESA subsidy on eligible households, together with 95\% confidence intervals, for a range of values around $\overline{\epsilon}=(a+d)^2$ (which we show in the dashed vertical line). In the horizontal dotted line, we show estimates based on the reference model. In the top panel, we show \textit{ex-ante} prediction results based on control villages only. We see that the minimum-MSE estimator and the one based on the reference model are equal in this case. This is intuitive, since $\pi=\lambda$ is scalar, and control villages provide no information about it.\footnote{This is analogous to the case of a linear regression with endogeneity and no instrument, which we mentioned in footnote \ref{ftnote1}.} However, the confidence intervals --- which account for model misspecification --- are large. When $\epsilon=\overline{\epsilon}$, the 95\% intervals include zero for both genders. This quantifies the uncertainty associated with \textit{ex-ante} prediction when the researcher does not rule out the presence of stigma.

In the bottom panel of Figure \ref{Fig_prog}, we show the results of \textit{ex-post} prediction based on both control and treated villages. In the sample of boys, the minimum-MSE estimator is lower than the one based on the reference model, suggesting that the reference model is misspecified. In contrast, the two estimators are close to each other in the sample of girls.\footnote{Note that, for boys, minimum-MSE estimates at all $\epsilon$ values --- including $\epsilon=0$ --- are lower than the estimate from the reference model {(note that here we estimate the reference model using control villages only, and use both controls and treated to compute the minimum-MSE estimator)}.  This suggests that the functional form of the schooling decision is {not} invariant to treatment status, highlighting that predictions based off control villages are less satisfactory for boys (as also found by TW).} In addition, the 95\% confidence intervals are substantially tighter than when using control villages only. When $\epsilon=\overline{\epsilon}$ (shown in the vertical dashed line), the program estimates on school attendance are positive and marginally significant at the 5\% level for girls, and positive and marginally significant at 10\% for boys. In the vertical solid line, we highlight the value $\epsilon_1$ given by (\ref{seq_eps_gen}) for $k=1$. Since $\pi=\lambda$ is scalar, setting $\epsilon\geq \epsilon_1$ ensures that, for all models outside the neighborhood, a 5\%-likelihood ratio specification test has local power larger than 80\%. Taking $\epsilon= \epsilon_1$ implies that the ratio of marginal utility of the subsidy to marginal utility of consumption is bounded by 1.8 (girls) and 1.3 (boys). While $\epsilon_1$ is larger than $\overline{\epsilon}$, the implied minimum-MSE estimators and confidence intervals are similar.\footnote{Note that $\epsilon_1$ is infinite in the \textit{ex-ante} case (top panel of Figure \ref{Fig_prog}). This is due to control villages not providing any information about $\pi$ in this case. }

In Table \ref{Table_TW}, we report estimates of the program impacts, as well as predictions of counterfactual policies. The left two columns correspond to \textit{ex-ante} prediction based on control villages only, while the right two columns correspond to \textit{ex-post} prediction based on both controls and treated. We show the results for $\epsilon=\overline{\epsilon}$, corresponding to equal marginal utilities of subsidy and consumption in absolute value. In the top panel, we focus on the impact of the PROGRESA subsidy on eligible households. We see that PROGRESA has a positive impact on attendance of both boys and girls. The impacts predicted by the reference model are large, approximately 8 percentage points, and are quite close to the results reported in Todd and Wolpin (2006, 2008).  However, the confidence intervals which account for model misspecification (third row, left two columns) are very large for both genders. 

When adding treated villages to the sample (right two columns in Table \ref{Table_TW}), confidence intervals accounting for misspecification are tighter. Moreover, the minimum-MSE point estimates and those based on the reference model differ in this case. For boys, the minimum-MSE estimate is substantially lower than the one based on the reference model (3.6\% versus 7.8\%), while for girls the effects are similar. Interestingly, for boys the minimum-MSE estimates are closer to the experimental differences in means between treated and control villages. 

Lastly, in the middle and bottom panels of Table \ref{Table_TW}, we show estimates of the effects of two counterfactual policies: doubling the PROGRESA subsidy, and removing the conditioning of the income transfer on school attendance. Unlike for the main PROGRESA impacts, there is no experimental counterpart to such counterfactuals. While \textit{ex-ante} predictions are associated with wide confidence intervals, \textit{ex-post} minimum-MSE estimates based on both control and treated villages predict a substantial effect of doubling the subsidy on girls' attendance, and a more moderate effect on boys. By contrast, we find insignificant effects of an unconditional income transfer.

\section{Numerical illustrations: binary choice models\label{Sec_numeric}}

In this section, we apply our approach to cross-sectional and panel data binary choice models, where we allow for misspecification of the distribution of unobservables. {In both applications there is a substantial amount of misspecification, and we use simulations to assess the behavior of the minimum-MSE estimator --- which is theoretically justified under local misspecification --- in these settings of practical relevance.}

\subsection{Cross-sectional binary choice\label{subsec_bin}}

Consider the binary choice model
\begin{equation}Y=\mathbbm{1}\{X'\beta_0+A\geq 0\},\label{mod_bc_cs}\end{equation}
where $A$ follows a distribution $\pi_0$, independent of $X$. We are interested in estimating the prediction function ${\delta}_{\beta_0,\pi_0}=\mathbb{E}_{\pi_0}[\mathbbm{1}\{x_0'\beta_0+A\geq 0\}]$,
at some $x_0$ not necessarily in the support of $X$. We focus on the reference specification $A\sim {\cal{N}}(0,1)$, independent of $X$. We allow for the possibility that this parametric model is misspecified, while maintaining independence between $A$ and $X$ under $\pi_0$. We observe an i.i.d. sample $(Y_i,X_i)$ for $i=1,...,n$.

\begin{figure}[b!]\caption{Distributions of $X$ and $A$ in the binary choice models\label{Fig_CS_DGP}}
	\begin{center}
		\begin{tabular}{ccc}
			(a) Cross-section: $X$ & (b) Cross-section: $A$ & (c) Panel data: $A$ 	\\
			\includegraphics[width=50mm, height=50mm]{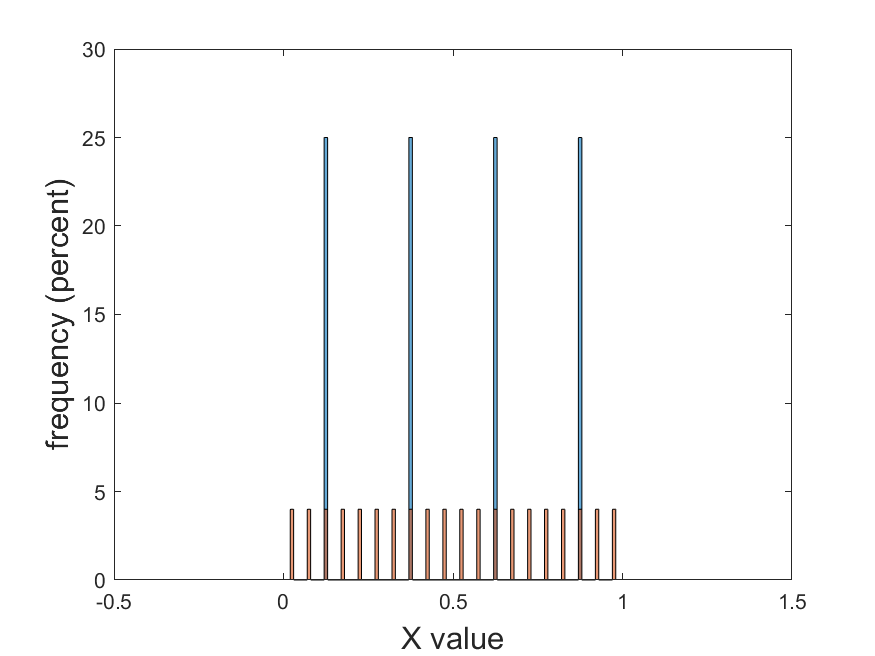}&\includegraphics[width=50mm, height=50mm]{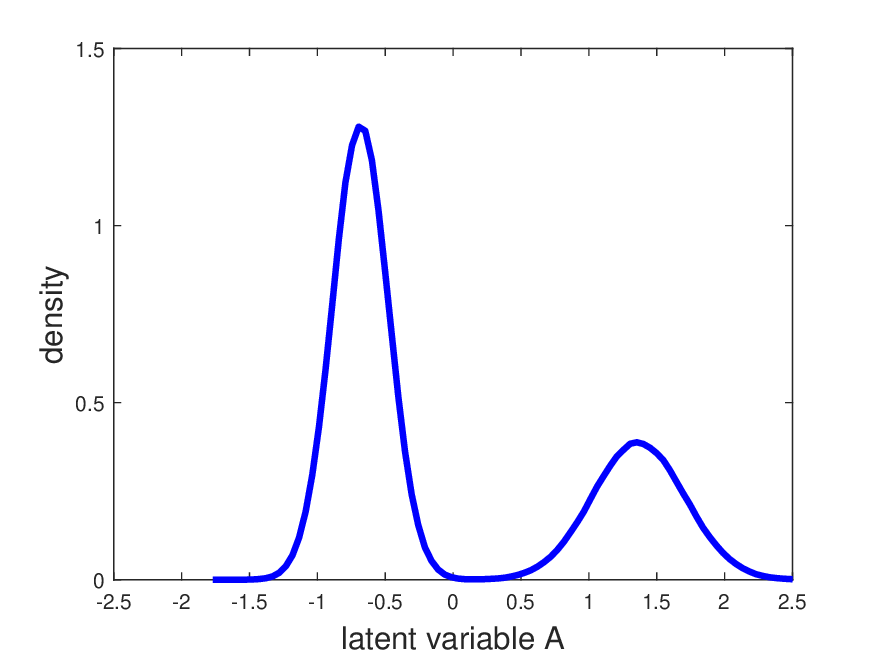}&	\includegraphics[width=50mm, height=50mm]{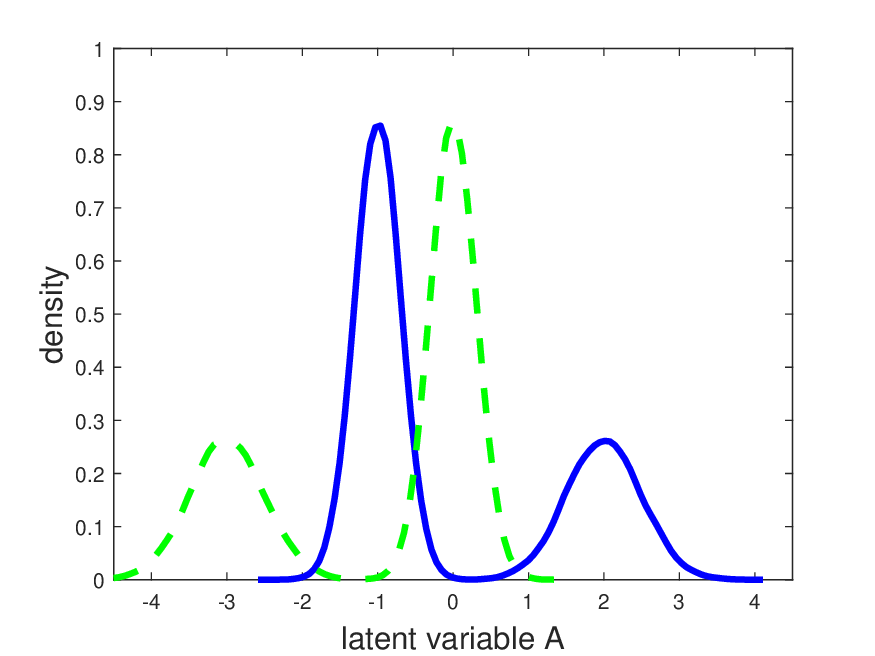}\\
		\end{tabular}
	\end{center}
	\par
	\textit{{\small Notes: In panel (a) we show the frequencies of covariates (i.e., the first component of $X$) in the cross-sectional model, and in (b) we show the true density of $A$ for the same model. In (c) we show the true densities of $A$ in the panel data model for $Y_0=0$ (in solid) and $Y_0=1$ (in dashed).}}
\end{figure}

In neighborhoods that consist of distributions of $A$ independent of $X$, the minimum-MSE influence function is given by Corollary \ref{MMSE_prob}, with $\nabla_{\beta}\delta=x_0\phi(x_0'\beta_0)$ for $\phi$ the standard normal density, $\Delta(a)=\mathbbm{1}\{x_0'\beta_0+a\geq 0\}$, and without $\gamma$ parameter. Given a preliminary estimator $\widehat{\beta}$ (e.g., obtained by probit), an empirical counterpart to $\overline{h}^{\rm MMSE}_{\epsilon}(a)$ is
$$\frac{1}{n}\sum_{i=1}^n \mathbbm{1}\{X_i'\widehat\beta+a\geq 0\}h_{\epsilon}^{\rm MMSE}(1,X_i)+(1-\mathbbm{1}\{X_i'\widehat\beta+a\geq 0\})h_{\epsilon}^{\rm MMSE}(0,X_i).$$
We compute $h_{\epsilon}^{\rm MMSE}(1,X_i)$ and $h_{\epsilon}^{\rm MMSE}(0,X_i)$, for $i=1,...,n$, based on Corollary \ref{MMSE_prob} by solving a linear system. In this model, all conditional expectations are available in closed form.

\begin{figure}[t!]\caption{Bounds on $\mathbb{E}_{\pi_0}[\boldsymbol{1}\{a+A\geq 0\}]$ in $\epsilon$-neighborhoods of the standard normal\label{Fig_CS_eps}}
	\begin{center}
		\begin{tabular}{ccc}
			(a) $\epsilon=0.1$ & (b) $\epsilon=1$\\
			\includegraphics[width=50mm, height=50mm]{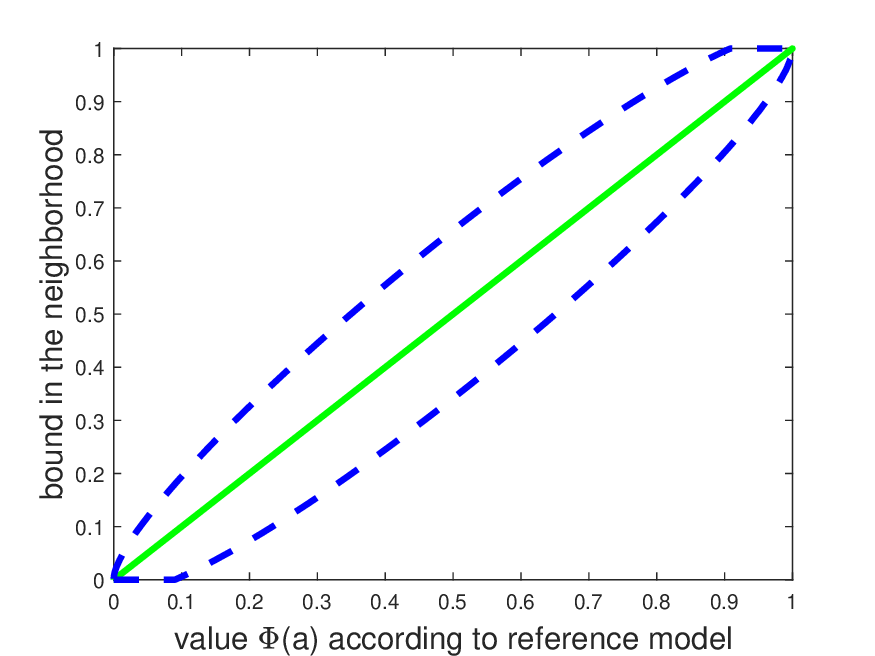}&\includegraphics[width=50mm, height=50mm]{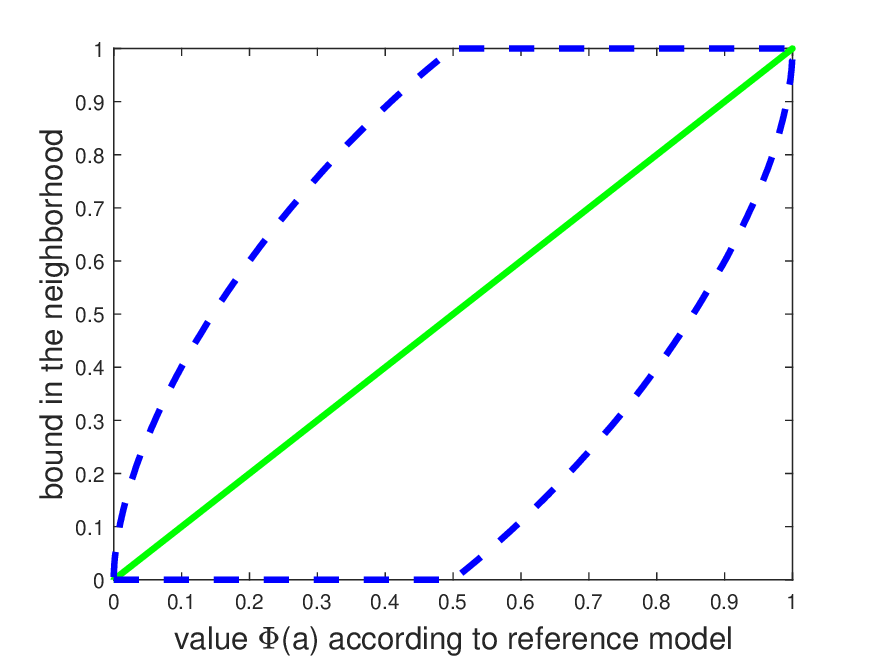}\\
		\end{tabular}
	\end{center}
	\par
	\textit{{\small Notes: The dashed lines show ${\limfunc{sup}}_{\pi_0\in\Gamma_{\epsilon}(\gamma_*)}\mathbb{E}_{\pi_0}[\boldsymbol{1}\{a+A\geq 0\}]$ and ${\limfunc{inf}}_{\pi_0\in\Gamma_{\epsilon}(\gamma_*)}\mathbb{E}_{\pi_0}[\boldsymbol{1}\{a+A\geq 0\}]$ as a function of $\Phi(a)$. The 45-degree line is in solid. Computations are based on small-$\epsilon$ approximations.}}
\end{figure}

In model (\ref{mod_bc_cs}), under independence between $A$ and $X$, $\beta_0$ and $\pi_0$ are point-identified up to scale under sufficiently rich support of $X$ (Manski, 1988). Under such conditions, ${\delta}_{\beta_0,\pi_0}$ is identified. More generally, it is partially identified. We now set up a simulation where the support of $X$ is discrete, and we vary the number of support points and the target $x_0$. In this way, we learn how our estimators and confidence intervals perform in settings where the support of $X$, and hence the size of the identified set, vary. 

We will show estimates in data generating processes (DGPs) with a scalar covariate and an intercept, and $\beta_0=(2,-1)'$, where the second element corresponds to the intercept. We draw $1000$ simulated samples of size $n=500$, where $A$ has mean zero and variance one, and is distributed as a mixture of two normals whose means are approximately two standard deviations apart. Covariates are discrete uniform on $[0,1]$, with either $n_X=4$ or $n_X=20$ points of support. We show the densities of $X$ and $A$ in panels (a) and (b) of Figure \ref{Fig_CS_DGP}. We focus on the predicted values at $x_0=(0.5,1)'$ and $x_0=(-0.5,1)'$, respectively. We refer to the first case as \textit{interpolation}, and to the second one as \textit{extrapolation}.

	We report minimum-MSE estimates and confidence intervals on a range of $\epsilon$ values. To provide intuition about orders of magnitude, one can compute (local approximations to twice) the KL divergence between the standard normal and other common distributions. For example, a scaled student-$t$ distribution with unitary variance and $5$, $3$, or $2.1$ degrees of freedom, respectively, corresponds to a distance of $0.14$, $0.24$, and $1.5$; the true bimodal $\pi_0$ 
	in the DGP corresponds to a distance of $1.6$; and the standardized logistic corresponds to a distance of $0.07$. Restricting $\pi_0$ to belong to an $\epsilon$-neighborhood of the normal also has implications for its functionals. As an example, in Figure \ref{Fig_CS_eps}, we show pointwise bounds on $\mathbb{E}_{\pi_0}[\boldsymbol{1}\{a+A\geq 0\}]$, as a function of $\Phi(a)$, computed using small-$\epsilon$ approximations. We see that taking $\epsilon=0.1$ tightly restricts possible values that the parameter can take. By contrast, when $\epsilon=1$, the \textit{a priori} bounds on the parameter are very wide for $a$ close to $0.5$ --- which is relevant for the \textit{interpolation} case --- but the neighborhood does restrict the parameter value when $a$ in close to zero or one --- which is relevant for the \textit{extrapolation} case. Given this, we will interpret $\epsilon$ values of the order of 0.1 or lower as reflecting ``mild'' misspecification, and values of the order of 1 or larger as corresponding to ``large'' misspecification.

	\begin{figure}[t!]	\caption{Minimum-MSE estimator in the cross-sectional binary choice model\label{Fig_CS}}
		\begin{center}
			\begin{tabular}{cc}
				\multicolumn{2}{c}{A. \textit{Interpolation}: $x_0=(0.5,1)'$}\\
				(a) $n_X=4$ & (b) $n_X=20$ 	\\
				\includegraphics[width=70mm, height=42mm]{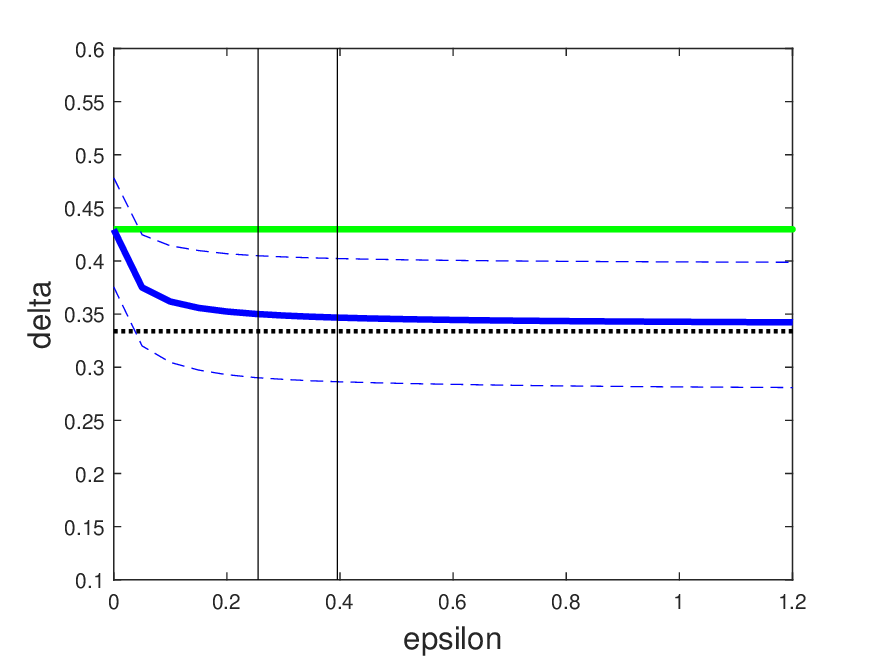}&	\includegraphics[width=70mm, height=42mm]{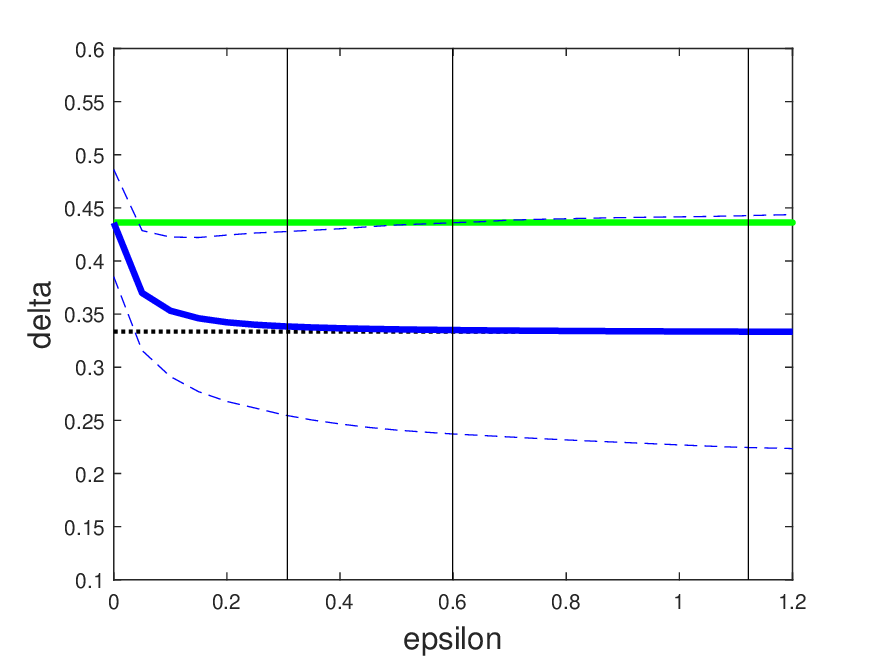}\\
				\multicolumn{2}{c}{B. \textit{Extrapolation}: $x_0=(-0.5,1)'$}\\
				(c) $n_X=4$  & (d) $n_X=20$ 	\\
				\includegraphics[width=70mm, height=42mm]{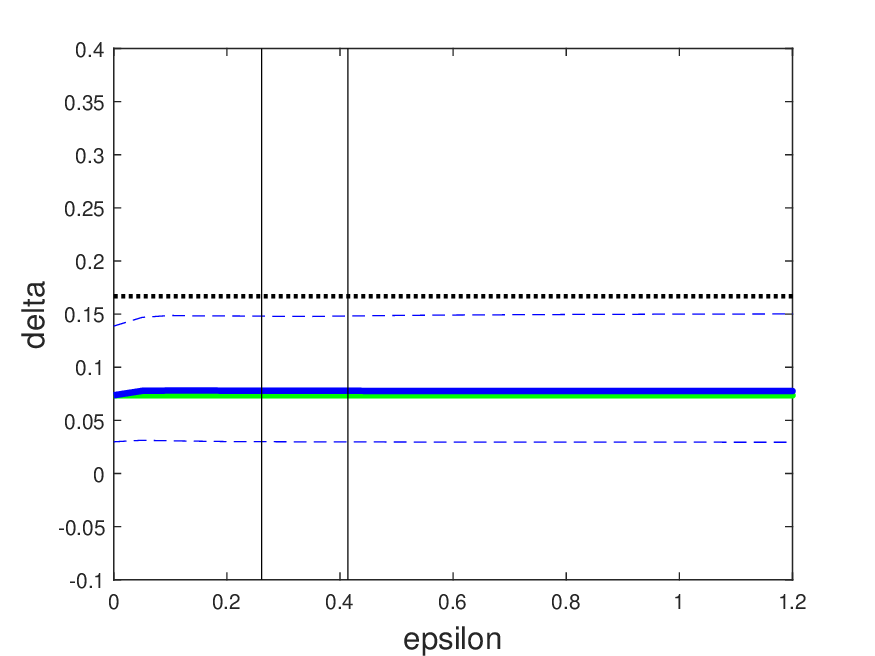}&	\includegraphics[width=70mm, height=42mm]{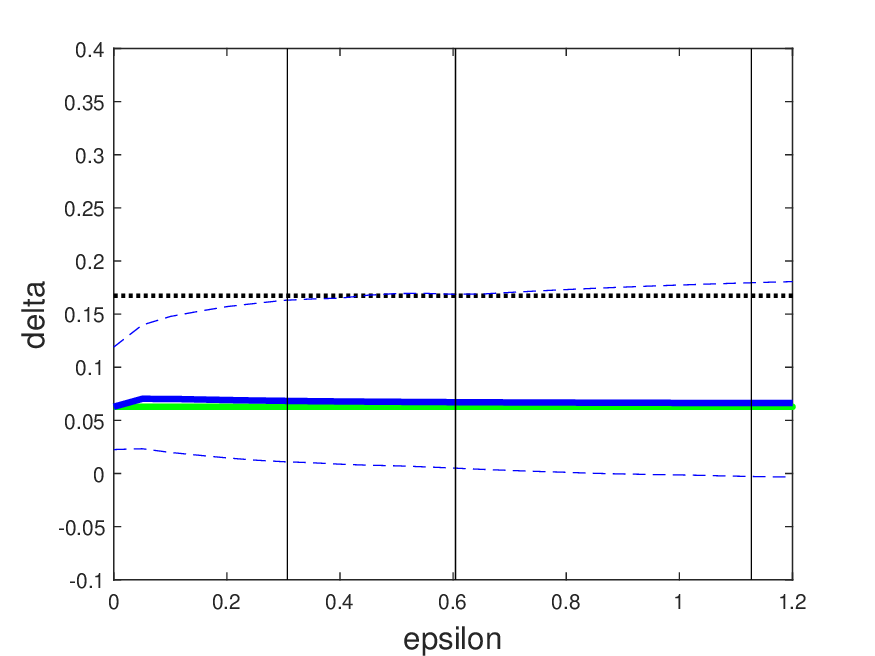}\\			\end{tabular}
		\end{center}
		\par
		\textit{{\small Notes: The solid horizontal line corresponds to the mean probit estimator among 1000 simulations, the solid curve to the mean minimum-MSE estimator (with 2.5\% and 97.5\% percentiles in dashed), and the dotted horizontal line to the truth {(i.e., ${\delta}_{\beta_0,\pi_0}$, for the $\beta_0$ and $\pi_0$ we used in the DGP)}. $\epsilon$ is reported on the x-axis, and the vertical lines indicate $\epsilon_k$ given by (\ref{seq_eps_gen}), for $k\in\{1,2,3\}$. $n_X$ denotes the number of points of support of the first component of $X$. $n=500$.}}
	\end{figure}

In this example, it is also informative to interpret $\epsilon$ by relating it to the power of a specification test, as we described in Subsection \ref{sec42}. When $X$ has 4 support points, $\widetilde{H}_\pi$ has only $n_X-1=3$ non-zero eigenvalues (since the $X'\beta$ partition the real line into $n_X+1$ intervals, and the two elements in $\beta$ are estimated), two of them corresponding to non-constant eigenfunctions. In this case, $\epsilon_1$ and $\epsilon_2$ given by (\ref{seq_eps_gen}) are approximately $0.2$ and $0.4$ on average, where we set size to $\alpha=5\%$ and power to $p=80\%$. When $X$ has 20 support points, $\widetilde{H}_{\pi}$ has $n_X-1=19$ non-zero eigenvalues corresponding to non-constant eigenfunctions. The first three values $\epsilon_1$, $\epsilon_2$, and $\epsilon_3$ are approximately $0.3$, $0.6$, and $1.1$ on average. In contrast with the parametric case of Section \ref{App_TW}, here setting $\epsilon\geq\epsilon_k$ only provides power guarantees along particular directions. In Appendix \ref{App_Fig}, we plot those directions, and we provide additional intuition about the interpretation of $\epsilon$ based on statistical testing in this example.

We show the results of the simulation in Figure \ref{Fig_CS}. Consider first the top panel, where we wish to interpolate the prediction function at $x_0=0.5$. When $X$ has 4 support points, we see that the probit estimator based on the reference model, indicated by the solid horizontal line, is substantially biased. By contrast, the bias of the minimum-MSE estimator is smaller, and it decreases as $\epsilon$ increases. We see that the minimum-MSE estimator is close to unbiased for both $\epsilon_1$ and $\epsilon_2$. Moreover, the dispersion of the estimator is stable as $\epsilon$ increases. In addition, we compute the identified set for $\delta_{\beta_0,\pi_0}$ in the DGP using linear programming and a grid of $\beta_0$ values. We find $[0.334,0.345]$, which shows that the identified set is not wide in this DGP. 

The case where $X$ has 20 support points is overall quite similar, but with several differences. We see that the minimum-MSE estimator is virtually unbiased when $\epsilon\geq \epsilon_1$. In this case, the identified set for $\delta_{\beta_0,\pi_0}$ is essentially a singleton: 
$[0.334,0.335]$. Moreover, we see that the variance of the minimum-MSE estimator increases with $\epsilon$. Such a variance increase, and the associated regularization role of $\epsilon$, also characterize models with continuously distributed covariates and other ill-posed inverse problems.

Lastly, consider the lower panel in Figure \ref{Fig_CS}. This \textit{extrapolation} case is very different from the \textit{interpolation} one. Indeed, the data provide little information about the value of the prediction function at $x_0=-0.5$. To illustrate, the identified set for $\delta_{\beta_0,\pi_0}$ is $[0,0.3219]$ (respectively, $[0,0.2956]$) when $X$ has 4 (resp., 20) points of support. We see that the minimum-MSE estimator has approximately the same bias as the probit estimator in this case. 

We show additional information about the simulation results in Tables \ref{TabProbit_CS_1} and \ref{TabProbit_CS_2} in the appendix. In particular, we report the lengths of our 95\% confidence intervals (CI) for $\delta_{\beta_0,\pi_0}$, which are asymptotically valid under $\epsilon$-misspecification, and the associated coverage probabilities. In all DGPs, we find that, when taking $\epsilon\geq \epsilon_1$, the confidence intervals contain the true value with a probability that exceeds 95\%.\footnote{While this finding is interesting, note that our CI construction has coverage guarantees only when $\pi_0$ belongs to an $\epsilon$-neighborhood of $\pi(\gamma_*)$, which is not the case here since the true distribution of $A$ lies outside the neighborhoods for the range of $\epsilon$ that we consider.}

\subsection{Dynamic panel data binary choice\label{subsec_panel}}

In this subsection, we present simulations in the following dynamic panel data probit model with individual effects
\begin{equation}\label{mod_bc_panel}
Y_{t}=\mathbbm{1}\left\{\beta_0 Y_{t-1}+A+U_{t}\geq 0\right\},\quad t=1,...,T,
\end{equation}
where $U_{1},...,U_{T}$ are i.i.d. standard normal, independent of $A$ and $Y_{0}$. Here $Y_{0}$ is observed, so there are effectively $T+1$ time periods. We focus on the average {state dependence} effect
$\delta_{\beta_0,\pi_0}=\mathbb{E}_{\pi_0}\left[\Phi(\beta_0+A)-\Phi(A)\right]$, and we will also report estimates of the autoregressive parameter $\beta_0$. We assume that the probit conditional likelihood given individual effects and lagged outcomes is correctly specified. However, we do not assume knowledge of $\pi_0$ or its functional form. We specify a normal reference density for $A$ given $Y_{0}$, with mean $\mu_1+\mu_2 Y_{0}$ and variance $\sigma^2$; hence here $\gamma=(\mu_1,\mu_2,\sigma^2)'$. Binary choice panel data models are often partially identified for fixed $T$ (Chamberlain, 2010, Honor\'e and Tamer, 2006), and no semi-parametrically consistent estimators of $\beta_0$ and $\delta_{\beta_0,\pi_0}$ in the dynamic probit model are available in the literature. Here we report simulation results suggesting that minimum-MSE estimators can perform well under sizable misspecification of the reference density. 

\begin{figure}[tb!]	\caption{Minimum-MSE estimator in the dynamic panel binary choice model\label{Fig_probit_mispec_pd}}
	\begin{center}
		\begin{tabular}{ccc}
			\multicolumn{3}{c}{A. Average state dependence $\delta_{\beta_0,\pi_0}$}\\
			(d) $T=5$ & (e) $T=10$ & (f) $T=20$  	\\
			\includegraphics[width=50mm, height=40mm]{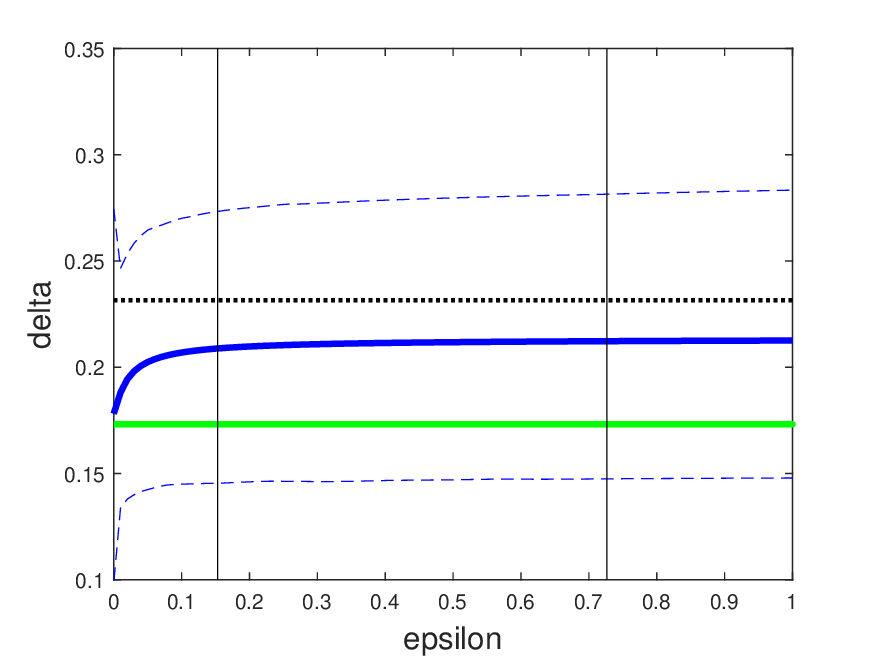}&	\includegraphics[width=50mm, height=40mm]{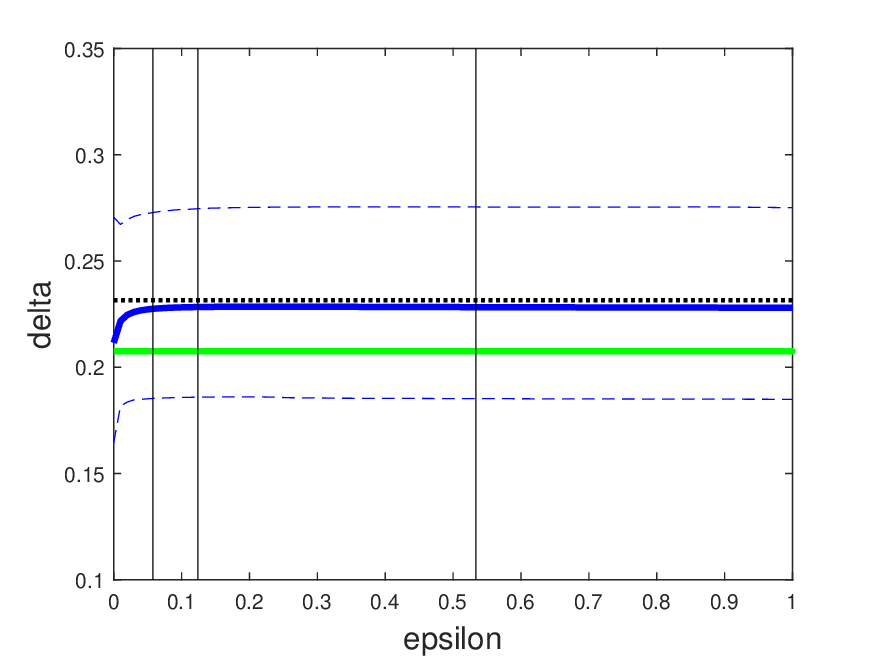}&\includegraphics[width=50mm, height=40mm]{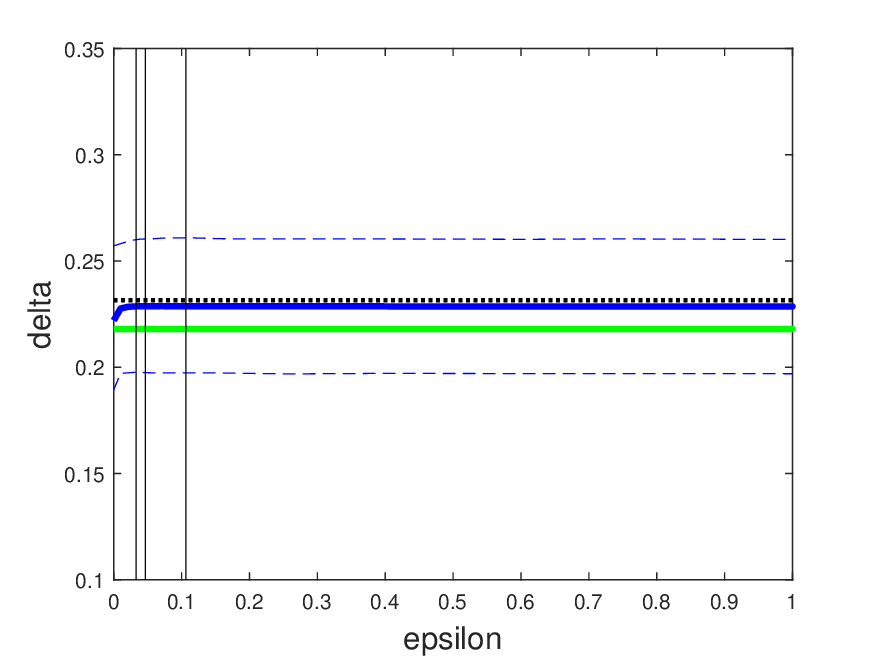}\\\multicolumn{3}{c}{B. Autoregressive parameter $\beta_0$}\\
			(a) $T=5$ & (b) $T=10$ & (c) $T=20$  	\\
			\includegraphics[width=50mm, height=40mm]{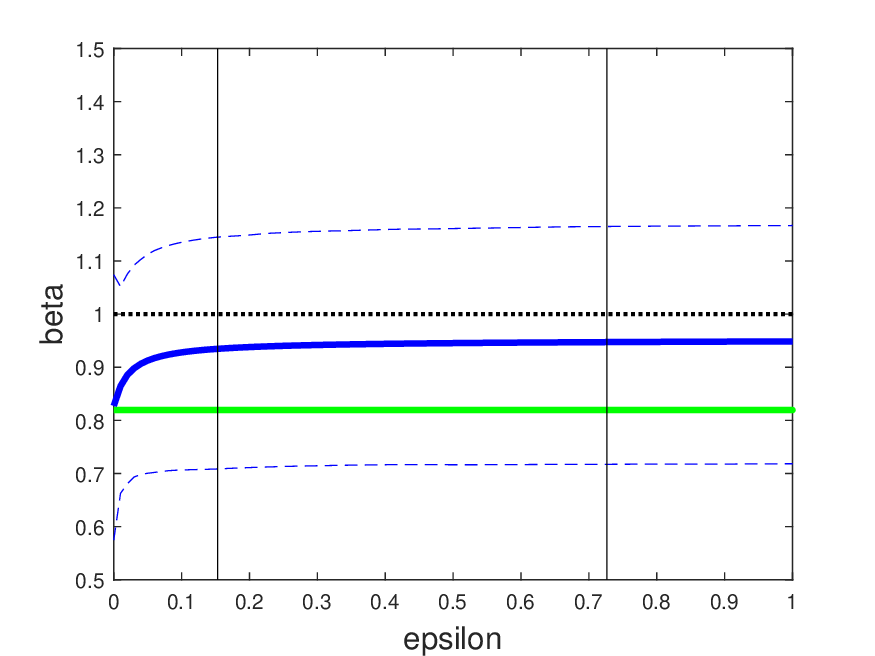}&	\includegraphics[width=50mm, height=40mm]{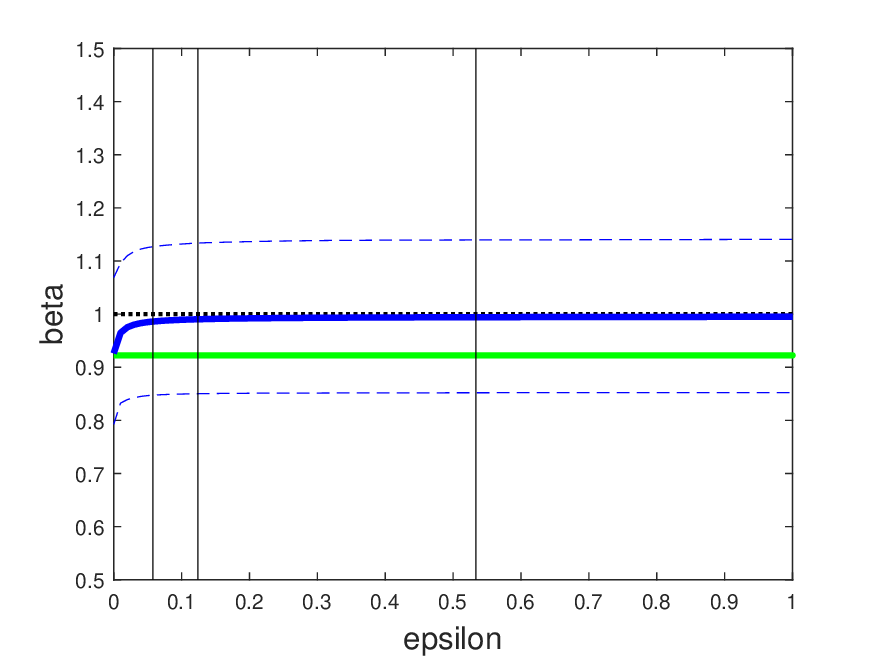}&\includegraphics[width=50mm, height=40mm]{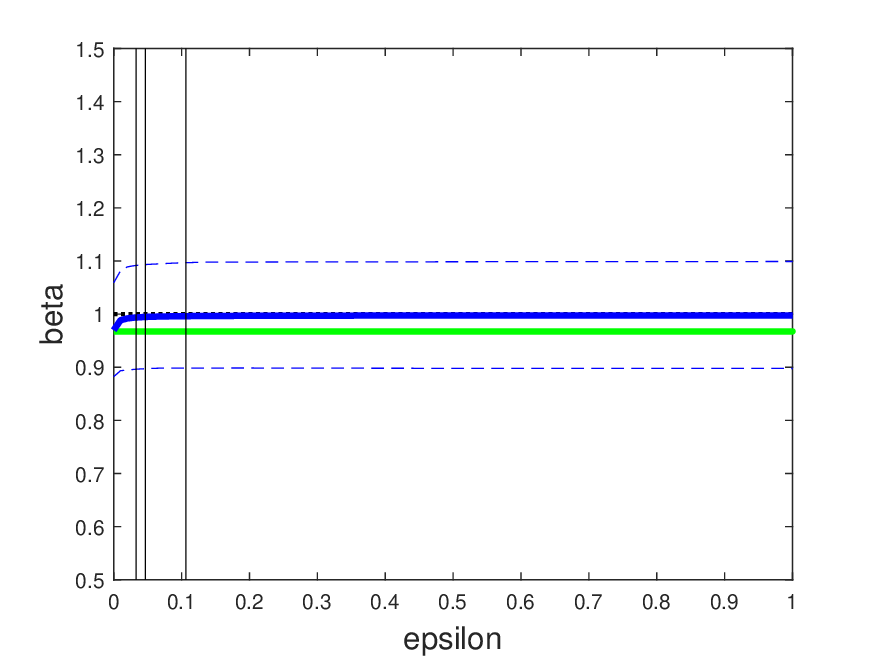}\\
		\end{tabular}
	\end{center}
	\par
	\textit{{\small Notes: The solid horizontal line corresponds to the mean random-effects estimator among 1000 simulations, the solid curve to the mean minimum-MSE estimator (with 2.5\% and 97.5\% percentiles in dashed), and the dotted horizontal line to the truth. $\epsilon$ is reported on the x-axis, and the vertical lines indicate $\epsilon_k$,  $k\in\{1,2\}$ (left column) and $k\in\{1,2,3\}$ (middle and right columns). In the left column, $\epsilon_3$ is too large to be included in the figure. $n=500$.}}
\end{figure}

In the simulation, we set a bimodal distribution that has modes $\{-1,2\}$ when $Y_0=0$ and $\{-3,0\}$ when $Y_0=1$, with some asymmetry between the two modes; see panel (c) of Figure \ref{Fig_CS_DGP}. {We draw $Y_0$ from a Bernoulli distribution with probability $0.5$.} We take $n=500$, and show the results for $T=5$, $10$, and $20$, based on $1000$ simulations. In neighborhoods that consist of unrestricted joint distributions $\pi_0$ of $(A,X)$, the minimum-MSE $h$ function is given by Corollary \ref{MMSE_prob2}, for $X=Y_{0}$, and either $\Delta(a)=\Phi(\beta_0+a)-\Phi(a)$ or $\Delta(a)=\beta_0$, depending on the quantity of interest. We use $S=1000$ simulated draws to compute the minimum-MSE estimators, since no closed-form solution is available in this case {(see Appendix \ref{App_compute})}.

\begin{table}[tbp]
	\caption{Monte Carlo simulation in the dynamic binary choice panel data model: comparison between various estimators\label{TabProbit}}
	\begin{center}
		\begin{tabular}{l||ccc|cccc}
			$T=$  & 5 & 10 & 20  & 5 & 10 & 20 \\\hline\hline 
			&\multicolumn{3}{c}{Bias} & \multicolumn{3}{c}{Root MSE}\\\hline
			&\multicolumn{6}{c}{A. Average state dependence $\delta_{\beta_0,\pi_0}$}\\\hline
			Random-effects  &  -0.0585  & -0.0252 &  -0.0140
			& 0.0633 &    0.0311  &  0.0198\\
			Empirical Bayes & -0.0574  & -0.0215 &  -0.0053 &  0.0622 &   0.0282 &    0.0141\\
			Linear probability  &  -0.2491 &  -0.0976  &  0.0012 &  0.2497  &  0.0990 &   0.0128\\
			Minimum-MSE ($\epsilon_1$) &   -0.0227 &  -0.0057   &-0.0029& 0.0397   & 0.0232 &   0.0154
			\\
			Minimum-MSE ($\epsilon_2$) &   -0.0194 &  -0.0048  & -0.0028& 0.0388  &  0.0233   & 0.0154
			\\
			Minimum-MSE ($\epsilon_3$) & -0.0196 &  -0.0049 &  -0.0027 & 0.0412  &  0.0235 &   0.0155
			\\\hline
			&\multicolumn{6}{c}{B. Autoregressive parameter $\beta_0$}\\\hline
			Maximum likelihood  & -0.1804 &  -0.0817 &  -0.0328
			&  0.2003 &   0.1001  &  0.0506 \\
			Minimum-MSE ($\epsilon_1$)   &  -0.0646 &  -0.0198  & -0.0055
			& 0.1288   & 0.0747 &   0.0479\\
			Minimum-MSE ($\epsilon_2$)   &  -0.0522  & -0.0155  & -0.0045 & 0.1258  &  0.0746 &   0.0481\\
			Minimum-MSE ($\epsilon_3$)   & -0.0432 &  -0.0116 &  -0.0030 &  0.1282  &  0.0747  &  0.0486\\\hline
		\end{tabular}
	\end{center}
	\par
	\textit{{\small Notes: Performance of various estimators in the dynamic panel data binary choice model, for different values of $T$. $n=500$, results for $1000$ simulations. }}
\end{table}

In Figure \ref{Fig_probit_mispec_pd}, we see that the parametric random-effects dynamic probit estimates of $\delta_{\beta_0,\pi_0}$ and  $\beta_0$ are substantially biased for $T=5$ and $T=10$, whereas the bias is smaller when $T=20$. By contrast, the minimum-MSE estimator performs better in terms of bias for both quantities of interest, in particular when taking $\epsilon$ to be one of the first few $\epsilon_k$'s given by (\ref{seq_eps_gen}). In the top panel of Table \ref{TabProbit}, we show the bias and root MSE of various estimators of $\delta_{\beta_0,\pi_0}$: the random-effects estimator based on the normal reference model, an empirical Bayes estimator, the linear probability estimator, and the minimum-MSE estimators based on $\epsilon_1,\epsilon_2,\epsilon_3$.\footnote{The random-effects and empirical Bayes estimators are given by $\frac{1}{n}\sum_{i=1}^n\mathbb{E}_{\pi(\widehat{\gamma})}[\Phi(\widehat{\beta}+A)-\Phi(A)]$ and $\frac{1}{n}\sum_{i=1}^n\mathbb{E}_{\pi(\widehat{\gamma})}[\Phi(\widehat{\beta}+A)-\Phi(A)\,|\, Y=Y_i]$, respectively. In fixed-lengths panels both estimators are consistent under the parametric reference specification, and the random-effects estimator is efficient. However, the two estimators are generally biased under misspecification. { Bonhomme and Weidner (2021) show that the empirical Bayes estimator has minimum local asymptotic worst-case specification error, albeit in neighborhoods of the reference model where the probit conditional likelihood given individual effects and lagged outcomes may be incorrectly specified.}
} We see that the minimum-MSE estimator dominates all other estimators, for these $\epsilon_k$ values, when $T=5$ and $T=10$. In the bottom panel  of Table \ref{TabProbit}, we show the results for the random-effects MLE and minimum-MSE estimators of $\beta_0$. The results are similar to the case of average state dependence. In this DGP, minimum-MSE estimators achieve bias reduction under misspecification, even when $T$ is quite small. Bias reduction comes with some increase in variance, yet, the overall MSE is lower for minimum-MSE estimators compared to the MLE. Lastly, in Tables \ref{TabProbit_PD_1} and \ref{TabProbit_PD_2} in the appendix, we show additional information about the simulation results, for the autoregressive parameter and the average state dependence parameter, respectively.

\section{Conclusion\label{Sec_conclu}}

We propose a framework for estimation and inference in the presence of model misspecification. This allows researchers to perform sensitivity analysis for existing estimators, and to construct improved estimators and confidence intervals that are less sensitive to model assumptions. Our approach is based on a minimax mean squared error rule, which consists of a one-step adjustment of the initial estimate. This adjustment is motivated by both robustness and efficiency, and it remains valid when the identification of the ``large'' model is irregular or point-identification fails. Hence, our approach provides a {complement to partial identification methods}, when the researcher sees her reference model as a plausible, albeit imperfect, approximation to reality. Given a parametric reference model, implementing our estimators and confidence intervals {does not} require estimating a {larger model}. This is an attractive feature in complex models such as dynamic structural models, for which sensitivity analysis methods are needed. {Lastly, while our theory applies quite generally, we have provided explicit expressions and described implementation in two specific classes of problems: parametric models, and semi-parametric likelihood models with a mixture structure. Generalizing the applicability of the approach to other semi-parametric models is an important task for future work.}

  \clearpage

\appendix

\renewcommand{\theequation}{\thesection \arabic{equation}}

\renewcommand{\thelemma}{\thesection \arabic{lemma}}

\renewcommand{\theproposition}{\thesection \arabic{proposition}}

\renewcommand{\thecorollary}{\thesection \arabic{corollary}}

\renewcommand{\thetheorem}{\thesection \arabic{theorem}}

\renewcommand{\theassumption}{\thesection \arabic{assumption}}

\renewcommand{\thefigure}{\thesection \arabic{figure}}

\renewcommand{\thetable}{\thesection \arabic{table}}

\setcounter{equation}{0}
\setcounter{table}{0}
\setcounter{figure}{0}
\setcounter{assumption}{0}
\setcounter{proposition}{0}
\setcounter{lemma}{0}
\setcounter{corollary}{0}
\setcounter{theorem}{0}

\baselineskip16pt

\begin{center}
	{\small {\LARGE APPENDIX} }
\end{center}

 \section{Main results\label{App_general}}

 In this section of the appendix we provide the proofs for the main results of Section \ref{Sec_outline}. As in the rest of the paper, we always implicitly assume that all functions of $y$
 are measurable, and that all expectations and integrals over $y$
 are well-defined.
 
 \subsection{Proof of Theorem \ref{theo1}}
 \setcounter{lemma}{0}
 
 \subsubsection{Notation and assumptions}
 \label{app:MoreNotation}

 In all our applications $\Pi$ is either a vector space or an affine space.
 Let $\overline {\cal T}$ and ${\cal T}$ be the tangent 
 and cotangent spaces of $\Pi$ at $\pi(\gamma_*)$. Thus, for $\pi_1,\pi_2 \in \Pi$
 we have $(\pi_1 - \pi_2) \in \overline {\cal T}$, and ${\cal T}$ is the set of linear maps 
 $u:\overline {\cal T} \rightarrow \mathbb{R}$.
 For a scalar function $q : \Pi \mapsto \mathbb{R}$,
 we have $\nabla_\pi q_{\pi(\gamma_*)} \in {\cal T}$; that is, the typical element of ${\cal T}$
 is a gradient. Conversely, for a map to $\Pi$, such as $\gamma \mapsto \pi(\gamma)$, we have
 $ \frac{\partial \pi(\gamma_*)} {\partial \gamma_k} \in \overline {\cal T}$.
 
 For $v \in \overline {\cal T}$  and $u \in {\cal T}$  we use the bracket notation
 $\langle v,u \rangle \in \mathbb{R}$ to denote the bilinear mapping. Here we are not assuming a Hilbert space structure, and we only use the bracket notation to combine vectors $v$ and covectors $u$ into a scalar.

 Our squared distance measure $d(\pi_0,\pi(\gamma_*))$ on $\Pi$ induces
 a norm on the tangent space $\overline {\cal T}$, namely for $v \in \overline {\cal T}$,
 \begin{align*}
 \| v \|^2_{{\rm ind},\gamma_*} = \lim_{\epsilon \rightarrow 0} \frac{ d\big( \pi(\gamma_*) + \epsilon^{1/2} v, \pi(\gamma_*) \big) } {\epsilon } .
 \end{align*}

 Throughout we assume that $\dim \beta$ and $\dim \gamma$ are finite. For any finite-dimensional vectors 
 we use the standard Euclidean norm $\| \cdot \|$, and for any finite-dimensional matrices
 we use the spectral matrix norm, which we also denote by $\| \cdot \|$. 
Let ${\cal Y}$ denote the support of $Y$.

 \begin{assumption}
 	\label{ass:Expansion}
 	We assume that $Y_i \sim i.i.d. f_{\beta_0,\pi_0}$.
 	In addition, we impose the following regularity conditions:
 	\begin{itemize}
 		\item[(i)] We consider $n \rightarrow \infty$ and $\epsilon \rightarrow 0$
 		such that $\epsilon  n \rightarrow c$, for some constant $c \in (0,\infty)$.
 		
 		\item[(ii)]
 		$   \sup_{\pi_0 \in \Gamma_{\epsilon}(\gamma_*)} 
 		\left\|  \nabla_{\pi} \delta_{\beta_0,\pi_0}  \right\|_{\gamma_*} = O(1)$, 
 		and \\
 		$\sup_{\pi_0 \in \Gamma_{\epsilon}(\gamma_*)} 
 		\left|  \delta_{\beta_0,\pi_0} - \delta_{\beta_0,\pi(\gamma_*)}
 		-  \left\langle \pi_0 - \pi(\gamma_*) , \nabla_\pi \delta_{\beta_0,\pi(\gamma_*)} \right\rangle  \right| = o(\epsilon^{1/2})$.

 		\item[(iii)] 
 		$ \sup_{\pi_0 \in \Gamma_{\epsilon}(\gamma_*)}   
 		\displaystyle
 		\left\{   \int_{\cal{Y}}    \left[  f^{1/2}_{\beta_0,\pi_0}(y) - f^{1/2}_{\beta_0,\pi(\gamma_*)}(y)  \right]^{2}  dy 
 		\right\}^{1/2} = O(\epsilon^{1/2})$,
 		\\  
 		$ \sup_{\pi_0 \in \Gamma_{\epsilon}(\gamma_*)}    
 		\displaystyle  \int_{\cal{Y}}   \left\| \nabla_\pi \log f_{\beta_0,\pi(\gamma_*)}(y)   \right\|^2_{\gamma_*}   \left[  f^{1/2}_{\beta_0,\pi_0}(y) - f^{1/2}_{\beta_0,\pi(\gamma_*)}(y)
 		\right]^2   dy = o(1)$,
 		\\
 		$  \sup_{\pi_0 \in \Gamma_{\epsilon}(\gamma_*)}    
 		\displaystyle  \int_{\cal{Y}}  \left[  f^{1/2}_{\beta_0,\pi_0}(y) - f^{1/2}_{\beta_0,\pi(\gamma_*)}(y)
 		-   \left\langle  \pi_0 - \pi(\gamma_*) , \nabla_\pi   f^{1/2}_{\beta_0,\pi(\gamma_*)}(y)  \right\rangle    \right]^2   dy = o(\epsilon)$.

 		\item[(iv)] 
 		$\sup_{\pi_0 \in \Gamma_{\epsilon}(\gamma_*)}  
 		\epsilon^{-1/2} \left\| \pi_0 - \pi(\gamma_*) \right\|_{{\rm ind},\gamma_*}  
 		= 1+ o(1)$. 
 		Furthermore, for $u_\epsilon  \in {\cal T}$ with 
 		$\left\| u_\epsilon    \right\|_{\gamma_*} = O(1)$ we have
 		$$\left|   \sup_{\pi_0 \in \Gamma_{\epsilon}(\gamma_*)} 
 		\epsilon^{-1/2}    \left\langle  \pi_0 - \pi(\gamma_*) ,  u_\epsilon \right\rangle
 		-
 		\left\| u_\epsilon 
 		\right\|_{\gamma_*}
 		\right|
 		=   o(1) .$$

 		\item[(v)]  
 		For some $\nu>0$ we have
 		$  \sup_{\pi_0 \in \Gamma_{\epsilon}(\gamma_*)}    \mathbb{E}_{\beta_0,\pi_0}    \left\|   \nabla_{\beta\gamma} \log f_{\beta_0,\pi(\gamma_*)}(Y)  	 \right\|^{2+\nu} = O(1)$, \\
 		and $  \sup_{\pi_0 \in \Gamma_{\epsilon}(\gamma_*)}    \mathbb{E}_{\beta_0,\pi_0}    \left\|     \nabla_{\pi} \log f_{\beta_0,\pi(\gamma_*)}(Y)  	 \right\|_{\gamma_*}^{2+\nu} = O(1)$. \\ Furthermore we assume that
 		$  \left\| \nabla_{\beta\gamma}  \delta_{\beta_0,\pi(\gamma_*)} \right\| = O(1)$, and
 		$ \left\|   H_{\beta\gamma} ^{-1}  \right\| = O(1)$.
 	\end{itemize}
 	
 \end{assumption}
 
 Part (i) of Assumption~\ref{ass:Expansion} describes our asymptotic framework, 
 where the assumption $\epsilon  n \rightarrow c$ is required to ensure that the squared worst-case bias (of order $\epsilon$)
 and the variance (of order $1/n$) of the estimators for $\delta_{\beta_0,\pi_0}$ are asymptotically of the same order, so that the MSE
 provides a meaningful balance between bias and variance asymptotically. Part (ii) requires $\delta_{\beta_0,\pi_0}$ to be 
 sufficiently smooth in $\pi_0$, so that a first-order Taylor expansion provides a good local approximation to $\delta_{\beta_0,\pi_0}$.
 
Part (iii) of Assumption~\ref{ass:Expansion} is a smoothness assumption on $f_{\beta_0,\pi_0}(y)$ in $\pi_0$. Those conditions may not 
look  intuitive, in particular when $\pi_0$ is infinite-dimensional, so we want to discuss that assumption in some more detail here for 
the case of the semi-parametric mixture models introduced in Section~\ref{subsec_semipar},
where$ f_{\beta_0,\pi_0}(y) = \int_{\cal A} \, g_{\beta_0}(y|a) \, \pi_0(a) \, da$.
In that case we have
\begin{align*} 
        \int_{\cal{Y}}    \left[  f^{1/2}_{\beta_0,\pi_0}(y) - f^{1/2}_{\beta_0,\pi(\gamma_*)}(y)  \right]^{2}  dy 
        &= 2 \, H^2(f_{\beta_0,\pi_0} ,  f_{\beta_0,\pi(\gamma_*)})
        \\
    &\leq 2 \, D_{\rm KL}(    f_{\beta_0,\pi_0} || f_{\beta_0,\pi(\gamma_*)})
      \leq 2 \, D_{\rm KL}(    \pi_0  || \pi(\gamma_*) ) ,  
\end{align*}
where the first inequality is the general relation 
$H^2(f_{\beta_0,\pi_0} ,  f_{\beta_0,\pi(\gamma_*)}) \leq D_{\rm KL}(    f_{\beta_0,\pi_0} || f_{\beta_0,\pi(\gamma_*)})$ between the squared Hellinger distance
$H^2$ and the Kullback-Leibler divergence $ D_{\rm KL}$, and the second inequality is sometimes called the 
``chain rule'' for the Kullback-Leibler divergence, which can be derived by an application of Jensen's inequality. 
Since we defined our distance measure $d(\pi_0,\pi(\gamma_*))$ in the semi-parametric mixture case to be twice the 
Kullback-Leibler divergence $2D_{\rm KL}(    \pi_0  || \pi(\gamma_* )) =2 \, \mathbb{E}_{\pi_0}  \log[\pi_0(A) / \pi(A\,|\,\gamma_*)]$
we find that
\begin{align*}
      \sup_{\pi_0 \in \Gamma_{\epsilon}(\gamma_*)}   
 		\displaystyle
 		\left\{   \int_{\cal{Y}}    \left[  f^{1/2}_{\beta_0,\pi_0}(y) - f^{1/2}_{\beta_0,\pi(\gamma_*)}(y)  \right]^{2}  dy 
 		\right\}^{1/2} 
    &\leq   \sup_{\pi_0 \in \Gamma_{\epsilon}(\gamma_*)}  \left\{   d(\pi_0,\pi(\gamma_*)) \right\}^{1/2} = \epsilon^{1/2} .
\end{align*}
Thus, the first condition in Assumption~\ref{ass:Expansion}(iii) is satisfied for those semi-parametric mixture models.

 The second condition in Assumption~\ref{ass:Expansion}(iii)  can be justified by 
imposing that $$\sup_{y \in {\cal Y}}   \left\| \nabla_\pi \log f_{\beta_0,\pi(\gamma_*)}(y)   \right\|^2_{\gamma_*}  = O(1),$$
which for the semi-parametric mixture model can equivalently be written as
\begin{align}
\sup_{y \in {\cal Y}} 
   \frac{  {\rm Var}_{\pi(\gamma_*)} \, \left[ g_{\beta_0}(y\,|\,A) \right]   }
  {\left[ \mathbb{E}_{\pi(\gamma_*)} \, g_{\beta_0} (y\,|\,A) \right]^2} = O(1) .
     \label{BoundSupGradient}
 \end{align} 
 For any standard discrete choice model  (as those discussed in Section~\ref{Sec_numeric}) 
we have that $\sup_{y \in {\cal Y}} {\rm Var}_{\pi(\gamma_*)} \,[ g_{\beta_0}(y\,|\,A) ]< \infty$,
and $\inf_{y \in {\cal Y}}  \mathbb{E}_{\pi(\gamma_*)} \, g_{\beta_0} (y\,|\,A) >0$,
implying that equation \eqref{BoundSupGradient} is satisfied. 
We then have
\begin{align*}
    & \sup_{\pi_0 \in \Gamma_{\epsilon}(\gamma_*)}    
 		   \int_{\cal{Y}}   \left\| \nabla_\pi \log f_{\beta_0,\pi(\gamma_*)}(y)   \right\|^2_{\gamma_*}   \left[  f^{1/2}_{\beta_0,\pi_0}(y) - f^{1/2}_{\beta_0,\pi(\gamma_*)}(y)
 		\right]^2   dy 
\\		
    & \qquad
    \leq	\underbrace{ \left[\sup_{y \in {\cal Y}}   \left\| \nabla_\pi \log f_{\beta_0,\pi(\gamma_*)}(y)   \right\|^2_{\gamma_*} \right]
    }_{=O(1)}
\underbrace{ \left\{
     \sup_{\pi_0 \in \Gamma_{\epsilon}(\gamma_*)}    
 		   \int_{\cal{Y}}    \left[  f^{1/2}_{\beta_0,\pi_0}(y) - f^{1/2}_{\beta_0,\pi(\gamma_*)}(y)
 		\right]^2   dy 	 \right\}
		}_{\leq \, \epsilon \, = o(1)}
		= o(1) .
\end{align*}
Thus, one way to justify the second condition  in Assumption~\ref{ass:Expansion}(iii) is to argue that equation \eqref{BoundSupGradient} holds,
which is the case for our illustrations in Section~\ref{Sec_numeric}.
 The last condition in Assumption~\ref{ass:Expansion}(iii) could be broken down analogously for semi-parametric mixture models,
 but it is actually a standard condition of  differentiability in quadratic mean that is also regularly
 imposed when $\pi$ is infinite-dimensional (see, e.g., equation (5.38) in Van der Vaart, 2007).

 Part (iv) of Assumption~\ref{ass:Expansion} requires that our
 distance measure $d(\pi_0,\pi(\gamma_*))$ converges to the associated norm for small values $\epsilon$ in a smooth fashion. Finally, part (v) requires invertibility of $H_{\beta\gamma}$  (but invertibility of $ H_{\pi}$ or $\widetilde H_{\pi}$ are {\it not} required), 
 uniform boundedness of various derivatives, and of the $(2+\nu)$-th moment of $\nabla_\pi \log f_{\beta_0,\pi(\gamma_*)}(Y) $  --- which again can be justified by equation \eqref{BoundSupGradient}, because we then have
  $\sup_{y \in {\cal Y}}   \left\| \nabla_\pi \log f_{\beta_0,\pi(\gamma_*)}(y)   \right\|^2_{\gamma_*}  = O(1)$.
 
 For many of the proofs we only need the regularity 
 conditions in  Assumption~\ref{ass:Expansion}. However, in order to describe the 
 properties of our minimum-MSE estimator $\widehat \delta\,^{\rm MMSE}_\epsilon 
 = \delta_{\widehat \beta,\pi(\widehat \gamma)}+\frac 1 n \sum_{i=1}^n h_\epsilon^{\rm MMSE}(Y_i,\widehat \beta,\widehat \gamma)$ 
 we also need to account for the fact that $\widehat \beta$ and $\widehat \gamma$ themselves are estimated. It turns our that the leading-order asymptotic properties
 of   $\widehat \delta\,^{\rm MMSE}_\epsilon $ are independent of whether $\beta_0$ and $\gamma_*$ are known or estimated 
 in the construction of $\widehat \delta\,^{\rm MMSE}_\epsilon $ (see, e.g., Lemma~\ref{lemma:EtaEstimatorEffect} below),
 but formally showing this requires some additional assumptions, which we present next.

 \begin{assumption}
 	\label{ass:Expansion2}
 	For some $\chi>2$ we have
 	\begin{itemize}
 		\item[(i)] 
 		$\displaystyle  \sup_{\pi_0 \in \Gamma_{\epsilon}(\gamma_*)} \left( \mathbb{E}_{\beta_0,\pi_0}
 		\left\| {\widehat \beta \choose \widehat \gamma}
 		- {\beta_0 \choose \gamma_*} \right\|^{\chi} \right)^{\frac 1 {\chi}} = O\left( \frac 1 {\sqrt{n}} \right)$.
 		
 		\item[(ii)]
 		$\displaystyle  \sup_{\pi_0 \in \Gamma_{\epsilon}(\gamma_*)}  \mathbb{E}_{\beta_0,\pi_0}  
 		\left\|  \nabla_\eta \, h^{\rm MMSE}_\epsilon(Y,  \beta_0,\gamma_*) \right\|
 		= O(1)$, where $\eta=(\beta', \gamma')'$.
 		
 		\item[(iii)]
 		$\displaystyle     \sup_{\pi_0 \in \Gamma_{\epsilon}(\gamma_*)}   \mathbb{E}_{\beta_0,\pi_0} 
 		\sup_{\beta \in {\cal B}, \gamma \in {\cal G}}
 		\left\|  \frac 1 n \sum_{i=1}^n   \nabla^2_{\eta \eta'} h^{\rm MMSE}_\epsilon(Y_i,  \beta,\gamma)  \right\|  = O(1) $,
 		where $\eta=(\beta', \gamma')'$.
		
		\item[(iv)] 
  $ \displaystyle
 	\sup_{\pi_0 \in \Gamma_{\epsilon}(\gamma_*)}  \mathbb{E}_{\beta_0,\pi_0}  \left[ h^{\rm MMSE}_{\epsilon}(Y,\beta_0,\gamma_*)  \right]^{2+\nu}
 	=  O(1) ,
 	$ for some $\nu>0$.

 	\end{itemize}

 \end{assumption}    
 
 Part (i) of Assumption~\ref{ass:Expansion2} requires $ \widehat \beta$ and $\widehat \gamma$ to converge at 
 $\sqrt{n}$ rate. As discussed in the main text, we assume that preliminary estimators have finite $\chi$-moments where $\chi>2$. Part (ii) of  Assumption~\ref{ass:Expansion2}  requires a uniformly bounded second moment
 for $ \nabla_\eta h^{\rm MMSE}_\epsilon(y,\beta_0,\gamma_*)$. Since equation \eqref{SolutionMMSE} in the main text gives  an explicit expression for 
 $h^{\rm MMSE}_\epsilon(y,\beta_0,\gamma_*)$, we could replace Assumption~\ref{ass:Expansion2}(ii) by appropriate assumptions
 on the model primitives $f_{\beta_0,\pi_0}(y)$ and $\delta_{\beta_0,\pi_0}$, but for the sake of brevity  we state the assumption in terms of $h^{\rm MMSE}_\epsilon(y,\beta_0,\gamma_*)$.
 The same is true for part (iii) of Assumption~\ref{ass:Expansion2}. Notice that this last part of the assumption involves a supremum over $\beta$ and $\gamma$
 inside of an expectation -- in order to verify it, one either
 requires a uniform Lipschitz bound on the dependence of $h^{\rm MMSE}_\epsilon(Y_i,\beta,\gamma) $ on $\beta$
 and $\gamma$, or some empirical process method to control the entropy of that function (e.g., a bracketing argument). But since $\beta$ and $\gamma$ are finite-dimensional parameters these are all standard arguments.

We verified Assumption~\ref{ass:Expansion2}(iv) in the locally quadratic case of Section~\ref{Sec_param}. Formally, we have the following lemma.

  \begin{lemma}
 	\label{lemma:hMMSEmoment}
 	Let Assumption~\ref{ass:Expansion} hold,
	and assume that $h^{\rm MMSE}_{\epsilon}(\cdot,\beta_0, \gamma_*)$ is given by Lemma~\ref{lem_locquad}
	in the main text.
        Then, Assumption~\ref{ass:Expansion2}(iv) holds with the constant $\nu$ specified in Assumption~\ref{ass:Expansion}.
 \end{lemma}

 \subsubsection{Proof of  Theorem~\ref{theo1}}

 For a function $h_\epsilon=h_\epsilon(y,\beta_0,\gamma_*)$ we define
 \begin{align*}
 \widehat \delta(h_\epsilon,\beta_0,\gamma_*)  &:=  \delta_{\beta_0,\pi(\gamma_*)} 
 +    \frac 1 n \sum_{i=1}^n  h_{\epsilon}(Y_i, \beta_0, \gamma_*)  .
 \end{align*}
 It is useful to establish some preliminary lemmas before showing the main result. The proofs for those lemmas are provided in Section~\ref{app:ProofLemmas}. 
 
 \begin{lemma}
 	\label{lemma:MSEapprox}
 	Let Assumption~\ref{ass:Expansion} hold,
 	and let $h_{\epsilon}(\cdot,\beta_0, \gamma_*)$ be
 	a sequence of influence functions  that satisfy the unbiasedness constraint \eqref{Con:Unbiased}
 	as well as $\sup_{\pi_0 \in \Gamma_{\epsilon}(\gamma_*)}  \mathbb{E}_{\beta_0,\pi_0}  \left| h_{\epsilon}(Y, \beta_0, \gamma_*)  \right|^{\kappa} = O(1)$, for some 
 	$\kappa > 2$.
 	Then,
 	\begin{align*}
 	\sup_{\pi_0 \in \Gamma_\epsilon(\gamma_*)}   
 	\mathbb{E}_{\beta_0,\pi_0}  \left[  
 	\widehat \delta(h_\epsilon,\beta_0,\gamma_*) -  \delta_{\beta_0,\pi_0}      \right]^2 
 	&= b_{\epsilon}(h_\epsilon,\beta_0,\gamma_*)^2
 	+  \frac{{\rm Var}_{\beta_0,\pi(\gamma_*)}(h_\epsilon(Y,\beta_0,\gamma_*))  } {n} 
 	+ o(\epsilon) .
 	\end{align*}	
 \end{lemma}
 
 Lemma~\ref{lemma:MSEapprox} provides a formal justification for the worst-case MSE approximation
 introduced in equation \eqref{eqMSE} of the main text.

 Recall that 
 $\widehat \delta\,^{\rm MMSE}_\epsilon =  \widehat \delta(h^{\rm MMSE}_\epsilon, \widehat \beta, \widehat \gamma)$.
 This differs from 
 $ \widehat \delta(h^{\rm MMSE}_\epsilon,\beta_0,\gamma_*) $, because
 $\beta_0$ and $\gamma_*$ have to be estimated.
 The following lemma shows that the fact that $\beta_0$ and $\gamma_*$  are estimated in the construction of $\widehat \delta\,^{\rm MMSE}_\epsilon $
 can be neglected to first order. Notice that this result requires the additional regularity conditions in Assumption~\ref{ass:Expansion2}, 
 which are not required anywhere else in the proof of Theorem~\ref{theo1}.
 
 \begin{lemma}
 	\label{lemma:EtaEstimatorEffect}
 	Let Assumptions~\ref{ass:Expansion} and~\ref{ass:Expansion2} hold. Then,	
 	\begin{align*}	
 	\sup_{\pi_0\in \Gamma_\epsilon(\gamma_*)}   \mathbb{E}_{\beta_0,\pi_0} \left|     \widehat \delta^{\, \rm MMSE}_{\epsilon}  -  \widehat \delta(h^{\rm MMSE}_\epsilon,\beta_0,\gamma_*)  \right|
 	&= O\left( \frac 1 n \right) .
 	\end{align*}
 \end{lemma}
 
 Thus, Lemma~\ref{lemma:EtaEstimatorEffect} guarantees that
 $ \widehat \delta^{\, \rm MMSE}_{\epsilon} =  \widehat \delta(h^{\rm MMSE}_\epsilon,\beta_0,\gamma_*) + O_{P_0}(1/n)$.
 This may be surprising given that the differences 
 $\widehat \beta - \beta_0$ and $\widehat \gamma - \gamma_*$ are themselves of order $1/\sqrt{n}$.
 However, recall that by construction $h^{\rm MMSE}_\epsilon$ satisfies the local robustness
 condition \eqref{CharacterizeInfluenceH}, which is imposed through our constraints \eqref{Con:Unbiased}
 and \eqref{Con:EtaGradient}. Local robustness ensures that $\widehat \beta - \beta_0$ and $\widehat \gamma - \gamma_*$ 
 have no leading-order effect on $  \widehat \delta^{\, \rm MMSE}_{\epsilon}  -  \widehat \delta(h^{\rm MMSE}_\epsilon,\beta_0,\gamma_*) $.
 
 For the next lemma,
 recall the decomposition of $\widehat \delta_{\epsilon}  $ in Theorem~\ref{theo1} in the main text:
 \begin{align}
 \widehat \delta_{\epsilon}  &=  \delta_{\beta_0,\pi(\gamma_*)} 
 +    \frac 1 n \sum_{i=1}^n  h_{\epsilon}(Y_i, \beta_0, \gamma_*)  
 +  n^{-1/2} \,   R_n 
 \nonumber   \\
 &=   \widehat \delta(h_\epsilon,\beta_0,\gamma_*)  +  n^{-1/2} \,   R_n      .
 \label{DefRemainder1app}          
 \end{align}	
 Here, $  \widehat \delta(h_\epsilon,\beta_0,\gamma_*)  $  is the well-behaved leading-order contribution to $\widehat \delta_{\epsilon} $,
 whereas $ R_n$ is an asymptotically vanishing remainder term that may, however, have heavy tails (it only satisfies a trimmed second moment condition).
 The following lemma shows that the worst-case trimmed MSE of $   \widehat \delta_{\epsilon}  $ is bounded from below by the
 MSE of the leading-order term $  \widehat \delta(h_\epsilon,\beta_0,\gamma_*)  $.
 
 \begin{lemma}
 	\label{lemma:trimmedMSEtrunc}
 	Let Assumption~\ref{ass:Expansion} hold,
 	and let $h_{\epsilon}(\cdot,\beta_0, \gamma_*)$ be
 	a sequence of influence functions  that satisfy the unbiasedness constraint \eqref{Con:Unbiased}
 	as well as $\sup_{\pi_0 \in \Gamma_{\epsilon}(\gamma_*)}  \mathbb{E}_{\beta_0,\pi_0}  \left| h_{\epsilon}(Y, \beta_0, \gamma_*)  \right|^{\kappa} = O(1)$, for some 
 	$\kappa > 2$.
 	Assume that \eqref{DefRemainder1app} holds, and
 	let $m_n>0$ be a sequence such that $m_n \, n^{1/2}   \,  [\log(n)]^{-1}    \rightarrow \infty$.
 	Furthermore, assume that
 	\begin{itemize}

 		\item[(i)]
 		$\displaystyle  \sup_{\pi_0\in \Gamma_\epsilon(\gamma_*)}
 		{\rm P}_{\beta_0,\pi_0} \left(     \left|     R_n \right| >  \log(n)  \right) = o(1)$,
 		
 		\item[(ii)]
 		$\displaystyle  \sup_{\pi_0 \in \Gamma_{\epsilon}(\gamma_*)} \mathbb{E}_{\beta_0,\pi_0}
 		\left[    R_n^2 \, \mathbbm{1}\left(  |R_n| \leq 2\,  \log(n)  \right)  \right] = o( 1)$.

 	\end{itemize}
 	Then we have
 	\begin{align}
 	& 
 	\sup_{\pi_0\in \Gamma_\epsilon(\gamma_*)}
 	\mathbb{E}_{\beta_0,\pi_0}  \left[ \left(    \widehat \delta(h_\epsilon,\beta_0,\gamma_*)    - \delta_{\beta_0,\pi_0}
 	\right)^2  \right]
 	\nonumber \\		
 	& \qquad \leq	               	               
 	\sup_{\pi_0\in \Gamma_\epsilon(\gamma_*)}  \mathbb{E}_{\beta_0,\pi_0}  \left[ \left( \widehat \delta_{\epsilon}  - \delta_{\beta_0,\pi_0} \right)^2 
 	\mathbbm{1}\left( \left| \widehat \delta_{\epsilon}   - \delta_{\beta_0,\pi_0} \right| \leq 
 	m_n \right)
 	\right]  + o(\epsilon).
 	\label{StatementLemmaMSEtrunc}	
 	\end{align}
 	
 \end{lemma}
 
 We now have all the preliminary results required to show the main theorem.

 \begin{proof}[\bf Proof of Theorem~\ref{theo1}]
 	Define
 	\begin{align*}
 	r_\epsilon :=   \widehat \delta^{\, \rm MMSE}_{\epsilon}  -  \widehat \delta(h^{\rm MMSE}_\epsilon,\beta_0,\gamma_*) .
 	\end{align*}
 	We then have
 	\begin{align*}
 	&
 	\mathbb{E}_{\beta_0,\pi_0}  \left[ 
 	\left( \widehat \delta^{\, \rm MMSE}_{\epsilon}  - \delta_{\beta_0,\pi_0} \right)^2 
 	\mathbbm{1}\left( \left| \widehat \delta^{\, \rm MMSE}_{\epsilon}  - \delta_{\beta_0,\pi_0} \right|
 	\leq  m_n
 	\right)		
 	\right]
 	\\ 
 	&=  
 	\mathbb{E}_{\beta_0,\pi_0}  \left[ 
 	\left( \widehat \delta(h^{\rm MMSE}_\epsilon,\beta_0,\gamma_*)  - \delta_{\beta_0,\pi_0} + r_\epsilon \right)^2 
 	\mathbbm{1}\left( \left|  \widehat \delta^{\, \rm MMSE}_{\epsilon}  - \delta_{\beta_0,\pi_0}   \right|
 	\leq  m_n
 	\right)		
 	\right]
 	\\ 
 	&=  
 	\mathbb{E}_{\beta_0,\pi_0} \bigg[
 	\left( \widehat \delta(h^{\rm MMSE}_\epsilon,\beta_0,\gamma_*)  - \delta_{\beta_0,\pi_0}   \right)^2 
 	\underbrace{	
 		\mathbbm{1}\left( \left| \widehat \delta^{\, \rm MMSE}_{\epsilon}  - \delta_{\beta_0,\pi_0}   \right|
 		\leq  m_n
 		\right)
 	}_{\leq 1}		
 	\bigg]
 	\\ & \quad 
 	+ 2 \,   \mathbb{E}_{\beta_0,\pi_0}  \left[ 
 	r_\epsilon
 	\left( \widehat \delta(h^{\rm MMSE}_\epsilon,\beta_0,\gamma_*)  - \delta_{\beta_0,\pi_0}  +   r_\epsilon  \right)
 	\mathbbm{1}\left( \left|  \widehat \delta^{\, \rm MMSE}_{\epsilon}  - \delta_{\beta_0,\pi_0}    \right|
 	\leq  m_n
 	\right)		
 	\right]		 
 	\\ & \quad 
 	\underbrace{  -   
 		\mathbb{E}_{\beta_0,\pi_0}  \left[ 
 		r_\epsilon^2 \;
 		\mathbbm{1}\left( \left|  \widehat \delta^{\, \rm MMSE}_{\epsilon}   - \delta_{\beta_0,\pi_0}    \right|
 		\leq  m_n
 		\right)		
 		\right]
 	}_{\leq 0}	 
 	\\ 
 	&\leq 	  \mathbb{E}_{\beta_0,\pi_0} \left[
 	\left( \widehat \delta(h^{\rm MMSE}_\epsilon,\beta_0,\gamma_*)  - \delta_{\beta_0,\pi_0}   \right)^2 
 	\right]
 	\\ & \quad 
 	+ 2 \,   \mathbb{E}_{\beta_0,\pi_0}  \bigg[
 	\underbrace{
 		r_\epsilon
 		\left(  \widehat \delta^{\, \rm MMSE}_{\epsilon}  - \delta_{\beta_0,\pi_0}    \right)
 		\mathbbm{1}\left( \left|  \widehat \delta^{\, \rm MMSE}_{\epsilon}  - \delta_{\beta_0,\pi_0}   \right|
 		\leq  m_n
 		\right)		
 	}_{\leq |  r_\epsilon | \, m_n }
 	\bigg]			
 	\\ 
 	&\leq 	  \mathbb{E}_{\beta_0,\pi_0} \left[
 	\left( \widehat \delta(h^{\rm MMSE}_\epsilon,\beta_0,\gamma_*)  - \delta_{\beta_0,\pi_0}   \right)^2 
 	\right] + 2  \, m_n \,  \mathbb{E}_{\beta_0,\pi_0} \left|    r_\epsilon \right|.
 	\end{align*} 
 	According to Lemma~\ref{lemma:EtaEstimatorEffect} 
 	we have $\sup_{\pi_0\in \Gamma_\epsilon(\gamma_*)}   \mathbb{E}_{\beta_0,\pi_0} \left|    r_\epsilon \right|
 	= O(1/n) = O(\epsilon)$,
 	and the assumptions of the theorem guarantee that $m_n = o(1)$.
 	We thus obtain
 	\begin{align}
 	\sup_{\pi_0\in \Gamma_\epsilon(\gamma_*)}  &
 	\mathbb{E}_{\beta_0,\pi_0}  \left[ 
 	\left( \widehat \delta^{\, \rm MMSE}_{\epsilon}  - \delta_{\beta_0,\pi_0} \right)^2 
 	\mathbbm{1}\left( \left| \widehat \delta^{\, \rm MMSE}_{\epsilon}  - \delta_{\beta_0,\pi_0} \right|
 	\leq  m_n
 	\right)		
 	\right]
 	\nonumber \\	 	 
 	& \qquad \leq
 	\sup_{\pi_0\in \Gamma_\epsilon(\gamma_*)}  \mathbb{E}_{\beta_0,\pi_0}  \left[ \left( \widehat \delta(h^{\rm MMSE}_\epsilon,\beta_0,\gamma_*)  - \delta_{\beta_0,\pi_0}   \right)^2 
 	\right] + o(\epsilon)  .
 	\label{ProofTheo1step1}	
 	\end{align} 
 	By definition $h^{\rm MMSE}_{\epsilon}$ also satisfies 
 	the unbiasedness constraint \eqref{Con:Unbiased}. Together with Assumption~\ref{ass:Expansion2}(iv)
 	this implies that $h^{\rm MMSE}_{\epsilon}$ satisfies the conditions on $h_{\epsilon}$ in Lemma~\ref{lemma:MSEapprox}
 	with $\kappa=2+\nu$. Thus, we can apply Lemma~\ref{lemma:MSEapprox} with $h_{\epsilon}=h^{\rm MMSE}_{\epsilon}$
 	to find that 
 	\begin{align}
 	& \sup_{\pi_0\in \Gamma_\epsilon(\gamma_*)}  \mathbb{E}_{\beta_0,\pi_0}  \left[ \left( \widehat \delta(h^{\rm MMSE}_\epsilon,\beta_0,\gamma_*)  - \delta_{\beta_0,\pi_0}   \right)^2 
 	\right] 
 	\nonumber    \\
 	& \qquad = b_{\epsilon}(h^{\rm MMSE}_\epsilon,\beta_0,\gamma_*)^2
 	+  \frac{{\rm Var}_{\beta_0,\pi(\gamma_*)}(h^{\rm MMSE}_\epsilon(Y,\beta_0,\gamma_*))  } {n} 
 	+ o(\epsilon) .
 	\label{ProofTheo1step2}	
 	\end{align}	
 	 The function $h^{\rm MMSE}_\epsilon(\cdot,\beta_0,\gamma_*)$
 	is defined by the  minimization problem \eqref{MMSEproblem} in the main text.
	In other words, $h^{\rm MMSE}_\epsilon(\cdot,\beta_0,\gamma_*)$  
	minimizes the objective function  $b_{\epsilon}(h,\beta_0,\gamma_*)^2 +  n^{-1} {\rm Var}_{\beta_0,\pi(\gamma_*)}(h(Y,\beta_0,\gamma_*))   $,
	subject to the constraints  \eqref{Con:Unbiased} and \eqref{Con:EtaGradient}.
	Theorem~\ref{theo1} assumes that   $h_\epsilon = h_\epsilon(\cdot,\beta_0,\gamma_*)$  satisfies those constraints,
	and the definition of $h^{\rm MMSE}_\epsilon(\cdot,\beta_0,\gamma_*)$   therefore implies that
 	\begin{align}
 	& b_{\epsilon}(h^{\rm MMSE}_\epsilon,\beta_0,\gamma_*)^2 +  \frac{{\rm Var}_{\beta_0,\pi(\gamma_*)}(h^{\rm MMSE}_\epsilon(Y,\beta_0,\gamma_*))  } {n}
 	\nonumber \\ & \qquad \qquad \qquad \qquad \qquad \qquad
 	\leq
 	b_{\epsilon}(h_\epsilon,\beta_0,\gamma_*)^2 +  \frac{{\rm Var}_{\beta_0,\pi(\gamma_*)}(h_\epsilon(Y,\beta_0,\gamma_*))  } {n}.
 	\label{ProofTheo1step3}	      
 	\end{align}
 	Theorem~\ref{theo1} also imposes all the assumptions on $h_{\epsilon}$ in Lemma~\ref{lemma:MSEapprox} .
 	By applying that lemma we thus have
 	\begin{align}
 	& b_{\epsilon}(h_\epsilon,\beta_0,\gamma_*)^2
 	+  \frac{{\rm Var}_{\beta_0,\pi(\gamma_*)}(h_\epsilon(Y,\beta_0,\gamma_*))  } {n} 	
 	\nonumber    \\
 	& \qquad = \sup_{\pi_0\in \Gamma_\epsilon(\gamma_*)}  \mathbb{E}_{\beta_0,\pi_0}  \left[ \left( \widehat \delta(h_\epsilon,\beta_0,\gamma_*)  - \delta_{\beta_0,\pi_0}   \right)^2 
 	\right] 
 	+ o(\epsilon) .
 	\label{ProofTheo1step4}	
 	\end{align}	
 	Finally, Theorem~\ref{theo1} also guarantees all the assumptions  of Lemma~\ref{lemma:trimmedMSEtrunc}, implying that the inequality
 	\eqref{StatementLemmaMSEtrunc} holds.
 	Now,  combining \eqref{ProofTheo1step1}, \eqref{ProofTheo1step2}, \eqref{ProofTheo1step3}, \eqref{ProofTheo1step4} and \eqref{StatementLemmaMSEtrunc}
 	gives
 	\begin{align}
 	\sup_{\pi_0\in \Gamma_\epsilon(\gamma_*)} 
 	& \mathbb{E}_{\beta_0,\pi_0}  \left[ 
 	\left( \widehat \delta^{\, \rm MMSE}_{\epsilon}  - \delta_{\beta_0,\pi_0} \right)^2 
 	\mathbbm{1}\left( \left| \widehat \delta^{\, \rm MMSE}_{\epsilon}  - \delta_{\beta_0,\pi_0} \right|
 	\leq  m_n
 	\right)		
 	\right]
 	\nonumber \\ & \quad 	\leq 
 	\sup_{\pi_0\in \Gamma_\epsilon(\gamma_*)}  \mathbb{E}_{\beta_0,\pi_0}  \left[ \left( \widehat \delta_{\epsilon}  - \delta_{\beta_0,\pi_0} \right)^2 
 	\mathbbm{1}\left( \left| \widehat \delta_{\epsilon}   - \delta_{\beta_0,\pi_0} \right| \leq 
 	m_n \right)
 	\right]   
 	+ o(\epsilon)  ,
 	\end{align} 
 	which is what we wanted to show.
 \end{proof}

 \subsection{Proof of Theorem \ref{theo_CI}\label{App_CI}}

 \begin{assumption}$\quad$\label{ass_CI}
 	
 	\begin{enumerate}[(i)]

 		\item\label{ass_CI1} %
 		$ \widehat{\delta}-\delta_{\beta_0,\pi(\gamma_*)}-\frac{1}{n}\sum_{i=1}^n h(Y_i,\beta_0,\gamma_*) =o_{P_{\beta_0,\pi_0}}(n^{-\frac{1}{2}})$,
 		uniformly in $\pi_0 \in\Gamma_{\epsilon}(\gamma_*)$.
 		
 		\item\label{ass_CI2} Let $\sigma_h^2(\beta_0,\pi_0,\gamma_*)=\limfunc{Var}_{\beta_0,\pi_0} h(Y,\beta_0,\gamma_*)$. We assume that there exists a constant $c$, independent of $\epsilon$, such that
 		${\limfunc{inf}}_{\pi_0 \in\Gamma_\epsilon(\gamma_*) }\,\sigma_h(\beta_0,\pi_0,\gamma_*) \geq c >0$. 
 		Furthermore, for all sequences $a_n = c_{1-\alpha/2} + o(1)$ we have
 		$$
 		{\limfunc{inf}}_{\pi_0 \in\Gamma_\epsilon(\gamma_*) }\,{\limfunc{Pr}}_{\beta_0,\pi_0} \left[\left|\frac{1}{\sqrt{n}}\sum_{i=1}^n \frac{h(Y_i,\beta_0,\gamma_*)-\mathbb{E}_{\beta_0,\pi_0}h(Y,\beta_0,\gamma_*)}{\sigma_h(\beta_0,\pi_0,\gamma_*)}\right|\leq a_n \right]\geq 
 		1- \alpha + o(1).
 		$$%
 		\item ${\limfunc{sup}}_{\pi_0 \in\Gamma_\epsilon(\gamma_*) }\,{\mathbb{E}}_{\beta_0,\pi_0} \|\widehat{\beta}-\beta_0\|^2=o(1)$, \; 
		 ${\limfunc{sup}}_{\pi_0 \in\Gamma_\epsilon(\gamma_*) }\,{\mathbb{E}}_{\beta_0,\pi_0} \|\widehat{\gamma}-\gamma_*\|^2=o(1)$, \\ 
		${\limfunc{sup}}_{\pi_0 \in\Gamma_{\epsilon}(\gamma_*) }\,{\mathbb{E}}_{\beta_0,\pi_0} [\widehat{\sigma}_h-\sigma_h(\beta_0,\pi_0,\gamma_*)]^2=o(1)$.
 		
 		\item\label{ass_CI4} $ \|\nabla_{\beta\gamma}b_{\epsilon}(h,\beta,\gamma)\|=O(\epsilon^{\frac{1}{2}})$,
		uniformly in some neighborhood around $\beta_0$, $\gamma_*$.
 	\end{enumerate}
 	
 \end{assumption}
 
 Part (\ref{ass_CI1}) is weaker than the local regularity of the estimator $\widehat{\delta}$ that we assumed when analyzing the minimum-MSE estimator; see equation (\ref{DefRemainder1}). In turn, related to but differently from the conditions we used for Theorem \ref{theo1}, part (\ref{ass_CI2}) requires a form of local asymptotic normality of the estimator.  
 
 \begin{proof}[\bf Proof of Theorem \ref{theo_CI}.]
 	Let $\widehat{\delta}$ be an estimator and $h(y,\beta_0,\gamma_*)$ be the corresponding influence function
 	such that part (\ref{ass_CI1}) in Assumption \ref{ass_CI} holds. Define $\widehat{R}_{\beta_0,\gamma_*}:=\widehat{\delta}-\delta_{\beta_0,\pi(\gamma_*)}-\frac{1}{n}\sum_{i=1}^n h(Y_i,\beta_0,\gamma_*)$. 
 	We then have
 	\begin{align*}
 	&\widehat{\delta}-\delta_{\beta_0,\pi_0}  =\frac{1}{n}\sum_{i=1}^nh(Y_i,\beta_0,\gamma_*) + \delta_{\beta_0,\pi(\gamma_*)} -\delta_{\beta_0,\pi_0}   + \widehat{R}_{\beta_0,\gamma_*}
 	\\
 	& =  \frac{1}{n}\sum_{i=1}^n \left[ h(Y_i,\beta_0,\gamma_*) - \mathbb{E}_{\beta_0,\pi_0}h(Y,\beta_0,\gamma_*) \right] 
 	- \left[ \delta_{\beta_0,\pi_0}  - \delta_{\beta_0,\pi(\gamma_*)}  - \mathbb{E}_{\beta_0,\pi_0}h(Y,\beta_0,\gamma_*)  \right]   + \widehat{R}_{\beta_0,\gamma_*} ,
 	\end{align*}
 	and therefore
 	\begin{align}
 &	\underbrace{ 
 		\frac{|\widehat{\delta}-\delta_{\beta_0,\pi_0}|
 			- b_{\epsilon}(h,\widehat \beta,\widehat \gamma)-  \widehat{\sigma}_h \, c_{1-\alpha/2} / \sqrt{n}  
 		}
 		{\sigma_h(\beta_0,\pi_0,\gamma_*) / \sqrt{n}}
 	}_{= {\rm lhs}}
	\nonumber \\
 	&\qquad  \qquad \leq 
 	\underbrace{
 		\left|\frac{1}{\sqrt{n}}\sum_{i=1}^n \frac{h(Y_i,\beta_0,\gamma_*)-\mathbb{E}_{\beta_0,\pi_0}h(Y,\beta_0,\gamma_*)}{\sigma_h(\beta_0,\pi_0,\gamma_*)}\right|
 		- c_{1-\alpha/2} + \widehat r_{\beta_0,\pi_0,\gamma_*} 
 	}_{= {\rm rhs}},
 	\label{BoundCIderive}  
 	\end{align}
 	where
 	\begin{align*}
 	& \widehat r_{\beta_0,\pi_0,\gamma_*}
	\\
	 &:=  c_{1-\alpha/2}  +
 	\frac{ \left| \delta_{\beta_0,\pi_0}  - \delta_{\beta_0,\pi(\gamma_*)}  - \mathbb{E}_{\beta_0,\pi_0}h(Y,\beta_0,\gamma_*)  \right|+ \left| \widehat{R}_{\beta_0,\gamma_*} \right| 
 		- b_{\epsilon}(h,\widehat \beta,\widehat \gamma)-  \widehat{\sigma}_h \, c_{1-\alpha/2} / \sqrt{n}  }
 	{\sigma_h(\beta_0,\pi_0,\gamma_*) / \sqrt{n}} 
 	\\
 	&=
 	\frac{\sqrt{n}}{\sigma_h(\beta_0,\pi_0,\gamma_*)}
	\Bigg\{ |\delta_{\beta_0,\pi_0}-\delta_{\beta_0,\pi(\gamma_*)} -\mathbb{E}_{\beta_0,\pi_0}h(Y,\beta_0,\gamma_*)| 
 	+ |\widehat{R}_{\beta_0,\gamma_*}|
\\ & \qquad \qquad \qquad	 \qquad \qquad \qquad	\qquad \qquad \qquad	 \qquad  
 	- b_{\epsilon}(h,\widehat \beta,\widehat \gamma)
 	-\frac{\widehat{\sigma}_h-\sigma_h(\beta_0,\pi_0,\gamma_*)}{\sqrt{n}}c_{1-\alpha/2}
	 \Bigg\}        .
 	\end{align*}     
 	From \eqref{BoundCIderive}, we conclude that the event ${\rm rhs} \leq 0$
 	implies the event ${\rm lhs} \leq 0$, and therefore
 	${\limfunc{Pr}}_{\beta_0,\pi_0}({\rm lhs} \leq 0) \geq {\limfunc{Pr}}_{\beta_0,\pi_0}({\rm rhs} \leq 0)$, which we can also write as
 	\begin{align}
 	&{\limfunc{Pr}}_{\beta_0,\pi_0} \left[|\widehat{\delta}-\delta_{\beta_0,\pi_0}|\leq b_{\epsilon}(h,\widehat \beta,\widehat \gamma)+\frac{\widehat{\sigma}_h}{\sqrt{n}}c_{1-\alpha/2}\right]
 	\nonumber \\ & \qquad \qquad\qquad
 	\geq {\limfunc{Pr}}_{\beta_0,\pi_0} \left[\left|\frac{1}{\sqrt{n}}\sum_{i=1}^n \frac{h(Y_i,\beta_0,\gamma_*)-\mathbb{E}_{\beta_0,\pi_0}h(Y,\beta_0,\gamma_*)}{\sigma_h(\beta_0,\pi_0,\gamma_*)}\right|\leq c_{1-\alpha/2}-  \widehat r_{\beta_0,\pi_0,\gamma_*} \right].
 	\label{CIprobBound}
 	\end{align}                   
 	By part (\ref{ass_CI4}) in Assumption \ref{ass_CI} there exists a constant $C>0$ such that $ \|\nabla_{\beta\gamma}b_{\epsilon}(h,\beta,\gamma)\|\leq C \, \epsilon^{\frac{1}{2}}$, uniformly in a neighborhood of $(\beta_0,\gamma_*)$, and therefore
 	$$
     \left| b_{\epsilon}(h,\widehat \beta,\widehat \gamma) -   b_{\epsilon}(h, \beta_0, \gamma_* ) \right|
 	\leq  C \, \epsilon^{\frac{1}{2}} \,  \left\| { \widehat{\beta} -\beta_0 \choose \widehat \gamma - \gamma_*}    \right\|  .
 	$$
 	Using this we find that
 	\begin{align*}
 	\left| \widehat r_{\beta_0,\pi_0,\gamma_*} \right|
 	& \leq    \frac{\sqrt{n}}{\sigma_h(\beta_0,\pi_0,\gamma_*)}
 	\Bigg\{
 	\bigg|  \left|\delta_{\beta_0,\pi_0}-\delta_{\beta_0,\pi(\gamma_*)} -\mathbb{E}_{\beta_0,\pi_0}h(Y,\beta_0,\gamma_*) \right| 
 	- b_{\epsilon}(h,\beta_0,\gamma_*)    \bigg|    
 	\\ & \qquad  \qquad \qquad
 	+  \frac{\left| \widehat{\sigma}_h-\sigma_h(\beta_0,\pi_0,\gamma_*) \right| }{\sqrt{n}}c_{1-\alpha/2}     
 	+  C \, \epsilon^{\frac{1}{2}} \, \left\| { \widehat{\beta} -\beta_0 \choose \widehat \gamma - \gamma_*}    \right\| 
 	+ |\widehat{R}_{\beta_0,\gamma_*}|
 	\Bigg\}  .
 	\end{align*}
 	Parts (\ref{ass_CI1}) and (\ref{ass_CI2}) of Assumption \ref{ass_CI} imply that, uniformly in $\pi_0 \in\Gamma_{\epsilon}(\gamma_*)$,
 	we have
 	\begin{align*}
 	\frac{\sqrt{n}}{\sigma_h(\beta_0,\pi_0,\gamma_*)}  \widehat{R}_{\beta_0,\gamma_*} = o_{P_{\beta_0,\pi_0}}(1),
 	\end{align*}
 	and analogously we find from the conditions  in Assumption \ref{ass_CI} that
 	\begin{align*}
 	\frac{  \widehat{\sigma}_h-\sigma_h(\beta_0,\pi_0,\gamma_*)  }{\sigma_h(\beta_0,\pi_0,\gamma_*)}   &= o_{P_{\beta_0,\pi_0}}(1),
 	&
 	\frac{\sqrt{n}}{\sigma_h(\beta_0,\pi_0,\gamma_*)}
 	\epsilon^{\frac{1}{2}} \,  \left\| { \widehat{\beta} -\beta_0 \choose \widehat \gamma - \gamma_*}    \right\|   &= o_{P_{\beta_0,\pi_0}}(1),
 	\end{align*}
 	uniformly in $\pi_0 \in\Gamma_{\epsilon}(\gamma_*)$.
 	Finally, since we also impose  Assumption \ref{ass:Expansion}
 	and $ \sup_{\pi_0 \in \Gamma_{\epsilon}}  \allowbreak   \mathbb{E}_{\beta_0,\pi_0}   \allowbreak  h^2(Y,\beta_0,\gamma_*) = O(1)$ we obtain, analogously to
 	the proof of Lemma~\ref{lemma:Expansions}(iii) in Section \ref{App_comp1}, that 
 	\begin{align*}
 	\sup_{\pi_0 \in\Gamma_{\epsilon}(\gamma_*)}
 	\frac{\sqrt{n}}{\sigma_h(\beta_0,\pi_0,\gamma_*)}
 	\bigg|  \left|\delta_{\beta_0,\pi_0}-\delta_{\beta_0,\pi(\gamma_*)} -\mathbb{E}_{\beta_0,\pi_0}h(Y,\beta_0,\gamma_*) \right| 
 	- b_{\epsilon}(h,\beta_0,\gamma_*)    \bigg|    
 	&= o(1).
 	\end{align*}%
 	We thus conclude that $ \widehat r_{\beta_0,\pi_0,\gamma_*}   = o_{P_{\beta_0,\pi_0}}(1)$, uniformly in $\pi_0 \in\Gamma_{\epsilon}(\gamma_*)$.
 	Together with \eqref{CIprobBound} and  
 	part (\ref{ass_CI2}) in Assumption \ref{ass_CI} this implies (\ref{res_theo_CI}), hence Theorem \ref{theo_CI}.    	
 \end{proof}

 \clearpage
 
 \baselineskip15pt
 
 \begin{center}
 	\textbf{\LARGE SUPPLEMENTARY APPENDIX}
 \end{center}

 \begin{center}
	\textbf{\large ``Minimizing Sensitivity to Model Misspecification''}
\end{center}

\begin{center}
{\large St\'ephane Bonhomme and Martin Weidner}
\end{center}

\setcounter{page}{1}

\setcounter{section}{0}\renewcommand{\thesection}{S\arabic{section}}

\setcounter{figure}{0}\renewcommand{\thefigure}{S\arabic{figure}}

\setcounter{table}{0}\renewcommand{\thetable}{S\arabic{table}}

\setcounter{footnote}{0}\renewcommand{\thefootnote}{\arabic{footnote}}

\setcounter{theorem}{0}\renewcommand{\thetheorem}{S\arabic{theorem}}

\setcounter{assumption}{0}\renewcommand{\theassumption}{S\arabic{assumption}}

\setcounter{equation}{0}\renewcommand{\theequation}{S\arabic{equation}}

\setcounter{lemma}{0}\renewcommand{\thelemma}{S\arabic{lemma}}

\setcounter{proposition}{0}\renewcommand{\theproposition}{S\arabic{proposition}}

\setcounter{corollary}{0}\renewcommand{\thecorollary}{S\arabic{corollary}}

 In Sections \ref{App_comp1} and \ref{App_der_sec3}, we provide details about the proofs in the paper. In Section \ref{App_compute}, we describe our computational approach. In Section \ref{App_GMM}, we outline how to extend our approach to models defined by moment restrictions. Lastly, we report additional simulation and estimation results in Section \ref{App_Fig}.

 \section{Complements to main results of Section \ref{Sec_outline}\label{App_comp1}}
  
   \subsection{Proof of  intermediate lemmas for  Theorem~\ref{theo1}}
   \label{app:ProofLemmas}
   
   The proofs of the Lemmas~\ref{lemma:MSEapprox},
   \ref{lemma:EtaEstimatorEffect} and \ref{lemma:trimmedMSEtrunc} are provided in this subsection.
   Before those proofs it is useful to first establish one additional lemma.
   
   \begin{lemma} 
   	\label{lemma:Expansions}
   	Let Assumption~\ref{ass:Expansion} hold.
   	Let
   	$q_\epsilon(y)$
   	and	
   	$h_{\epsilon}(y,\beta_0, \gamma_*)$ be
   	sequences of  functions with
   	$\sup_{\pi_0 \in \Gamma_{\epsilon}(\gamma_*)}  \mathbb{E}_{\beta_0,\pi_0}  \left| q(Y)  \right|^{\zeta} = O(1)$, for some 
   	$\zeta > 1$,
   	and  $\sup_{\pi_0 \in \Gamma_{\epsilon}(\gamma_*)} \allowbreak \mathbb{E}_{\beta_0,\pi_0}   \left| h_{\epsilon}(Y, \beta_0, \gamma_*)  \right|^2 = O(1)$.	
   	Then we have
   	\begin{itemize}
   		\item[(i)]   $\displaystyle  \sup_{\pi_0\in \Gamma_\epsilon(\gamma_*)} 
   		\left|  \delta_{\beta_0,\pi(\gamma_*)}  - \delta_{\beta_0,\pi_0}
   		\right| = O(\epsilon^{1/2}) ,
   		$

   		\item[(ii)]
   		$ \displaystyle
   		\sup_{\pi_0 \in \Gamma_{\epsilon}(\gamma_*)}   
   		\left|  \mathbb{E}_{\beta_0,\pi_0}   q_\epsilon(Y) - \mathbb{E}_{\beta_0,\pi(\gamma_*)}   q_\epsilon(Y)  \right|
   		=   O(\epsilon^{1/2}) ,
   		$
   		
   		\item[(iii)]
   		$\displaystyle
   		\sup_{\pi_0 \in \Gamma_{\epsilon}(\gamma_*)}  
   		\bigg|  \mathbb{E}_{\beta_0,\pi_0}   h_\epsilon(Y,\beta_0,\gamma_*) - \mathbb{E}_{\beta_0,\pi(\gamma_*)}   h_\epsilon(Y,\beta_0,\gamma_*)
   		\\ $\phantom{a}$ \qquad \qquad\quad 
   		-  \left\langle  \pi_0 - \pi(\gamma_*) ,  \mathbb{E}_{\beta_0,\pi(\gamma_*)}   h_\epsilon(Y,\beta_0,\gamma_*)   \nabla_\pi \log f_{\beta_0,\pi(\gamma_*)}(Y)  \right\rangle   \bigg|
   		=    o(\epsilon^{1/2}) .
   		$

   	\end{itemize}       
   	
   \end{lemma}

   \begin{proof}[\bf Proof of Lemma~\ref{lemma:Expansions}]
   	\# \underline{Part (i):}
   	By a mean-value expansion around $\pi(\gamma_*)$ we find
   	\begin{align*}
   	\left|  \delta_{\beta_0,\pi_0} - \delta_{\beta_0,\pi(\gamma_*)} \right|
   	&= \left|  \left\langle \pi_0 - \pi(\gamma_*) , \nabla_\pi \delta_{\beta_0,\widetilde \pi} \right\rangle 
   	\right|
   	\leq   \left\|  \pi_0 - \pi(\gamma_*) \right\|_{{\rm ind},\gamma_*} 
   	\left\| \nabla_\pi \delta_{\beta_0,\widetilde \pi} \right\|_{\gamma_*} ,
   	\end{align*}
   	where $\widetilde \pi$ is between $\pi(\gamma_*)$ and $\pi_0$. Therefore
   	\begin{align*}
   	\sup_{\pi_0 \in \Gamma_{\epsilon}(\gamma_*)}   \left|  \delta_{\beta_0,\pi_0} - \delta_{\beta_0,\pi(\gamma_*)} \right|
   	& \leq     \sup_{\pi_0 \in \Gamma_{\epsilon}(\gamma_*)}    \left\|  \pi_0 - \pi(\gamma_*) \right\|_{{\rm ind},\gamma_*} 
   	\sup_{\pi_0 \in \Gamma_{\epsilon}(\gamma_*)}      \left\| \nabla_\pi \delta_{\beta_0,\pi_0} \right\|_{\gamma_*}
   	\\
   	&= O(\epsilon^{1/2}) \, O(1) = O(\epsilon^{1/2}) .
   	\end{align*}

   	\# \underline{Part (ii):}
   	Without loss of generality we assume that $\zeta \leq 2$.
   	Let $\xi :=   \zeta / (\zeta-1) \geq 2$. We then have
   	\begin{align*}
   	\int_{\cal{Y}}    \left|  f^{1/\xi}_{\beta_0,\pi_0}(y) - f^{1/\xi}_{\beta_0,\pi(\gamma_*)}(y)  \right|^{\xi}  dy 
   	&\leq 
   	\int_{\cal{Y}}    \left[  f^{1/2}_{\beta_0,\pi_0}(y) - f^{1/2}_{\beta_0,\pi(\gamma_*)}(y)  \right]^{2}  dy  ,
   	\end{align*}
   	where we used that $|a - b| \leq |a^c - b^c|^{1/c}$, for any $a,b\geq 0$ and $c\geq 1$, and
   	plugged in $a= f^{1/\xi}_{\beta_0,\pi_0}(y) $, $b= f^{1/\xi}_{\beta_0,\pi(\gamma_*)}(y) $, and $c=\xi/2$. Thus, the first part of
   	Assumption~\ref{ass:Expansion}(iii) also implies
   	\begin{align}
   	\sup_{\pi_0 \in \Gamma_{\epsilon}(\gamma_*)}   
   	\displaystyle 
   	\left\{
   	\int_{\cal{Y}}    \left|  f^{1/\xi}_{\beta_0,\pi_0}(y) - f^{1/\xi}_{\beta_0,\pi(\gamma_*)}(y)  \right|^{\xi}  dy
   	\right\}^{\frac 1 \xi} =  O(\epsilon^{1/2})  .
   	\label{ass:Fxi}       
   	\end{align}           
   	Next, we find
   	\begin{align*}
   	&\sup_{\pi_0 \in \Gamma_{\epsilon}(\gamma_*)}  
   	\left|  \mathbb{E}_{\beta_0,\pi_0}   q_\epsilon(Y) - \mathbb{E}_{\beta_0,\pi(\gamma_*)}   q_\epsilon(Y)  \right|
   	\nonumber \\
   	&= \sup_{\pi_0 \in \Gamma_{\epsilon}(\gamma_*)} 
   	\left|   \int_{\cal{Y}} q_\epsilon(Y) 
   	\frac{  f_{\beta_0,\pi_0}(y) - f_{\beta_0,\pi(\gamma_*)}(y) } 
   	{ f^{1/\xi}_{\beta_0,\pi_0}(y) - f^{1/\xi}_{\beta_0,\pi(\gamma_*)}(y) } 
   	\left[  f^{1/\xi}_{\beta_0,\pi_0}(y) - f^{1/\xi}_{\beta_0,\pi(\gamma_*)}(y)  \right]
   	dy \right|
   	\nonumber  \\
   	&\leq  
   	\left\{ \sup_{\pi_0 \in \Gamma_{\epsilon}(\gamma_*)}  \int_{\cal{Y}} \left| q_\epsilon(Y) \right|^{\frac{\xi} {\xi-1}} 
   	\left|   \frac{  f_{\beta_0,\pi_0}(y) - f_{\beta_0,\pi(\gamma_*)}(y) } 
   	{ f^{1/\xi}_{\beta_0,\pi_0}(y) - f^{1/\xi}_{\beta_0,\pi(\gamma_*)}(y) }   \right|^{ \frac{\xi} {\xi-1}} dy \right\}^{\frac{\xi-1}\xi}
   	\\ & \qquad \qquad \qquad \qquad \qquad \qquad \qquad \qquad
   	\times \left\{ \sup_{\pi_0 \in \Gamma_{\epsilon}(\gamma_*)}    \int_{\cal{Y}}    \left|  f^{1/\xi}_{\beta_0,\pi_0}(y) - f^{1/\xi}_{\beta_0,\pi(\gamma_*)}(y)  \right|^\xi  dy     \right\}^{\frac 1 \xi}
   	\nonumber  \\
   	&\leq  
   	\xi \left\{ \sup_{\pi_0 \in \Gamma_{\epsilon}(\gamma_*)}  \int_{\cal{Y}} \left| q_\epsilon(Y) \right|^{\frac{\xi} {\xi-1}}
   	\left|     f_{\beta_0,\pi_0}(y) + f_{\beta_0,\pi(\gamma_*)}(y)    \right|  dy \right\}^{\frac{\xi-1}\xi}
   	\\ & \qquad \qquad \qquad \qquad \qquad \qquad \qquad \qquad
   	\times
   	\left\{ \sup_{\pi_0 \in \Gamma_{\epsilon}(\gamma_*)}    \int_{\cal{Y}}    \left|  f^{1/\xi}_{\beta_0,\pi_0}(y) - f^{1/\xi}_{\beta_0,\pi(\gamma_*)}(y)  \right|^{\xi}  dy     \right\}^{\frac 1 \xi}
   	\nonumber  \\
   	&\leq  
   	\xi \left\{2   \sup_{\pi_0 \in \Gamma_{\epsilon}(\gamma_*)}     \mathbb{E}_{\beta_0,\pi_0}  \left| q_\epsilon(Y) \right|^{\zeta}  \right\}^{\frac{\xi-1}\xi}
   	\left\{ \sup_{\pi_0 \in \Gamma_{\epsilon}(\gamma_*)}   \int_{\cal{Y}}    \left|  f^{1/\xi}_{\beta_0,\pi_0}(y) - f^{1/\xi}_{\beta_0,\pi(\gamma_*)}(y)  \right|^{\xi}  dy     \right\}^{\frac 1 \xi}
   	\\
   	&= o(1) ,
   	\end{align*}
   	where 
   	the first inequality is an application of  H\"older's inequality,
   	the second inequality uses that
   	$\left| \frac{  f_{\beta_0,\pi_0}(y) - f_{\beta_0,\pi(\gamma_*)}(y) } 
   	{ f^{1/\xi}_{\beta_0,\pi_0}(y) - f^{1/\xi}_{\beta_0,\pi(\gamma_*)}(y) }   \right|^{\xi / (\xi-1)}
   	\leq  \xi^{\xi / (\xi-1)}   \left[ f_{\beta_0,\pi_0}(y) + f_{\beta_0,\pi(\gamma_*)}(y) \right]$,\footnote{%
   		For $a,b \geq 0$  there exists $c \in [a,b]$ such that by the mean value theorem we have
   		$(a^\xi - b^\xi) / (a-b) = \xi c^{\xi-1} \leq \xi \max(a^{\xi-1},b^{\xi-1})$, and therefore 
   		$[(a^\xi - b^\xi) / (a-b)]^{\xi / (\xi-1)} \leq \xi^{\xi / (\xi-1)} \max(a^\xi,b^\xi) \leq \xi^{\xi / (\xi-1)} (a^\xi + b^\xi)$,
   		which we apply here with $a= f^{1/\xi}_{\beta_0,\pi_0}(y)$ and $b=f^{1/\xi}_{\beta_0,\pi(\gamma_*)}(y) $.
   	}
   	the last line uses that $\kappa =  \xi / (\xi-1)$, and the final conclusion follows
   	from our assumptions and \eqref{ass:Fxi}.

   	\# \underline{Part (iii):}
   	We have
   	\begin{align*}
   	&   \mathbb{E}_{\beta_0,\pi_0}   h_\epsilon(Y,\beta_0,\gamma_*) - \mathbb{E}_{\beta_0,\pi(\gamma_*)}   h_\epsilon(Y,\beta_0,\gamma_*)
   	\\ & \qquad \qquad \qquad \qquad \qquad \qquad
   	-  \left\langle  \pi_0 - \pi(\gamma_*) ,  \mathbb{E}_{\beta_0,\pi(\gamma_*)}   h_\epsilon(Y,\beta_0,\gamma_*)   \nabla_\pi \log f_{\beta_0,\pi(\gamma_*)}(Y)  \right\rangle  
   	\\
   	&=
   	\int_{\cal{Y}}  h_\epsilon(y,\beta_0,\gamma_*) \left[  f_{\beta_0,\pi_0}(y) - f_{\beta_0,\pi(\gamma_*)}(y)
   	-  \left\langle  \pi_0 - \pi(\gamma_*) , \nabla_\pi \log f_{\beta_0,\pi(\gamma_*)}(y)  \right\rangle  f_{\beta_0,\pi(\gamma_*)}(y)  \right]   dy 
   	\\
   	&=  
   	\int_{\cal{Y}}  h_\epsilon(y,\beta_0,\gamma_*)
   	\left[  f^{1/2}_{\beta_0,\pi_0}(y) + f^{1/2}_{\beta_0,\pi(\gamma_*)}(y)
   	\right]
   	\\ & \quad
   	\underbrace{ \qquad \qquad \qquad
   		\times \left[  f^{1/2}_{\beta_0,\pi_0}(y) - f^{1/2}_{\beta_0,\pi(\gamma_*)}(y)
   		- \frac 1 2  \left\langle  \pi_0 - \pi(\gamma_*) , \nabla_\pi \log f_{\beta_0,\pi(\gamma_*)}(y)  \right\rangle   f^{1/2}_{\beta_0,\pi(\gamma_*)}(y)  \right]   dy
   	}_{=: a^{(1)}_{\beta_0,\gamma_*,\pi_0}}
   	\\ & \quad  
   	+  
   	\frac 1 2   
   	\underbrace{   \int_{\cal{Y}}  h_\epsilon(y,\beta_0,\gamma_*)  
   		f^{1/2}_{\beta_0,\pi(\gamma_*)}(y)
   		\left\langle  \pi_0 - \pi(\gamma_*) , \nabla_\pi \log f_{\beta_0,\pi(\gamma_*)}(y)  \right\rangle
   		\left[   f^{1/2}_{\beta_0,\pi_0}(y) - f^{1/2}_{\beta_0,\pi(\gamma_*)}(y)  \right]     dy
   	}_{=: a^{(2)}_{\beta_0,\gamma_*,\pi_0}}   .
   	\end{align*}
   	Applying the Cauchy-Schwarz inequality and our assumptions we find that
   	\begin{align*}
   	&  \sup_{\pi_0 \in \Gamma_{\epsilon}(\gamma_*)}    \left|  a^{(1)}_{\beta_0,\gamma_*,\pi_0} \right|^2
   	\\ 
   	&\leq 4  \left\{   \sup_{\pi_0 \in \Gamma_{\epsilon}(\gamma_*)}     \mathbb{E}_{\beta_0,\pi_0}  h_\epsilon^2(Y,\beta_0,\gamma_*)  \right\}
   	\\ & \qquad\qquad\qquad 
   	\times \left\{   \sup_{\pi_0 \in \Gamma_{\epsilon}(\gamma_*)}    
   	\int_{\cal{Y}}  \left[  f^{1/2}_{\beta_0,\pi_0}(y) - f^{1/2}_{\beta_0,\pi(\gamma_*)}(y)
   	-   \left\langle  \pi_0 - \pi(\gamma_*) , \nabla_\pi   f^{1/2}_{\beta_0,\pi(\gamma_*)}(y)  \right\rangle    \right]^2   dy
   	\right\}
   	\\     
   	&= O(\epsilon^{1/2}) ,
   	\end{align*}
   	and
   	\begin{align*}
   	&  \sup_{\pi_0 \in \Gamma_{\epsilon}(\gamma_*)}    \left|  a^{(2)}_{\beta_0,\gamma_*,\pi_0} \right|^2
   	\\
   	&\leq  \left\{    \mathbb{E}_{\beta_0,\pi(\gamma_*)}  h^2_\epsilon(Y,\beta_0,\gamma_*)  \right\}
   	\\ & \qquad \quad \times 
   	\left\{   \sup_{\pi_0 \in \Gamma_{\epsilon}(\gamma_*)}    
   	\left\| \pi_0 - \pi(\gamma_*) \right\|_{{\rm ind},\gamma_*}^2 
   	\int_{\cal{Y}}   \left\| \nabla_\pi \log f_{\beta_0,\pi(\gamma_*)}(y)   \right\|^2_{\gamma_*}   \left[  f^{1/2}_{\beta_0,\pi_0}(y) - f^{1/2}_{\beta_0,\pi(\gamma_*)}(y)
   	\right]^2   dy
   	\right\} 
   	\\
   	&= o(\epsilon).  
   	\end{align*}
   	Combining this gives the statement in the lemma.
   \end{proof}

     \begin{proof}[\bf Proof of Lemma~\ref{lemma:MSEapprox}]
   	Applying part   (ii) of Lemma~\ref{lemma:Expansions}
   	with $q_\epsilon(y)=  h_{\epsilon}(y, \beta_0, \gamma_*) $
   	and using the unbiasedness constraint \eqref{Con:Unbiased}
   	we find that $ \mathbb{E}_{\beta_0,\pi_0}   h_\epsilon(Y,\beta_0,\gamma_*)  =   o(1)$,
   	uniformly in $\pi_0 \in \Gamma_{\epsilon}(\gamma_*)$.
   	Part (i) of Lemma~\ref{lemma:Expansions}  guarantees that 
   	$ \left| \delta_{\beta_0,\pi_0} - \delta_{\beta_0,\pi(\gamma_*)}  \right| = o(1)$,  uniformly in $\pi_0 \in \Gamma_{\epsilon}(\gamma_*)$. 
   	We therefore have	
   	\begin{align*}
   	&\mathbb{E}_{\beta_0,\pi_0}  \left[ 
   	h_\epsilon(Y,\beta_0,\gamma_*)   + \delta_{\beta_0,\pi(\gamma_*)} -  \delta_{\beta_0,\pi_0}   \right]^2 
   	\\ & \qquad
   	=    \mathbb{E}_{\beta_0,\pi_0}  \left[ h_\epsilon(Y,\beta_0,\gamma_*)    \right]^2 
   	-  2 \, \left(  \delta_{\beta_0,\pi_0} - \delta_{\beta_0,\pi(\gamma_*)}  \right) \,  \mathbb{E}_{\beta_0,\pi_0}   h_\epsilon(Y,\beta_0,\gamma_*) 
   	+  \left(  \delta_{\beta_0,\pi_0} - \delta_{\beta_0,\pi(\gamma_*)}  \right)^2 
   	\\
   	&  \qquad=     \mathbb{E}_{\beta_0,\pi_0}  \left[ h_\epsilon(Y,\beta_0,\gamma_*)    \right]^2 + o(1)  ,
   	\end{align*}
   	uniformly in $\pi_0 \in \Gamma_{\epsilon}(\gamma_*)$.
   	Applying part   (ii) of Lemma~\ref{lemma:Expansions}
   	with $q_\epsilon(y)=  \left[ h_{\epsilon}(y, \beta_0, \gamma_*) \right]^2$
   	we find that
   	$\mathbb{E}_{\beta_0,\pi_0}  \left[ h_\epsilon(Y,\beta_0,\gamma_*)    \right]^2
   	=   \mathbb{E}_{\beta_0,\pi(\gamma_*)}  \left[ h_\epsilon(Y,\beta_0,\gamma_*)    \right]^2 + o(1)
   	=   {\rm Var}_{\beta_0,\pi(\gamma_*)}(h_\epsilon(Y,\beta_0,\gamma_*)) + o(1) 
   	$, uniformly in $\pi_0 \in \Gamma_{\epsilon}(\gamma_*)$,
   	where in the last step we have also used that $h_\epsilon(y,\beta_0,\gamma_*)$ satisfies the unbiasedness constraint
   	\eqref{Con:Unbiased}.
   	Therefore,
   	\begin{align}
   	\sup_{\pi_0 \in \Gamma_{\epsilon}(\gamma_*)} 
   	\mathbb{E}_{\beta_0,\pi_0}  \left[ 
   	h_\epsilon(Y,\beta_0,\gamma_*)   + \delta_{\beta_0,\pi(\gamma_*)} -  \delta_{\beta_0,\pi_0}   \right]^2 
   	&=     {\rm Var}_{\beta_0,\pi(\gamma_*)}(h_\epsilon(Y,\beta_0,\gamma_*)) + o(1)  . 
   	\label{ApproxVAR}
   	\end{align}	
   	Using  the unbiasedness constraint again, as well as Lemma~\ref{lemma:Expansions}(iii)  and Assumptions~\ref{ass:Expansion}(ii) and~\ref{ass:Expansion}(iv) we find
   	\begin{align}
   	&  \sup_{\pi_0 \in \Gamma_{\epsilon}(\gamma_*)} 
   	\left|  \mathbb{E}_{\beta_0,\pi_0}   h_\epsilon(Y,\beta_0,\gamma_*) + \delta_{\beta_0,\pi(\gamma_*)}  -  \delta_{\beta_0,\pi_0}   \right|
   	\nonumber \\
   	&  =     \sup_{\pi_0 \in \Gamma_{\epsilon}(\gamma_*)} 
   	\left|   \left\langle  \pi_0 - \pi(\gamma_*) ,  \mathbb{E}_{\beta_0,\pi(\gamma_*)}   h_\epsilon(Y,\beta_0,\gamma_*)   \nabla_\pi \log f_{\beta_0,\pi(\gamma_*)}(Y) 
   	-  \nabla_\pi \delta_{\beta_0,\pi(\gamma_*)} \right\rangle
   	\right|    +   o(\epsilon^{1/2}) 
   	\nonumber\\  
   	&    = 
   	\epsilon^{1/2} \left\|
   	\mathbb{E}_{\beta_0,\pi(\gamma_*)}   h_\epsilon(Y,\beta_0,\gamma_*)   \nabla_\pi \log f_{\beta_0,\pi(\gamma_*)}(Y) 
   	-  \nabla_\pi \delta_{\beta_0,\pi(\gamma_*)}
   	\right\|_{\gamma_*}
   	+   o(\epsilon^{1/2})   
   	\nonumber \\
   	&
   	= b_{\epsilon}(h_\epsilon,\beta_0,\gamma_*) +  o(\epsilon^{1/2})    ,
   	\label{ApproxBIAS}
   	\end{align}
   	where in the last step we used the definition of the worst-case bias in \eqref{eqbepsilon} of the main text.
   	We furthermore have
   	\begin{align*}
   	&   
   	\mathbb{E}_{\beta_0,\pi_0}  \left[  
   	\widehat \delta(h_\epsilon,\beta_0,\gamma_*) -  \delta_{\beta_0,\pi_0}      \right]^2 
   	\\
   	& =      \mathbb{E}_{\beta_0,\pi_0}  \left(  
   	\frac 1 n \sum_{i=1}^n h(Y_i,\beta_0,\gamma_*) + \delta_{\beta_0,\pi(\gamma_*)}  -  \delta_{\beta_0,\pi_0}    \right)^2  
   	\\
   	&=  \left[   \mathbb{E}_{\beta_0,\pi_0}   h(Y,\beta_0,\gamma_*) + \delta_{\beta_0,\pi(\gamma_*)}  -  \delta_{\beta_0,\pi_0}    \right]^2  
   	+ \frac 1 n {\rm Var}_{\beta_0,\pi_0}\left[ h(Y,\beta_0,\gamma_*)  + \delta_{\beta_0,\pi(\gamma_*)}  -  \delta_{\beta_0,\pi_0}    \right] 
   	\\    
   	&=  \frac {n-1} n  \left[   \mathbb{E}_{\beta_0,\pi_0}   h(Y,\beta_0,\gamma_*) -  \delta_{\beta_0,\pi_0} +\delta_{\beta_0,\pi(\gamma_*)}   \right]^2  
   	+ \frac 1 n  \mathbb{E}_{\beta_0,\pi_0}  \left[ h(Y,\beta_0,\gamma_*) + \delta_{\beta_0,\pi(\gamma_*)}  -  \delta_{\beta_0,\pi_0}    \right]^2 .
   	\end{align*}
   	Taking the supremum of this last result over $\pi_0 \in \Gamma_\epsilon(\gamma_*)$,
   	and then applying \eqref{ApproxVAR} and \eqref{ApproxBIAS} gives
   	\begin{align*}
   	\sup_{\pi_0 \in \Gamma_\epsilon(\gamma_*)}   
   	\mathbb{E}_{\beta_0,\pi_0}  \left[  
   	\widehat \delta(h_\epsilon,\beta_0,\gamma_*) -  \delta_{\beta_0,\pi_0}      \right]^2 
   	&=
   	b_{\epsilon}(h_\epsilon,\beta_0,\gamma_*)^2
   	+  \frac{{\rm Var}_{\beta_0,\pi(\gamma_*)}(h_\epsilon(Y,\beta_0,\gamma_*))  } {n} 
   	+ o(\epsilon) ,
   	\end{align*}
   	which is the statement of the lemma.
   \end{proof}

   \begin{proof}[\bf Proof of Lemma~\ref{lemma:EtaEstimatorEffect}]
   	Let $\eta=(\beta', \gamma')'$, $\widehat \eta:=(\widehat \beta', \widehat \gamma')'$, and
   	$\eta_*:=(\beta_0', \gamma_*')'$.
   	By a Taylor expansion in $\eta$ around $\eta_*$ we find that
   	\begin{align}    
   	\widehat \delta\,^{\rm MMSE}_\epsilon 
   	&= \delta_{\widehat \beta,\pi(\widehat \gamma)}+\frac 1 n \sum_{i=1}^n h_\epsilon^{\rm MMSE}(Y_i,\widehat \beta,\widehat \gamma)
   	\nonumber  \\
   	&=  \delta_{\beta_0,\pi(\gamma_*)}+\frac 1 n \sum_{i=1}^n h_\epsilon^{\rm MMSE}(Y_i,\beta_0,\gamma_*)
   	\nonumber   \\
   	& \quad
   	\underbrace{
   		\left( \widehat \eta - \eta_* \right)'  \left[ \nabla_\eta \delta_{\beta_0,\pi(\gamma_*)}+\mathbb{E}_{\beta_0,\pi(\gamma_*)}     \nabla_\eta  h^{\rm MMSE}_\epsilon(Y,  \beta_0,\gamma_*) \right]
   	}_{=r^{(1)}} 
   	\nonumber   \\
   	& \quad
   	+     \underbrace{  \left( \widehat \eta - \eta_* \right)'   \frac 1 n \sum_{i=1}^n  \left[   \nabla_\eta h^{\rm MMSE}_\epsilon(Y_i,\beta_0,\gamma_*) - \mathbb{E}_{\beta_0,\pi_0}     \nabla_\eta  h^{\rm MMSE}_\epsilon(Y_i,\beta_0,\gamma_*) \right]
   	}_{=r^{(2)}}
   	\nonumber   \\
   	& \quad
   	+     \underbrace{  \left( \widehat \eta - \eta_* \right)'   \left[\mathbb{E}_{\beta_0,\pi(\gamma_*)}   \nabla_\eta h^{\rm MMSE}_\epsilon(Y,  \beta_0,\gamma_*) - \mathbb{E}_{\beta_0,\pi_0}     \nabla_\eta  h^{\rm MMSE}_\epsilon(Y,  \beta_0,\gamma_*) \right]
   	}_{=r^{(3)}}
   	\nonumber   \\
   	& \quad      +  \frac{1}{2}  \underbrace{ \left( \widehat \eta - \eta_* \right)'  
   		\left[    \frac 1 n \sum_{i=1}^n   \nabla^2_{\eta \eta'} h^{\rm MMSE}_\epsilon(Y_i,  \widetilde \beta, \widetilde \gamma)  \right] \left( \widehat \eta - \eta_* \right) 
   	}_{=r^{(4)}} ,
   	\label{DefR1234}
   	\end{align}   
   	where $\widetilde  \eta=(\widetilde \beta', \widetilde \gamma')'$ is a value between $\widehat \eta$ and $ \eta_* $. 
   	Our  constraints \eqref{Con:Unbiased}
   	and \eqref{Con:EtaGradient} guarantee that
   	$\nabla_\eta\delta_{\beta_0,\pi(\gamma_*)}+\mathbb{E}_{\beta_0,\pi(\gamma_*)}     \nabla_\eta  h^{\rm MMSE}_\epsilon(Y,  \beta_0,\gamma_*) = 0$; that is,
   	we have $r^{(1)} =0$.
   	Using Assumption~\ref{ass:Expansion2} and the Cauchy-Schwarz inequality we furthermore find
   	\begin{align*}
   	&
   	\left(  \mathbb{E}_{\beta_0,\pi_0}    \left| r^{(2)}  
   	\right| \right)^2
   	\\
   	& \leq   
   	\mathbb{E}_{\beta_0,\pi_0}\left\| \widehat \eta - \eta_* \right\|^2
   	\; \mathbb{E}_{\beta_0,\pi_0}    \left\| \frac 1 n \sum_{i=1}^n  \left[   \nabla_\eta h^{\rm MMSE}_\epsilon(Y_i,\beta_0,\gamma_*) - \mathbb{E}_{\beta_0,\pi_0}     \nabla_\eta  h^{\rm MMSE}_\epsilon(Y_i,\beta_0,\gamma_*) \right] \right\|^2
   	\\
   	&\leq       \mathbb{E}_{\beta_0,\pi_0}\left\| \widehat \eta - \eta_* \right\|^2
   	\,  \frac 1 n \,   \mathbb{E}_{\beta_0,\pi_0}  \left\|  \nabla_\eta h^{\rm MMSE}_\epsilon(Y,  \beta_0,\gamma_*) \right\|^2
   	=  O\left( \frac 1 {n^2} \right),
   	\end{align*}
   	uniformly in $\pi_0\in\Gamma_{\epsilon}(\gamma_*)$,
   	where in the second step we have used the independence of $Y_i$ across $i$.
   	Similarly, we have
   	\begin{align*}
   	\left( \mathbb{E}_{\beta_0,\pi_0}    \left| r^{(3)}
   	\right|  \right)^2
   	& \leq    
   	\mathbb{E}_{\beta_0,\pi_0}\left\| \widehat \eta - \eta_* \right\|^2
   	\left\| \mathbb{E}_{\beta_0,\pi_0}   \nabla_\eta h^{\rm MMSE}_\epsilon(Y,  \beta_0,\gamma_*) - \mathbb{E}_{\beta_0,\pi(\gamma_*)}     \nabla_\eta  h^{\rm MMSE}_\epsilon(Y,  \beta_0,\gamma_*) \right\|^2
   	\\
   	&=  O\left( \frac 1 {n} \right)  O\left( \epsilon \right) =  O\left( \frac 1 {n^2} \right),
   	\end{align*}
   	uniformly in $\pi_0\in\Gamma_{\epsilon}(\gamma_*)$, where we have used that
   	$$
   	\sup_{\pi_0 \in \Gamma_{\epsilon}(\gamma_*)}  \left\| \mathbb{E}_{\beta_0,\pi_0}   \nabla_\eta h^{\rm MMSE}_\epsilon(Y,  \beta_0,\gamma_*) - \mathbb{E}_{\beta_0,\pi(\gamma_*)}     \nabla_\eta  h^{\rm MMSE}_\epsilon(Y,  \beta_0,\gamma_*) \right\| = O(\epsilon^{1/2}),
   	$$
   	which   follows from Assumptions~\ref{ass:Expansion}(iii)
   	and \ref{ass:Expansion2}(ii) by using the proof strategy of part (ii) of Lemma~\ref{lemma:Expansions}.
   	Finally, applying  H\"older's inequality we have
   	\begin{align*}
   	\mathbb{E}_{\beta_0,\pi_0}    \left| r^{(4)} 
   	\right|
   	& \leq     \mathbb{E}_{\beta_0,\pi_0}
   	\left[
   	\left\| \widehat \eta - \eta_* \right\|^2
   	\left\|   \frac 1 n \sum_{i=1}^n   \nabla^2_{\eta \eta'} h^{\rm MMSE}_\epsilon(Y_i,   \widetilde \beta, \widetilde \gamma)  \right\|
   	\right]
   	\\
   	& \leq   
   	\left\{  \mathbb{E}_{\beta_0,\pi_0}\left\| \widehat \eta - \eta_* \right\|^{\chi} \right\}^{\frac 2 {\chi}}
   	\left\{
   	\mathbb{E}_{\beta_0,\pi_0}
   	\left\|   \frac 1 n \sum_{i=1}^n   \nabla^2_{\eta \eta'} h^{\rm MMSE}_\epsilon(Y_i,   \widetilde \beta, \widetilde \gamma)  \right\|^{\frac {\chi} {\chi-2}}
   	\right\}^{\frac {\chi-2} {\chi}}
   	\\
   	&= O\left( \frac 1 n \right),
   	\end{align*}
   	uniformly in $\pi_0\in\Gamma_{\epsilon}(\gamma_*)$, where we have used Assumption~\ref{ass:Expansion2}(iii). We have thus shown that
   	\begin{align*}
   	\sup_{\pi_0 \in \Gamma_{\epsilon}(\gamma_*)}  \mathbb{E}_{\beta_0,\pi_0}  
   	\left|  r^{(1)}+ r^{(2)}+ r^{(3)}+ \frac 1 2 \, r^{(4)}\right|
   	= O\left( \frac 1 n \right),
   	\end{align*}
   	which together with \eqref{DefR1234} gives the statement of the lemma.
   \end{proof}

   The proof of the next lemma uses the following theorem of Petrov (1975), which generalizes the Berry-Esseen theorem
   to sample averages of random variables without a third moment.
   \begin{theorem}[Theorem 5 on p. 112 in Petrov 1975]
   	\label{th:Petrov1975}
   	Let $X_1,\ldots,X_n$ be independent random variables, such that
   	$\mathbb{E} X_j = 0$, $\mathbb{E}(X_j^2 g(|X_j|)) < \infty$ for $j=1,\ldots,n$, and for some function $g : [0,\infty) \rightarrow [0,\infty)$
   	such that both $g(x)$ and $x/g(x)$ are non-decreasing for $x>0$.
   	We write 
   	\begin{align*}
   	\sigma^2_j &= \mathbb{E} X_j^2 ,
   	&
   	B_n &= \sum_{j=1}^n \sigma^2_j ,
   	&
   	F_n(x) &= \Pr\left(  B_n^{-1/2} \sum_{j=1}^n X_j <x  \right).
   	\end{align*}
   	Then there exists an absolute constant $A>0$ such that
   	$$
   	\sup_x \left| F_n(x) - \Phi(x) \right| \leq \frac{A} {B_n g(\sqrt{B_n})} \sum_{j=1}^n \mathbb{E}\left( X_j^2 g(X_j) \right) .
   	$$
   	
   \end{theorem}

   \begin{proof}[\bf Proof of Lemma~\ref{lemma:trimmedMSEtrunc}]
   	\# \underline{Preliminaries:}   
   	We first establish some preliminary results on the sample averages of
   	$$
   	\widetilde h_{\epsilon}(Y_i, \beta_0, \gamma_*,\pi_0) 
   	:= h_{\epsilon}(Y_i, \beta_0, \gamma_*) 
   	-   \mathbb{E}_{\beta_0,\pi_0}   h_{\epsilon}(Y_i, \beta_0, \gamma_*) .
   	$$
   	According to our assumptions  the $ \widetilde h_{\epsilon}(Y_i, \beta_0, \gamma_*,\pi_0) $
   	are independent random variables with zero mean and finite absolute moments of order $\kappa>2$,
   	under $P_0 = P(\beta_0,\pi_0)$.
   	By applying the result in Dharmadhikari and Jogdeo (1969) we thus find that\footnote{
   		This result is an extension of the Bahr-Esseen inequality to moments larger than two. See also
   		inequality number 16 on p. 60 of Petrov (1975).
   	}
   	\begin{align*}
   	\mathbb{E}_{\beta_0,\pi_0} 
   	\left|  \frac 1 {\sqrt{n}} \sum_{i=1}^n  \widetilde h_{\epsilon}(Y_i, \beta_0, \gamma_*,\pi_0) 
   	\right|^{\kappa}  
   	\leq 
   	C_{\kappa} \,
   	\mathbb{E}_{\beta_0,\pi_0} \left|  \widetilde  h_{\epsilon}(Y_i, \beta_0, \gamma_*,\pi_0)   \right|^{\kappa} ,
   	\end{align*}
   	where the constant $C_{\kappa}>0$ only depends on $\kappa$.
   	Through a combination of the Minkowski and H\"older's inequalities
   	we find that our assumption
   	$\sup_{\pi_0 \in \Gamma_{\epsilon}(\gamma_*)}  \mathbb{E}_{\beta_0,\pi_0}  \left| h_{\epsilon}(Y, \beta_0, \gamma_*)  \right|^{\kappa} = O(1)$
   	also guarantees 
   	$\sup_{\pi_0 \in \Gamma_{\epsilon}(\gamma_*)}  \mathbb{E}_{\beta_0,\pi_0}  \left|  \widetilde  h_{\epsilon}(Y, \beta_0, \gamma_*)  \right|^{\kappa} = O(1)$.
   	We therefore obtain that
   	\begin{align}
   	\sup_{\pi_0\in \Gamma_\epsilon(\gamma_*)}  
   	\left( \mathbb{E}_{\beta_0,\pi_0} 
   	\left|  \frac 1 {\sqrt{n}} \sum_{i=1}^n  \widetilde h_{\epsilon}(Y_i, \beta_0, \gamma_*,\pi_0)   \right|^{\kappa} \right)^{\frac 1 {\kappa}}
   	= O(1) .
   	\label{TildeHresult1}
   	\end{align}
   	Next, we apply Theorem 5 of Chapter V in Petrov (1975), which is restated above as Theorem~\ref{th:Petrov1975},
   	with $X_i$ equal to $\widetilde h_{\epsilon}(Y_i, \beta_0, \gamma_*,\pi_0) $ 
   	and $g(x) = x^{\min\{1,\kappa-2\}}$ to find that
   	\begin{align*}
   	\sup_{\pi_0\in \Gamma_\epsilon(\gamma_*)}
   	\sup_{x \in \mathbb{R}}
   	\left| 
   	{\rm P}_{\beta_0,\pi_0} \left( 
   	\frac { \sum_{i=1}^n  \widetilde h_{\epsilon}(Y_i, \beta_0, \gamma_*,\pi_0)  } {\sqrt{n} \,  \sigma(\beta_0, \gamma_*,\pi_0)  }
   	\leq  x  \right)
   	- \Phi(x) \right| = o(1),
   	\end{align*}
   	where $\sigma^2(\beta_0, \gamma_*,\pi_0) =  \mathbb{E}_{\beta_0,\pi_0}  \widetilde h^2_{\epsilon}(Y_i, \beta_0, \gamma_*,\pi_0)  $. This, in particular, implies that
   	\begin{align}
   	\sup_{\pi_0\in \Gamma_\epsilon(\gamma_*)}
   	{\rm P}_{\beta_0,\pi_0} \left(     \left|  \frac 1 {\sqrt{n}} \sum_{i=1}^n 
   	\widetilde h_{\epsilon}(Y_i, \beta_0, \gamma_*,\pi_0)    \right| >   \log(n) \right) = o(1) .
   	\label{TildeHresult2}
   	\end{align}
   	By an application of H\"older's inequality we find that \eqref{TildeHresult1} and \eqref{TildeHresult2} also imply 
   	\begin{align}      
   	\sup_{\pi_0\in \Gamma_\epsilon(\gamma_*)}  
   	\mathbb{E}_{\beta_0,\pi_0}  \left[  
   	\left(   \frac 1 {\sqrt{n}} \sum_{i=1}^n  \widetilde h_{\epsilon}(Y_i, \beta_0, \gamma_*)   \right)^2
   	\mathbbm{1}\left(   \left|  \frac 1 {\sqrt{n}} \sum_{i=1}^n  \widetilde h_{\epsilon}(Y_i, \beta_0, \gamma_*)   \right|  >       \log n  \right) 
   	\right] 
   	= o(1).
   	\label{TildeHresult3}
   	\end{align}	
   	Finally, we notice that
   	\begin{align}
   	\sup_{\pi_0\in \Gamma_\epsilon(\gamma_*)} 
   	\left|  \delta_{\beta_0,\pi(\gamma_*)}  - \delta_{\beta_0,\pi_0}
   	+  \mathbb{E}_{\beta_0,\pi_0}  h_{\epsilon}(Y, \beta_0, \gamma_*) 
   	\right| = O(\epsilon^{1/2})	,
   	\label{AnotherPreliminaryBound}
   	\end{align}   
   	which follows by applying part (i) and (ii) of Lemma~\ref{lemma:Expansions}
   	with $q_\epsilon(y)=  h_{\epsilon}(y, \beta_0, \gamma_*) $
   	and noting that $  \mathbb{E}_{\beta_0,\pi(\gamma_*)}  h_{\epsilon}(Y, \beta_0, \gamma_*) = 0 $
   	by  the unbiasedness constraint \eqref{Con:Unbiased}.

   	\# \underline{Main result of the Lemma~\ref{lemma:trimmedMSEtrunc}:}      
   	Having established those preliminary results, we now derive the statement of the lemma.   
   	Define 
   	\begin{align*}
   	k_n &:= \frac 1 {\sqrt{n}} \sum_{i=1}^n  h_{\epsilon}(Y_i, \beta_0, \gamma_*)
   	+ \sqrt{n} \left[ \delta_{\beta_0,\pi(\gamma_*)}  -  \delta_{\beta_0,\pi_0}  \right] 
   	\\ &=
   	\frac 1 {\sqrt{n}} \sum_{i=1}^n  \widetilde h_{\epsilon}(Y_i, \beta_0, \gamma_*)
   	+ \sqrt{n} \left[ \delta_{\beta_0,\pi(\gamma_*)}  -  \delta_{\beta_0,\pi_0}   +  \mathbb{E}_{\beta_0,\pi_0}  h_{\epsilon}(Y_i, \beta_0, \gamma_*)   \right]. 
   	\end{align*}
   	The decomposition of $\widehat \delta_{\epsilon} $ in \eqref{DefRemainder1app}
   	can then be rewritten as
   	\begin{align*}
   	\sqrt{n} \, \left(   \widehat \delta_{\epsilon}  -    \delta_{\beta_0,\pi_0}  \right)
   	&=   
   	k_n + R_n .
   	\end{align*} 
   	We have
   	\begin{align*}
   	&  n \,   \mathbb{E}_{\beta_0,\pi_0}  \left[ \left( \widehat \delta_{\epsilon}  - \delta_{\beta_0,\pi_0} \right)^2
   	\mathbbm{1}\left(
   	\left| \widehat \delta_{\epsilon}  - \delta_{\beta_0,\pi_0} \right| \leq m_n		 
   	\right) \right] 
   	\\
   	&\quad \qquad = \mathbb{E}_{\beta_0,\pi_0}  \left[ \left( k_n + R_n \right)^2
   	\mathbbm{1}\left(
   	\left|k_n + R_n \right| \leq n^{1/2} \, m_n		 
   	\right) \right] 		  
   	\\
   	&\quad  \qquad  =
   	\mathbb{E}_{\beta_0,\pi_0}  k_n^2
   	-
   	\underbrace{ \mathbb{E}_{\beta_0,\pi_0}  \left[   k_n^2 \,
   		\mathbbm{1}\left(
   		\left|k_n + R_n \right| > n^{1/2} \, m_n		 
   		\right) \right] 	}_{= \text{term I}}
   	\\[10pt]
   	&\qquad  \qquad   \qquad  \qquad \qquad  \qquad   \qquad  \qquad 
   	+ \underbrace{  \mathbb{E}_{\beta_0,\pi_0}  \left[  (R_n^2  + 2 k_n \, R_n  )
   		\mathbbm{1}\left(
   		\left|k_n + R_n \right| \leq n^{1/2} \, m_n		 
   		\right) \right] 			}_{= \text{term II}}.
   	\end{align*}
   	Thus, Lemma~\ref{lemma:trimmedMSEtrunc} is proved if we can show that term I is $o(1)$,
   	and that term II is larger or equal to minus $o(1)$,
   	both uniformly over $\pi_0\in \Gamma_\epsilon(\gamma_*)$.
   	For term I we use H\"older's inequality to obtain that
   	\begin{align*}
   	& \sup_{\pi_0\in \Gamma_\epsilon(\gamma_*)}  \mathbb{E}_{\beta_0,\pi_0}  \left[   k_n^2 \,
   	\mathbbm{1}\left(
   	\left|k_n + R_n \right| > n^{1/2} \, m_n		 
   	\right) \right] 
   	\\	  
   	& \leq 	
   	\left\{  \sup_{\pi_0\in \Gamma_\epsilon(\gamma_*)} 
   	\left( \mathbb{E}_{\beta_0,\pi_0}  \left| k_n \right|^{\kappa} \right)^{\frac 2 {\kappa}}	
   	\right\}
   	\left\{ \sup_{\pi_0\in \Gamma_\epsilon(\gamma_*)} 
   	\left[ \mathbb{E}_{\beta_0,\pi_0}    \mathbbm{1}\left(
   	\left|k_n + R_n \right| > n^{1/2} \, m_n			 
   	\right) \right]^{\frac {\kappa-2} {\kappa}}
   	\right\}
   	\\	  
   	& \leq 	
   	\Bigg\{ 
   	\underbrace{
   		\sup_{\pi_0\in \Gamma_\epsilon(\gamma_*)} 
   		\left( \mathbb{E}_{\beta_0,\pi_0}  \left| \frac 1 {\sqrt{n}} \sum_{i=1}^n  \widetilde h_{\epsilon}(Y_i, \beta_0, \gamma_*) \right|^{\kappa} \right)^{\frac 2 {\kappa}}
   	}_{=O(1)}
   	\\ & \qquad \qquad   \qquad \qquad  \qquad \qquad  \quad +
   	\underbrace{
   		\sup_{\pi_0\in \Gamma_\epsilon(\gamma_*)} 
   		\left( n^{1/2}  \left|  \delta_{\beta_0,\pi(\gamma_*)}  -  \delta_{\beta_0,\pi_0}    +  \mathbb{E}_{\beta_0,\pi_0}  h_{\epsilon}(Y_i, \beta_0, \gamma_*)    \right|  \right) 	
   	}_{=O(1)}
   	\Bigg\}
   	\\ & \quad \times  
   	\Bigg\{ 
   	\bigg[ \underbrace{
   		\sup_{\pi_0\in \Gamma_\epsilon(\gamma_*)}  \mathbb{E}_{\beta_0,\pi_0}    \mathbbm{1}\left(
   		\left|k_n   \right| > \frac 1 2 \, n^{1/2} \, m_n			 
   		\right) }_{=o(1)}
   	+ 
   	\underbrace{
   		\sup_{\pi_0\in \Gamma_\epsilon(\gamma_*)}  \mathbb{E}_{\beta_0,\pi_0}    \mathbbm{1}\left(
   		\left|R_n   \right| > \frac 1 2 \, n^{1/2} \, m_n			 
   		\right)		  
   	}_{=o(1)} 
   	\bigg]^{\frac {\kappa-2} {\kappa}}
   	\Bigg\}
   	\\
   	&= o(1) ,		  
   	\end{align*}
   	where we also used the definition of $k_n$ together with the 
   	triangle inequality,
   	and we employed 
   	\eqref{TildeHresult1}, \eqref{TildeHresult2} and \eqref{AnotherPreliminaryBound}
   	and
   	Assumption (ii) of the lemma, together with our assumption that
   	$n^{1/2} \, m_n \gg \log(n)$ as $n \rightarrow \infty$.
   	
   	Next, for term II we use that $R_n^2  + 2 k_n \, R_n $ is positive whenever $ \left|R_n   \right| 	>   2 \left|k_n   \right| 	$
   	to obtain that
   	\begin{align*}
   	& \mathbb{E}_{\beta_0,\pi_0}  \left[  (R_n^2  + 2 k_n \, R_n  )
   	\mathbbm{1}\left(
   	\left|k_n + R_n \right| \leq n^{1/2} \, m_n		 
   	\right) \right] 
   	\\
   	& \quad =  \mathbb{E}_{\beta_0,\pi_0}  \left[  (R_n^2  + 2 k_n \, R_n  )
   	\mathbbm{1}\left(
   	\left|k_n + R_n \right| \leq n^{1/2} \, m_n		 
   	\right) 
   	\mathbbm{1}\left(  \left|R_n   \right| 	  \leq   2 \left|k_n   \right| 	   \right)
   	\right] 
   	\\ & \qquad		  
   	+
   	\underbrace{
   		\mathbb{E}_{\beta_0,\pi_0}  \left[  (R_n^2  + 2 k_n \, R_n  )
   		\mathbbm{1}\left(
   		\left|k_n + R_n \right| \leq n^{1/2} \, m_n		 
   		\right)
   		\mathbbm{1}\left(  \left|R_n   \right| 	>   2 \left|k_n   \right| 	   \right)
   		\right] 	
   	}_{ \geq 0} 	  
   	\\
   	& \quad \geq  \mathbb{E}_{\beta_0,\pi_0}  \left[  (R_n^2  + 2 k_n \, R_n  )
   	\mathbbm{1}\left(
   	\left|k_n + R_n \right| \leq n^{1/2} \, m_n		 
   	\right) 
   	\mathbbm{1}\left(  \left|R_n   \right| 	  \leq   2 \left|k_n   \right| 	   \right)
   	\right] 	
   	\\
   	& \quad \geq 
   	- 2 \, \mathbb{E}_{\beta_0,\pi_0}  \Big[   | k_n|  \;  | R_n|   \;
   	\mathbbm{1}\left(  \left|R_n   \right| 	  \leq   2 \left|k_n   \right| 	   \right)
   	\Big] 	  	  
   	\\
   	& \quad \geq 
   	- 2 \, \left\{ \mathbb{E}_{\beta_0,\pi_0} \,   k_n^2  \right\}^{1/2} 	  
   	\left\{    \mathbb{E}_{\beta_0,\pi_0}  \left[  R_n^2  \,
   	\mathbbm{1}\left(  \left|R_n   \right| 	  \leq   2 \left|k_n   \right| 	   \right)
   	\right] 	  
   	\right\}^{1/2}
   	\end{align*}
   	where in the last step we also used the Cauchy-Schwarz inequality.
   	Our preliminary results \eqref{TildeHresult1} and \eqref{AnotherPreliminaryBound} imply that
   	$  \sup_{\pi_0\in \Gamma_\epsilon(\gamma_*)}  \mathbb{E}_{\beta_0,\pi_0} \,   k_n^2 = O(1)$.
   	Furthermore we have
   	\begin{align*}
   	& \mathbb{E}_{\beta_0,\pi_0}  \left[  R_n^2 \,
   	\mathbbm{1}\left(  |R_n| \leq  2 \, |k_n| \right) 
   	\right] 
   	\\ & \qquad		  
   	= 	    \mathbb{E}_{\beta_0,\pi_0}  \left[  R_n^2 \,
   	\mathbbm{1}\left(  |R_n| \leq  2 \, |k_n| \right) 
   	\mathbbm{1}\left(  |k_n| \leq      \log n  \right) 
   	\right] 
   	\\ &\qquad \quad
   	+  \mathbb{E}_{\beta_0,\pi_0}  \left[  R_n^2 \,
   	\mathbbm{1}\left(  |R_n| \leq  2 \, |k_n| \right) 
   	\mathbbm{1}\left(  |k_n|  >      \log n  \right) 
   	\right] 	  
   	\\ &\qquad \leq	   \mathbb{E}_{\beta_0,\pi_0}  \left[  R_n^2 \,
   	\mathbbm{1}\left(  |R_n| \leq  2 \,    \log n  \right) 
   	\right] 	
   	+  4 \, \mathbb{E}_{\beta_0,\pi_0}  \left[  k_n^2 \,
   	\mathbbm{1}\left(  |k_n| >      \log n  \right) 
   	\right]	
   	\\ &\qquad = o(1),		  	  
   	\end{align*}
   	uniformly over $\pi_0\in \Gamma_\epsilon(\gamma_*)$, where we used \eqref{TildeHresult3} and Assumption (v)
   	of the lemma.
   	We thus conclude that term II indeed satisfies
   	$$
   	\sup_{\pi_0\in \Gamma_\epsilon(\gamma_*)}  
   	\left\{ -  \mathbb{E}_{\beta_0,\pi_0}  \left[  (R_n^2  + 2 k_n \, R_n  )
   	\mathbbm{1}\left(
   	\left|k_n + R_n \right| \leq n^{1/2} \, m_n		 
   	\right) \right]  \right\}
   	\leq    o(1)	  .
   	$$
   	Combining the above gives the statement of the lemma.
   \end{proof}

 \subsection{Lemma~\ref{lemma:hMMSEmoment}}
 
 \paragraph{Notation.}
 
 For the proof of Lemma~\ref{lemma:hMMSEmoment} (which assumes the locally quadratic case of Section \ref{Sec_param}) it is convenient to introduce some further notation. We assume that there exists a map $\Omega_{\gamma_*} : \overline {\cal T} \rightarrow {\cal T}$
 such that, for all $v \in \overline {\cal T}$,
 \begin{align*}
 	\| v \|^2_{{\rm ind},\gamma_*}  = \left\langle  v , \Omega_{\gamma_*}   v   \right\rangle .
 \end{align*}
 We assume that $\Omega_{\gamma_*}$ is invertible, and write $\Omega_{\gamma_*}^{-1}: {\cal T} \rightarrow \overline {\cal T}$
 for its inverse.  The map $\Omega_{\gamma_*}^{-1}$ is exactly the ``transposition'' map introduced
 less formally in the main text; that is, 
 for $u \in {\cal T}$ we have $u^{\top} =   \Omega_{\gamma_*}^{-1} \, u \in \overline {\cal T}$.
 Thus, our norm on the cotangent space from the main text
 $ \| u \|^2_{\gamma_*}  = u^{\top} u$ can now be written as
 \begin{align*}
 	\| u \|^2_{\gamma_*}  = \left\langle  \Omega_{\gamma_*}^{-1} \, u , u   \right\rangle .
 \end{align*}
 The norm $ \left\|  \cdot  \right\|_{\gamma_*}$
 is dual to $\|  \cdot  \|_{{\rm ind},\gamma_*}$; that is, we have
 \begin{align*}
 	\left\|  u  \right\|_{\gamma_*}
 	&=
 	\sup_{ v \in \overline {\cal T} \setminus \{0\}} \;   \frac{ \langle v,u \rangle  } {  \| v \|_{{\rm ind},\gamma_*} }  .
 \end{align*}
 Notice also that $ \| \cdot \|_{{\rm ind},\gamma_*} $, $ \| \cdot \|_{\gamma_*} $, $\Omega_{\gamma_*}$, and $\Omega_{\gamma_*}^{-1}$ could all be defined
 for general $\pi \in \Pi$, but since we use them only at the reference value $\pi(\gamma_*)$ we index them simply by $\gamma_*$.

 The vector norms  $ \| \cdot \|_{{\rm ind},\gamma_*} $, $ \| \cdot \|_{\gamma_*} $ and $\|.\|$
 on $\overline {\cal T}$, ${\cal T}$ and $\mathbb{R}^{\dim \beta + \dim \gamma}$ induce natural norms on 
 any maps between $\overline {\cal T}$, ${\cal T}$ and $\mathbb{R}^{\dim \beta + \dim \gamma}$.
 With a slight abuse of notation
 we denote all those norms simply by $\|.\|_{\gamma_*}$. 
 In particular,
 for $\Omega_{\gamma_*}^{-1}: {\cal T} \rightarrow \overline {\cal T}$ we have
 \begin{align}
 	\left\| \Omega_{\gamma_*}^{-1}  \right\|_{\gamma_*}
 	:=  \sup_{u \in   {\cal T} \setminus \{0\}}  \frac{ \|   \Omega_{\gamma_*}^{-1}  \, u   \|_{{\rm ind},\gamma_*} } {\| u\|_{\gamma_*}}
 	=   \sup_{u \in   {\cal T} \setminus \{0\}}  \frac{\left\langle   \Omega_{\gamma_*}^{-1}  \, u  ,  , u    \right\rangle^{1/2} } {\| u\|_{\gamma_*}}
 	= 1 ,
 	\label{propertyMatrixNorm1}
 \end{align}
 and for $H_{\pi,\beta\gamma}:\mathbb{R}^{\limfunc{dim}\beta+\limfunc{dim}\gamma}\rightarrow{\cal{T}}$
 defined in Section~\ref{sec:LocallyQuadratic}
 we have
 \begin{align*}
 	\left\| H_{\pi,\beta\gamma}  \right\|_{\gamma_*}
 	:=  \sup_{w \in \mathbb{R}^{\dim \beta + \dim \gamma} \setminus \{0\}}  \frac{ \|  H_{\pi,\beta\gamma}  w   \|_{\gamma_*} } {\| w\| }
 	=  \sup_{v \in \overline {\cal T} \setminus \{0\}}  \sup_{w \in \mathbb{R}^{\dim \beta + \dim \gamma} \setminus \{0\}}  
 	\frac{  \left\langle v  , H_{\pi,\beta\gamma}  w \right\rangle } {\| v\|_{{\rm ind},\gamma_*} \, \| w\|  }.
 \end{align*}
 Using Assumption~\ref{ass:Expansion}(v) and the Cauchy-Schwarz inequality we find that
 \begin{align}
 	\left\|  H_{\pi,\beta\gamma}     \right\|_{\gamma_*} 
 	&=   \left\| \mathbb{E}_{\beta_0,\pi(\gamma_*)} \left\{  \left[ \nabla_{\pi} \log f_{\beta_0,\pi(\gamma_*)}(Y) \right]
 	\left[ \nabla_{\beta\gamma} \log f_{\beta_0,\pi(\gamma_*)}(Y) \right]'  \right\}  \right\|_{\gamma_*} 
 	\nonumber \\
 	&\leq   \left[  \mathbb{E}_{\beta_0,\pi(\gamma_*)}    \left\|     \nabla_{\pi} \log f_{\beta_0,\pi(\gamma_*)}(Y)  	 \right\|_{\gamma_*}^{2} \right]^{1/2}
 	\left[ \mathbb{E}_{\beta_0,\pi(\gamma_*)}    \left\|   \nabla_{\beta\gamma} \log f_{\beta_0,\pi(\gamma_*)}(Y)  	 \right\|^{2}  \right]^{1/2}
 	\nonumber   \\
 	&= O(1) .
 	\label{propertyMatrixNorm2}
 \end{align}

 \begin{proof}[\bf Proof of Lemma~\ref{lemma:hMMSEmoment}]
 	Equation \eqref{SolutionMMSE} in Lemma~\ref{lem_locquad} in the main text provides an explicit solution for 
 	$ h^{\rm MMSE}_{\epsilon}(y,\beta_0,\gamma_*) $, which in the notation of this appendix can be
 	written as
 	\begin{align*}
 		h^{\rm MMSE}_{\epsilon}(y,\beta_0,\gamma_*) 
 		&=  \left[ \nabla_{\beta\gamma}  \delta_{\beta_0,\pi(\gamma_*)} \right]'
 		H_{\beta\gamma} ^{-1}        \left[ \nabla_{\beta\gamma} \log f_{\beta_0,\pi(\gamma_*)}(y) \right]	
 		\\
 		& \quad +     \left\langle  
 		\left[  \widetilde H_{{\pi}} \,  \Omega_{\gamma_*} + (\epsilon n)^{-1}  \Omega_{\gamma_*} \right]^{-1}
 		\, \widetilde \nabla_{\pi} \delta_{\beta_0,\pi(\gamma_*)} 
 		\, , \,  \widetilde \nabla_{\pi} \log f_{\beta_0,\pi(\gamma_*)}(y) 
 		\,   \right\rangle ,
 	\end{align*}	
 	where 
 	$  \widetilde \nabla_{\pi} \log f_{\beta_0,\pi(\gamma_*)}(y) 
 	=   \nabla_{\pi} \log f_{\beta_0,\pi(\gamma_*)}(y) 
 	-  H_{\pi,\beta\gamma}    H_{\beta\gamma}^{-1}   \,  \nabla_{\beta\gamma} \log f_{\beta_0,\pi(\gamma_*)}(y) $
 	and
 	$ \widetilde \nabla_{\pi} \delta_{\beta_0,\pi(\gamma_*)}
 	= \allowbreak \nabla_{\pi} \delta_{\beta_0,\pi(\gamma_*)} \allowbreak - H_{\pi,\beta\gamma}    H_{\beta\gamma}^{-1} \,  \nabla_{\beta\gamma} \delta_{\beta_0,\pi(\gamma_*)} $.
 	We thus have
 	\begin{align*}
 		\left|  h^{\rm MMSE}_{\epsilon}(y,\beta_0,\gamma_*)  \right|
 		& \leq   \left\| \nabla_{\beta\gamma}  \delta_{\beta_0,\pi(\gamma_*)} \right\|
 		\left\|   H_{\beta\gamma} ^{-1}  \right\|
 		\left\|    \left[ \nabla_{\beta\gamma} \log f_{\beta_0,\pi(\gamma_*)}(y) \right]	 \right\|
 		\\
 		& \quad +  
 		(\epsilon n)  \left\|   \Omega_{\gamma_*}^{-1} \right\|_{\gamma_*}
 		\left\| \widetilde \nabla_{\pi} \delta_{\beta_0,\pi(\gamma_*)}   \right\|_{\gamma_*}
 		\left\|  \widetilde \nabla_{\pi}  \, \log f_{\beta_0,\pi(\gamma_*)}(y)  \right\|_{\gamma_*} ,
 	\end{align*}
 	where we used that
 	$  \left\|  \left[  \widetilde H_{{\pi}} \,  \Omega_{\gamma_*} + (\epsilon n)^{-1}  \Omega_{\gamma_*} \right]^{-1}  \right\|_{\gamma_*} \leq
 	(\epsilon n)  \left\|   \Omega_{\gamma_*}^{-1} \right\|_{\gamma_*}$,
 	because both 	  $ \widetilde H_{{\pi}} \,  \Omega_{\gamma_*}$ and $\Omega_{\gamma_*} $ 
 	are positive semi-definite. We furthermore have
 	\begin{align*}
 		\left\| \widetilde \nabla_{\pi} \delta_{\beta_0,\pi(\gamma_*)}   \right\|_{\gamma_*}
 		&\leq   \left\|  \nabla_{\pi} \log f_{\beta_0,\pi(\gamma_*)}(y)  \right\|_{\gamma_*}
 		+  \left\|  H_{\pi,\beta\gamma}     \right\|_{\gamma_*}
 		\left\| H_{\beta\gamma}^{-1}   \right\|
 		\left\|  \nabla_{\beta\gamma} \log f_{\beta_0,\pi(\gamma_*)}(y)   \right\|
 		,	\\
 		\left\|  \widetilde \nabla_{\pi}  \, \log f_{\beta_0,\pi(\gamma_*)}(y)  \right\|_{\gamma_*} 
 		&\leq  
 		\left\| \allowbreak \nabla_{\pi} \delta_{\beta_0,\pi(\gamma_*)}  \right\|_{\gamma_*}
 		+  \left\| H_{\pi,\beta\gamma}     \right\|_{\gamma_*}
 		\left\|  H_{\beta\gamma}^{-1}   \right\|
 		\left\|   \nabla_{\beta\gamma} \delta_{\beta_0,\pi(\gamma_*)}	  \right\| .
 	\end{align*}
 	Combining those inequalities with our
 	Assumption~\ref{ass:Expansion}(ii) and (v) as well as the results \eqref{propertyMatrixNorm1} and \eqref{propertyMatrixNorm2} above
 	we   find that
 	\begin{align*}
 		\sup_{\pi_0 \in \Gamma_{\epsilon}(\gamma_*)}  & \mathbb{E}_{\beta_0,\pi_0}  \left[ h^{\rm MMSE}_{\epsilon}(Y,\beta_0,\gamma_*)  \right]^{2+\nu}
 		= O(1) .
 	\end{align*}          
 \end{proof} 
 
 \subsection{Lemma \ref{lem_locquad}}
 
 Before deriving the equivalent characterizations of $h_{\epsilon}^{\rm MMSE}(y,\beta_0,\gamma_*) $ given in the lemma
 we note that the optimization problem \eqref{MMSEproblem} that defines $h_{\epsilon}^{\rm MMSE}(y,\beta_0,\gamma_*) $  has a unique
 solution (up to possible deviations on a measure zero set of $y$'s, which are irrelevant for our purposes). This uniqueness follows,
 because under
 the unbiasedness constraint \eqref{Con:Unbiased}, we have
 ${\rm Var}_{\beta_0,\pi(\gamma_*)}(h(Y,\beta_0,\gamma_*) )= \mathbb{E}_{\beta_0,\pi(\gamma_*)} h^2(Y,\beta_0,\gamma_*)$,
 which is quadratic and strictly convex in $h(y,\beta_0,\gamma_*) $, while all other components of
 the objective function and constraints in \eqref{MMSEproblem} are linear in $h(y,\beta_0,\gamma_*) $.
 
 \paragraph{Equation (\ref{SolutionMMSE_linsys}).}
 
 Using simplified notation here,
 our goal is to find the function $h(y) = h(y,\beta_0,\gamma_*) $ that minimizes
 \begin{align*}
 \mathbb{E}  h^2(Y)   +
(\epsilon n) \, \left\{   \nabla_\pi \delta   -  \mathbb{E} \left[ h(Y)  s_\pi (Y) \right] \right\}^\top
\left\{  \nabla_\pi \delta   -  \mathbb{E} \left[ h(Y)s_\pi (Y) \right] \right\} ,
 \end{align*}
 subject to the constraints $ \mathbb{E}  h(Y) =0 $ and
 $\mathbb{E}  \,  h(Y)    s_{\beta\gamma} (Y) =  \nabla_{\beta\gamma}  \delta$. 
 
 Using the latter constraint
 and the definition of $\widetilde \nabla_{\pi}  $ 
 we can equivalently rewrite the objective function as
 \begin{align*}
 &   \mathbb{E}  h^2(Y)   +
 (\epsilon n) \, \left\{  \widetilde \nabla_\pi \delta   -  \mathbb{E} \left[ h(Y)\widetilde s_\pi (Y) \right] \right\}^\top
\left\{ \widetilde \nabla_\pi \delta   -  \mathbb{E} \left[ h(Y) \widetilde s_\pi (Y) \right]
 \right\} 
 \\
 & \qquad    \qquad    \qquad   \qquad    \qquad    \qquad   
 +2 \left\{ \nabla_{\beta\gamma}  \delta   -  \mathbb{E}  \left[  h(Y)     s_{\beta\gamma} (Y) \right] \right\}' 
 H_{\beta\gamma}^{-1}     \, \nabla_{\beta\gamma}  \delta  .
 \end{align*}
 
 The unconstrained minimizer of this rewritten quadratic objective function satisfies the first-order condition
 \begin{align*}
 h^{\rm MMSE}_\epsilon(y) =  s_{\beta\gamma} (y) '  H_{\beta\gamma}^{-1}      \nabla_{\beta\gamma}  \delta
 + (\epsilon n) \,
  \widetilde s_\pi (y)^\top
\left\{ \widetilde \nabla_\pi \delta   -  \mathbb{E} \left[ h^{\rm MMSE}_\epsilon(Y)  \widetilde s_\pi (Y) \right]
 \right\} ,
 \end{align*}
 and because 
 $ \mathbb{E}   s_{\beta\gamma}(Y) = 0$, $ \mathbb{E}   \widetilde s_\pi (Y) = 0$,
 and  
 $\mathbb{E}  [ s_{\beta\gamma} (Y)   s_{\beta\gamma} (Y) '  ]= H_{\beta\gamma}  $,
 we find that this unconstrained minimizer already satisfies both constraints 
 $ \mathbb{E}  h(Y) =0 $ and
 $\mathbb{E}   h(Y)    s_{\beta\gamma} (Y) =  \nabla_{\beta\gamma}  \delta$,
 and is therefore also the constrained minimizer that we wanted to derive.
 
 \paragraph{Equation (\ref{SolutionMMSE_rewrite}).}

 Note that, by \eqref{SolutionMMSE_linsys}, we have
 $h_{\epsilon}^{\rm MMSE}(y) =    s_{\beta\gamma} (y) '     H_{\beta\gamma}^{-1}    \nabla_{\beta\gamma}  \delta
 +   \widetilde s_{\pi} (y)^\top  \,  u $,
 for some $u \in {\cal T}$, and one can easily verify that this implies that 
 $ \widetilde \nabla_{\pi} \delta -\mathbb{E}\left[h_{\epsilon}^{\rm MMSE}(Y) \widetilde s_{\pi} (Y) \right]$ is equal to the same expression with $\widetilde s_{\pi}$ replaced by $s_{\pi}$.
 
 \paragraph{Equation (\ref{SolutionMMSE}).}
 We have already shown that equation \eqref{SolutionMMSE_linsys} is the FOC of the minimization problem \eqref{MMSEproblem}. 
 We now want to show that the solution for  $  h_{\epsilon}^{\rm MMSE}(y)$  given in equation
\eqref{SolutionMMSE} satisfies the FOC \eqref{SolutionMMSE_linsys}, which implies that it solves  \eqref{MMSEproblem}.
Equation \eqref{SolutionMMSE_linsys} can be rewritten as
 \begin{align}
h_{\epsilon}^{\rm MMSE}(y)
&=    s_{\beta\gamma}(y) '     H_{\beta\gamma}^{-1}     
\, \nabla_{\beta\gamma}  \delta +  (\epsilon n)    \,  \widetilde s_{\pi}(y)^\top \, u ,
&
u &:=   \widetilde \nabla_{\pi} \delta -\mathbb{E}\left[h_{\epsilon}^{\rm MMSE}(Y) \widetilde s_{\pi} (Y)  
\right]   .
   \label{Rewrite_SolutionMMSE_linsys}
\end{align}
Plugging the expression for  $  h_{\epsilon}^{\rm MMSE}(y)$ given  by equation
\eqref{SolutionMMSE} into this definition of $u$ and using that
$\mathbb{E}\left[\widetilde s_{\pi} (Y)  \widetilde s_{\pi} (Y)^\top   \right] =  \widetilde H_{{\pi}}  $,
and $\mathbb{E}\left[  \widetilde s_{\pi} (Y) s_{\beta\gamma} (Y)'     \right] = 0 $, we find that
\eqref{SolutionMMSE} implies that
\begin{align*}
    u &=     \widetilde \nabla_{\pi} \delta - 
     \widetilde H_{{\pi}}   \left[  \widetilde H_{{\pi}}  + (\epsilon n)^{-1}  \mathbb{I} \right]^{-1}
	\, \widetilde \nabla_{\pi} \delta 
    \\
     &=  \left\{	  \mathbb{I} -  \widetilde H_{{\pi}}   \left[  \widetilde H_{{\pi}}  + (\epsilon n)^{-1}  \mathbb{I} \right]^{-1} \right\}
	\, \widetilde \nabla_{\pi} \delta 
    \\
     &=  \left\{	  \left[  \widetilde H_{{\pi}}  + (\epsilon n)^{-1}  \mathbb{I} \right]
      \left[  \widetilde H_{{\pi}}  + (\epsilon n)^{-1}  \mathbb{I} \right]^{-1}
       -  \widetilde H_{{\pi}}   \left[  \widetilde H_{{\pi}}  + (\epsilon n)^{-1}  \mathbb{I} \right]^{-1} \right\}
	\, \widetilde \nabla_{\pi} \delta 
    \\
    &= 	 (\epsilon n)^{-1}   \left[  \widetilde H_{{\pi}}  + (\epsilon n)^{-1}  \mathbb{I} \right]^{-1}  \, \widetilde \nabla_{\pi} \delta .
\end{align*}
This expression for $u$ makes the first equation in  \eqref{Rewrite_SolutionMMSE_linsys} equivalent to \eqref{SolutionMMSE}.
Therefore, we have shown that    $  h_{\epsilon}^{\rm MMSE}(y)$ as given  by  
\eqref{SolutionMMSE} indeed solves  \eqref{SolutionMMSE_linsys}, and therefore also our optimization problem in \eqref{MMSEproblem}.

 \subsection{Lemma \ref{lem_locquad_cov}}
 
 Our goal is to choose the function $h(\cdot,\cdot,\beta,\gamma,f_X)$ such that
 the worst-case mean squared error
 $$  \sup_{ \pi_0 \in  \Gamma_\epsilon(\gamma_*)} \mathbb{E}_{\beta_0,\pi_0,f_X}\left[ \left( \widehat{\delta}_h  - \delta_{\beta_0,\pi_0,f_X}  \right)^2 \right]
 $$
 is minimized for small values of $\epsilon$, subject to unbiasedness under the reference model,
 and also subject to local robustness constraints to account for the fact that $\beta_0$, $\gamma_*$ and $f_X$ are estimated
 from the sample.

  Unbiasedness is
 \begin{align}
 \mathbb{E}_{f_X}\,\mathbb{E}_{\beta_0,\pi(\gamma_*)} \, h(Y,X,\beta_0,\gamma_*,f_X) = 0 , 
 \label{Con:Unbiased_cov}
 \end{align} 
while local robustness is
 \begin{align} 
 \begin{array}{l}
 \displaystyle
 \mathbb{E}_{f_X}\,\mathbb{E}_{\beta_0,\pi(\gamma_*)}  \,  h(Y, X, \beta_0,\gamma_*,f_X)    \, \nabla_{\beta\gamma} \log f_{\beta_0,\pi(\gamma_*)}(Y\,|\, X)
 =  \mathbb{E}_{f_X}\,\nabla_{\beta\gamma}\delta_{\beta_0,\pi(\gamma_*)}(X),
 \\
 [5pt]
 \displaystyle
 \mathbb{E}_{\beta_0,\pi(\gamma_*)}  \, \left[ h(Y, X, \beta_0,\gamma_*,f_X)   \,|\, X=x\right]
 =  \delta_{\beta_0,\pi(\gamma_*)}(x)-\mathbb{E}_{f_X}\,\delta_{\beta_0,\pi(\gamma_*)}(X).
 \end{array}
 \label{Con:EtaGradient_cov}
 \end{align}

 The minimum-MSE influence function satisfies
 \begin{align*}
 &h_{\epsilon}^{\rm MMSE}( \cdot ,\cdot,\beta_0,\gamma_*,f_X)  = \\
 &\argmin_{h( \cdot ,\cdot,\beta_0,\gamma_*,f_X)} \,
 \Bigg\{  \epsilon \, \left\|   \mathbb{E}_{f_X}\,\nabla_\pi \delta_{\beta_0,\pi(\gamma_*)}(X) -  \mathbb{E}_{f_X}\,\mathbb{E}_{\beta_0,\pi(\gamma_*)} \, h(Y,X,\beta_0,\gamma_*,f_X) \; \nabla_\pi \log f_{\beta_0,\pi(\gamma_*)}(Y\,|\, X) \right\|^2_{\gamma_*}  
 \nonumber \\ & \quad\quad\quad\quad\quad\quad
 +  \frac{\mathbb{E}_{f_X}\, {\rm Var}_{\beta_0,\pi(\gamma_*)}(h(Y,X,\beta_0,\gamma_*,f_X)\,|\, X)  } {n}   \Bigg\}
 \quad \quad \quad 
 \text{subject to \eqref{Con:Unbiased_cov} and \eqref{Con:EtaGradient_cov}}.
 \end{align*}
 
 In the locally quadratic case, following similar derivations as for equation (\ref{SolutionMMSE_linsys}) in Lemma \ref{lem_locquad}, we obtain (\ref{SolutionMMSE_COV}).

\subsection{Corollary \ref{SolutionMMSE_para}}

This is a direct implication of (\ref{SolutionMMSE}).

\subsection{Corollary \ref{SolutionMMSE_semiparam}}

This is a direct implication of (\ref{SolutionMMSE_rewrite}).

\subsection{Corollary \ref{MMSE_prob}}

Lemma \ref{lem_locquad_cov} implies, analogously to (\ref{SolutionMMSE_rewrite}), that
	\begin{align}
	&h_{\epsilon}^{\rm MMSE}(y,x)
	\,=\,     \delta(x)-\mathbb{E}_{f_X}\delta(X) +  s_{\beta\gamma}(y\,|\, x) '     [\mathbb{E}_{f_X}H_{\beta\gamma}(X)]^{-1}     
	\, \mathbb{E}_{f_X}\nabla_{\beta\gamma}  \delta(X) 
	\nonumber  \\& \quad \quad \quad \quad\quad \quad\quad \quad    +   (\epsilon n)  \widetilde s_{\pi} (y\,|\, x) ^\top
	\left\{  \mathbb{E}_{f_X}\nabla_{\pi} \delta(X) -\mathbb{E}_{f_X}\mathbb{E}\left[ h_{\epsilon}^{\rm MMSE}(Y,X) s_{\pi} (Y\,|\, X)\right]   \right\} .
	\label{SolutionMMSE_COV2}    
	\end{align}

Since $A$ and $X$ are independent, $\mathbb{E}_{f_X}\nabla_{\pi} \delta(X)$ can be represented by the function $$a\mapsto\mathbb{E}_{f_X}\left[\Delta(a,X)\right]-\mathbb{E}_{f_X}\delta(X).$$Likewise, $\mathbb{E}_{f_X}\mathbb{E}\left[ h_{\epsilon}^{\rm MMSE}(Y,X) s_{\pi} (Y\,|\, X)\right]$ can be represented by the function
$$a\mapsto\mathbb{E}_{f_X}\mathbb{E}\left[ h_{\epsilon}^{\rm MMSE}(Y,X) \,|\, A=a,X\right]=\overline{h}_{\epsilon}^{\rm MMSE}(a).$$ 

Moreover, we have for any cotangent element $u$ (a function of $a$),
\begin{align}
\widetilde s_{\pi} (y\,|\, x) ^\top \, u=& \mathbb{E}\left[ u(A) \,\big|\, Y=y,X=x\right]-\mathbb{E} \left[ u(A)\right]\notag \\
&- s_{\beta\gamma} (y\,|\, x)'[\mathbb{E}_{f_X}H_{\beta\gamma}(X)]^{-1}  \mathbb{E}_{f_X}\mathbb{E}\left[  s_{\beta\gamma} (Y\,|\, X)u(A) \right].\label{eq_score_X}
\end{align}

Corollary \ref{MMSE_prob} then follows from evaluating (\ref{eq_score_X}) at $$u(a):=\mathbb{E}_{f_X}\left[\Delta(a,X)\right]-\mathbb{E}_{f_X}\delta(X)-\overline{h}_{\epsilon}^{\rm MMSE}(a).$$

\subsection{Corollary \ref{MMSE_prob2}}

Let us start again from (\ref{SolutionMMSE_COV2}). In the correlated case, $\mathbb{E}_{f_X}\nabla_{\pi} \delta(X)$ can be represented by the function $$(a,x)\mapsto\Delta(a,x)f_X(x)-\delta(x)f_X(x).$$
Likewise, $\mathbb{E}_{f_X}\mathbb{E}\left[ h_{\epsilon}^{\rm MMSE}(Y,X) s_{\pi} (Y\,|\, X)\right]$ can be represented by the function
\begin{align*}(a,x)\mapsto&\mathbb{E}\left[ h_{\epsilon}^{\rm MMSE}(Y,X) \,|\, A=a,X=x\right]f_X(x)-\mathbb{E}\left[ h_{\epsilon}^{\rm MMSE}(Y,X) \,|\, X=x\right]f_X(x)\\
&\quad \quad \quad=\overline{h}_{\epsilon}^{\rm MMSE}(a,x)f_X(x)-\mathbb{E}\left[ h_{\epsilon}^{\rm MMSE}(Y,X) \,|\, X=x\right]f_X(x).\end{align*} 
Now, by (\ref{Con:EtaGradient_cov}) we have
\begin{equation}\mathbb{E}\left[ h_{\epsilon}^{\rm MMSE}(Y,X) \,|\, X=x\right]=\delta(x)-\mathbb{E}_{f_X} \delta(X).\label{eq_loc_rob_fX}\end{equation}
Hence, $\mathbb{E}_{f_X}\nabla_{\pi} \delta(X)-\mathbb{E}_{f_X}\mathbb{E}\left[ h_{\epsilon}^{\rm MMSE}(Y,X) s_{\pi} (Y\,|\, X)\right]$ can be represented by the function
$$(a,x)\mapsto \Delta(a,x)f_X(x)-\mathbb{E}_{f_X} \delta(X)f_X(x)-\overline{h}_{\epsilon}^{\rm MMSE}(a,x)f_X(x).$$

In the present case, cotangent elements are functions of $a$ and $x$. The corresponding squared dual norm is\footnote{This can be shown as in Subsection \ref{subsec_KL}, with the difference that here twice the KL divergence reads, using the notation of that subsection, $
	d( f_0 , f_* ) =  - \, 2 \,  \mathbb{E}_{f_X}\mathbb{E}_0  \log \frac {f_*(A\,|\, X)} {f_0(A \,|\, X)}$. Alternatively, Corollary \ref{MMSE_prob2} can be derived by defining $\pi_0$ as the joint distribution of $(A,X)$, and imposing the constraint that $\int_{\cal{A}} \pi_0(a,x)da=f_X(x)$.}
$$\|u\|_{\gamma_*}^2= \mathbb{E}_{f_X}\mathbb{E}\left[\left(\frac{u(A,X)-\mathbb{E}[u(A,X)\,|\, X]}{f_X(X)}\right)^2\right].$$  
In addition we have, for any cotangent element $u$ (a function of $a$ and $x$)
\begin{align}
\widetilde s_{\pi} (y\,|\, x) ^\top \, u=& \mathbb{E}\left[ \frac{u(A,X)}{f_X(X)} \,\big|\, Y=y,X=x\right]-\mathbb{E} \left[ \frac{u(A,X)}{f_X(X)}\,\big|\, X=x\right]\notag \\
&- s_{\beta\gamma} (y\,|\, x)'[\mathbb{E}_{f_X}H_{\beta\gamma}(X)]^{-1}  \mathbb{E}_{f_X}\mathbb{E}\left[  s_{\beta\gamma} (Y\,|\, X)\frac{u(A,X)}{f_X(X)} \right].\label{eq_score_X2}
\end{align}

 Corollary \ref{MMSE_prob2} then follows from evaluating (\ref{eq_score_X2}) at $$u(a,x):=\Delta(a,x)f_X(x)-\mathbb{E}_{f_X} \delta(X)f_X(x)-\overline{h}_{\epsilon}^{\rm MMSE}(a,x)f_X(x),$$
and noting that, by (\ref{eq_loc_rob_fX}), $\mathbb{E}[u(A,X)\,|\, X=x]=0$.

 \section{Complements to Section \ref{Sec_param}\label{App_der_sec3}}

 \subsection{Dual of the Kullback-Leibler divergence\label{subsec_KL}}
 
 Let $A$ be a random variable with domain ${\cal A}$,
 reference distribution $f_*(a)$ and ``true'' distribution $f_0(a)$. 
 We use notation $f_*(a)$ and  $f_0(a)$ as if those were densities, but point masses are also
 allowed. Twice the Kullback-Leibler (KL) divergence reads
 \begin{align*}
 d( f_0 , f_* ) =  - \, 2 \,  \mathbb{E}_0  \log \frac {f_*(A)} {f_0(A)} ,
 \end{align*}
 where $\mathbb{E}_0$ is the expectation under $f_0$. Let ${\cal F}$ be the set of all distributions, in particular, $ f \in {\cal F}$ implies $\int_{\cal{A}} f(a) da = 1$.
 Let $q : {\cal A} \rightarrow \mathbb{R}$ be a real valued function.
 For given $f_*  \in {\cal F}$ and $\epsilon > 0$ we define
 \begin{align*}
 \| q \|_{*,\epsilon} :=  \max_{\left\{ f_0 \in {\cal F} \, : \, d( f_0 , f_* ) \leq \epsilon \right\}} 
 \frac{    \mathbb{E}_0  \, q(A) - \mathbb{E}_*  \, q(A)   } {\sqrt{ \epsilon} } ,
 \end{align*}
 where $\mathbb{E}_*$ is the expectation under $f_*$.
 
 We have the following result.
 
 \begin{lemma}
 	For $q : {\cal A} \rightarrow \mathbb{R}$
 	and $f_* \in {\cal F}$
 	we assume that the 
 	moment-generating function 
 	$m_*(t) = \mathbb{E}_*  \exp( t \,  q(A) )$
 	exists for $t \in (\delta_-,  \delta_+)$
 	and some $\delta_- < 0$ and $\delta_+>0$.\footnote{%
 		Existence of $m_*(t)$ in an open interval around zero is equivalent
 		to having an exponential decay of the tails of the distribution of the random variable $Q=q(A)$. If $q(a)$ is bounded, then  $m_*(t)$ exists for all $t \in \mathbb{R}$.
 	}
 	For $\epsilon \in (0, \delta_+^2)$ we then have
 	\begin{align*}
 	\| q \|_{*,\epsilon}   &=     \sqrt{{\rm Var}_*(q(A))}  +  O(\epsilon^{\frac{1}{2}}).         
 	\end{align*}
 \end{lemma}
 
 \begin{proof}[\bf Proof]
 	Let the  cumulant-generating function of the random variable $q(A)$
 	under the reference measure $f_*$ be
 	$k_*(t) = \log m_*(t)$.
 	We assume existence of $m_*(t)$ and $k_*(t) $ for $t \in (\delta_-,  \delta_+)$.
 	This also implies that all derivatives of $m_*(t)$ and $k_*(t) $ exist in this interval.
 	We  denote the $p$-th derivative of $m_*(t)$ by $m^{(p)}_*(t)$,
 	and analogously for $k_*(t) $.

 	In the following we denote the maximizing $f_0$ 
 	in the definition of $\| q \|_{*,\epsilon}$ simply by $f_0$.
 	Applying standard optimization method (Karush-Kuhn-Tucker) we find the well-known
 	exponential tilting result
 	\begin{align*}
 	f_0(a) = c \, f_*(a) \,  \exp( t \,  q(a)    ) ,
 	\end{align*}
 	where the constants $c, t \in (0, \infty)$ are determined by the constraints
 	$\int_{\cal{A}} f_0(a) da = 1$ and $d( f_0 , f_* ) = \epsilon $. Using the constraint 
 	$\int_{\cal{A}} f_0(a) da = 1$ we can solve for $c$ to obtain
 	\begin{align*}
 	f_0(a) = \frac{ f_*(a) \,  \exp( t \,  q(a)    ) }
 	{\mathbb{E}_*  \exp( t \,  q(A) ) }  
 	=  \frac{ f_*(a) \,  \exp( t \,  q(a)    ) }
 	{ m_*(t) } .
 	\end{align*}
 	Using this we find that
 	\begin{align*}
 	d(t) & := d( f_0 , f_* )
 	\\ 
 	&=  2 \,  \mathbb{E}_*  \frac {f_0(A)} {f_*(A)}  \log \frac {f_0(A)} {f_*(A)}
 	\\  
 	&= \frac{ 2 \, t } {m_*(t) } \,  \mathbb{E}_*     \exp( t \,  q(A)    )     q(A) 
 	-  \frac{ 2 \log m_*(t)  } {m_*(t) } \,  \mathbb{E}_*     \exp( t \,  q(A)    )   
 	\\
 	&=      \frac{ 2 \, t \, m^{(1)}_*(t) } {m_*(t) }  - 2 \log m_*(t) 
 	\\  
 	&=    2 \left[ t \, k^{(1)}_*(t)    -   k_*(t) \right]  .
 	\end{align*}
 	We have $d(0)=0$, $d^{(1)}(0)=0$, $d^{(2)}(0) = 2 k_*^{(2)}(0) = 2 {\rm Var}_*(q(A))$,
 	$d^{(3)}(t) = 4  k^{(3)}_*(t) + 2 t k^{(4)}_*(t)$.
 	A mean-value expansion thus gives
 	\begin{align*}
 	d(t) &= {\rm Var}_*(q(A))  t^2 +  \frac {t^3} {6} \left[ 4 \,  k^{(3)}_*(\tilde t) + 2 \, \tilde t \, k^{(4)}_*( \tilde t) \right] ,
 	\end{align*}
 	where $0 \leq \tilde t \leq t \leq \delta_+$.
 	The value $t$ that satisfies the constraint  $d(t) = \epsilon$ therefore satisfies
 	\begin{align*}
 	t =  \frac{ \epsilon^{\frac{1}{2}} } {\sqrt{ {\rm Var}_*(q(A))}} + O(\epsilon).
 	\end{align*}
 	Next, using that
 	$       \| q \|_{*,\epsilon}  = \epsilon^{-\frac{1}{2}} \;  \mathbb{E}_* \left[  \left(  \frac{  f_0(A)   } { f_*(A)} -1 \right)  q(A) \right] $
 	we find
 	\begin{align*}
 	\| q \|_{*,\epsilon}
 	&= \epsilon^{-\frac{1}{2}} \left[   k^{(1)}_*(t)    -  k^{(1)}_*(0)   \right] .
 	\end{align*}
 	Again using that  $k_*^{(2)}(0) =  {\rm Var}_*(q(A))$ and applying a mean value expansion we obtain
 	\begin{align*}
 	\| q \|_{*,\epsilon}
 	&= \epsilon^{-\frac{1}{2}} \left[  t \, k^{(2)}_*(t)  + \frac 1 2 \, t^2 \,   k^{(3)}_*(\bar t)   \right] 
 	\\
 	&= \epsilon^{-\frac{1}{2}} \left[  t \,  {\rm Var}_*(q(A)) + \frac 1 2 \, t^2 \,   k^{(3)}_*(\bar t)   \right] 
 	\\ 
 	&=  \sqrt{{\rm Var}_*(q(A))}  +  O(\epsilon^{\frac{1}{2}}) ,
 	\end{align*}
 	where $\bar t \in [0,t]$.
 \end{proof}
 
 \subsection{Equations (\ref{eq_minimum_IP_sol}), (\ref{regu_IF_semi}) and (\ref{eq_minimum_IP_sol_X})}
 
Here we use simplified notation as in Section \ref{Sec_param}. Let us start by deriving (\ref{eq_minimum_IP_sol}). In this case $\beta_0$ and $\gamma_*$ are known, and Corollary \ref{SolutionMMSE_semiparam} gives
\begin{align*}
h_{\epsilon}^{\rm MMSE}
=    &   (\epsilon n) \mathbb{E}_{{\cal{A}}\,|\, {\cal{Y}}}\left[\Delta-\delta-\mathbb{E}_{{\cal{Y}}\,|\, {\cal{A}}}h^{\rm MMSE}\right],
\end{align*}
so
\begin{align*}
h_{\epsilon}^{\rm MMSE}
=    &   \left[(\epsilon n)^{-1}\mathbb{I}_{\cal{Y}}+\mathbb{E}_{{\cal{A}}\,|\, {\cal{Y}}}\circ\mathbb{E}_{{\cal{Y}}\,|\, {\cal{A}}}\right]^{-1} \mathbb{E}_{{\cal{A}}\,|\, {\cal{Y}}}\left[\Delta-\delta\right].
\end{align*}
(\ref{eq_minimum_IP_sol}) then follows from the operator identity
\begin{align*}
\left[(\epsilon n)^{-1}\mathbb{I}_{\cal{Y}}+\mathbb{E}_{{\cal{A}}\,|\, {\cal{Y}}}\circ\mathbb{E}_{{\cal{Y}}\,|\, {\cal{A}}}\right]^{-1} \mathbb{E}_{{\cal{A}}\,|\, {\cal{Y}}}= \mathbb{E}_{{\cal{A}}\,|\, {\cal{Y}}}\left[\mathbb{E}_{{\cal{Y}}\,|\, {\cal{A}}}\circ\mathbb{E}_{{\cal{A}}\,|\, {\cal{Y}}}+(\epsilon n)^{-1}\mathbb{I}_{\cal{A}}\right]^{-1}.
\end{align*}

Let us now derive (\ref{regu_IF_semi}). In this case $\gamma_*$ is known. Since $\Delta(A)=c'\beta_0={\delta}$, Corollary \ref{SolutionMMSE_semiparam} implies
\begin{align*}
h_{\epsilon}^{\rm MMSE}(y)
=    &  s_{\beta\gamma} (y) '     H_{\beta\gamma}^{-1}   c    - (\epsilon n) \bigg\{\mathbb{E}\left[\overline{h}^{\rm MMSE}(A)\,|\, Y=y\right]-s_{\beta\gamma} (y)'H_{\beta\gamma}^{-1}\mathbb{E} \left[s_{\beta\gamma} (Y) \overline{h}^{\rm MMSE}(A)     \right]\bigg\}.
\end{align*}
Hence, we have, for some vector $b$,
\begin{align*}
h_{\epsilon}^{\rm MMSE}
=    &s_{\beta\gamma} (y)'b - (\epsilon n) \mathbb{E}_{{\cal{A}}\,|\, {\cal{Y}}}\circ\mathbb{E}_{{\cal{Y}}\,|\, {\cal{A}}}h^{\rm MMSE}.
\end{align*}
Using the Woodbury identity
\begin{equation*}\left[\mathbb{I}_{\cal{Y}}+(\epsilon n)\mathbb{E}_{{\cal{A}}\,|\, {\cal{Y}}}\circ\mathbb{E}_{{\cal{Y}}\,|\, {\cal{A}}}\right]^{-1}=  \underset{=\mathbb{W}^{\epsilon}}{\underbrace{\mathbb{I}_{\cal{Y}}-\mathbb{E}_{{\cal{A}}\,|\, {\cal{Y}}} \left[\mathbb{E}_{{\cal{Y}}\,|\, {\cal{A}}}\circ\mathbb{E}_{{\cal{A}}\,|\, {\cal{Y}}}+(\epsilon n)^{-1}\mathbb{I}_{\cal{A}}\right]^{-1} \mathbb{E}_{{\cal{Y}}\,|\, {\cal{A}}}}},\end{equation*}
we thus obtain
\begin{align*}
h_{\epsilon}^{\rm MMSE}
=    & \mathbb{W}^{\epsilon}s_{\beta\gamma} (y) 'b.
\end{align*}
Lastly, since by (\ref{Con:EtaGradient}) $\mathbb{E}\,[h_{\epsilon}^{\rm MMSE}(Y)s_{\beta\gamma}(Y) ]=c$, we obtain (\ref{regu_IF_semi}) whenever the denominator is non-singular. 

Finally, let us derive (\ref{eq_minimum_IP_sol_X}). In this case $\beta_0$ and $\gamma_*$ are known and $\Delta(A)$ does not depend on $X$, and Corollary \ref{MMSE_prob} gives
\begin{align*}
h_{\epsilon}^{\rm MMSE}
=    &   (\epsilon n) \mathbb{E}_{{\cal{A}}\,|\, {\cal{Y}},{\cal{X}}}\left[\mathbb{E}_{f_X}(\Delta-\delta)-\mathbb{E}_{{\cal{Y}},{\cal{X}}\,|\, {\cal{A}}}h^{\rm MMSE}\right].
\end{align*}
Hence, denoting $\mathbb{I}_{{\cal{Y}},{\cal{X}}}h(y,x)=h(y,x)$ the identity operator, we have
\begin{align*}
h_{\epsilon}^{\rm MMSE}
=    &   \left[(\epsilon n)^{-1}\mathbb{I}_{{\cal{Y}},{\cal{X}}}+\mathbb{E}_{{\cal{A}}\,|\, {\cal{Y}},{\cal{X}}}\circ\mathbb{E}_{{\cal{Y}},{\cal{X}}\,|\, {\cal{A}}}\right]^{-1} \mathbb{E}_{{\cal{A}}\,|\, {\cal{Y}},{\cal{X}}}\mathbb{E}_{f_X}(\Delta-\delta).
\end{align*}
(\ref{eq_minimum_IP_sol_X}) then follows from
\begin{align*}
\left[(\epsilon n)^{-1}\mathbb{I}_{{\cal{Y}},{\cal{X}}}+\mathbb{E}_{{\cal{A}}\,|\, {\cal{Y}},{\cal{X}}}\circ\mathbb{E}_{{\cal{Y}},{\cal{X}}\,|\, {\cal{A}}}\right]^{-1} \mathbb{E}_{{\cal{A}}\,|\, {\cal{Y}},{\cal{X}}}= \mathbb{E}_{{\cal{A}}\,|\, {\cal{Y}},{\cal{X}}}\left[\mathbb{E}_{{\cal{Y}},{\cal{X}}\,|\, {\cal{A}}}\circ\mathbb{E}_{{\cal{A}}\,|\, {\cal{Y}},{\cal{X}}}+(\epsilon n)^{-1}\mathbb{I}_{\cal{A}}\right]^{-1}.
\end{align*}

 \section{Computation in semi-parametric mixture models\label{App_compute}}
 
 Here we describe how we compute a numerical approximation to the minimum-MSE estimator in semi-parametric mixture models 
 $$\widehat{\delta}_{\epsilon}^{\rm MMSE}=\mathbb{E}_{\widehat{\beta},\pi(\widehat{\gamma})}\,\Delta_{\widehat{\beta}}(A)+\frac{1}{n}\sum_{i=1}^nh_{\epsilon}^{\rm MMSE}(Y_i,\widehat{\beta},\widehat{\gamma}),$$ where $h_{\epsilon}^{\rm MMSE}$ is given by Corollary \ref{SolutionMMSE_semiparam}, and $\widehat{\beta},\widehat{\gamma}$ are preliminary estimates. {As we pointed out in Section \ref{Sec_param}, $h_{\epsilon}^{\rm MMSE}$ is the solution to a (well-posed) Tikhonov-regularized linear inverse problem, and many numerical methods are available to solve such problems; see Engl \textit{et al.} (2000) and Kress (2014) for classic references. The simulation-based approach that we have implemented and describe here is closely related to the strategy presented in Bonhomme (2012).} We abstract from conditioning covariates. In the presence of correlated covariates $X_i$, we use the same technique to approximate $h_{\epsilon}^{\rm MMSE}(\cdot\,|\, x)$ for each value of $X_i=x$.  We use this approach in the numerical illustration based on the dynamic panel data model in Section \ref{Sec_numeric}, where the covariate is the initial condition. We denote $\eta=(\beta',\gamma')'$.\footnote{Here we present a general method based on simulations. In the cross-sectional probit model (\ref{mod_bc_cs}), explicit closed-form expressions are available, and we use those for computation in our first illustration.}

 Draw an i.i.d. sample $(Y^{(1)},A^{(1)}),...,(Y^{(S)},A^{(S)})$ of $S$ draws from $ g_{\beta}\times \pi({\gamma})$. Let $G$ be $S\times S$ with $(\tau,s)$ element $g_{\beta}(Y^{(\tau)}\,|\, A^{(s)})/\sum_{s'=1}^Sg_{\beta}(Y^{(\tau)}\,|\,A^{(s')})$, $G_Y$ be $N\times S$ with $(i,s)$ element $g_{\beta}(Y_i\,|\, A^{(s)})/\sum_{s'=1}^Sg_{\beta}(Y_i\,|\,A^{(s')})$, $\Delta$ be $S\times 1$ with $s$-th element $\Delta_{\beta}(A^{(s)})$, $I$ be the $S\times S$ identity matrix, and $\iota$ and $\iota_Y$ be the $S\times 1$ and $N\times 1$ vectors of ones. In addition, let $D$ be the $S\times \dim\eta$ matrix with $(s,k)$ element 
 $$d_{\eta_k}(Y^{(s)})=\frac{\sum_{s'=1}^S\left(\nabla_{\eta_k}\log g_{\beta}(Y^{(s)}\,|\,A^{(s')})+\nabla_{\eta_k}\log\pi({\gamma})(A^{(s')})\right)g_{\beta}(Y^{(s)}\,|\,A^{(s')})}{\sum_{s'=1}^Sg_{\beta}(Y^{(s)}\,|\,A^{(s')})},$$
 and let $D_Y$ be $N\times \dim\eta$ with $(i,k)$ element $d_{\eta_k}(Y_i)$, $Q=I-DD^{\dagger}$, $\widetilde{G}_Y=G_Y-D_YD^\dagger G$, $\widetilde{\iota}_Y=\iota_Y-D_YD^{\dagger}\iota$, $\widetilde{G}=Q G$, $\widetilde{\iota}=Q\iota$, and $\partial \Delta$ be the $K\times 1$ vector with $k$-th element $\frac{1}{S}\sum_{s=1}^S\nabla_{\eta_k}\Delta(A^{(s)},\beta)+\Delta(A^{(s)},\beta)\nabla_{\eta_k}\log \pi (\gamma)(A^{(s)})$. 
 
 From Corollary \ref{SolutionMMSE_semiparam}, a fixed-$S$ approximation to the minimum-MSE estimator is then
 \begin{align*} 
 &\widetilde{\delta}_{\epsilon}^{\rm MMSE}=\iota^\dagger\Delta+\iota_Y^{\dagger}\widetilde{h}_{\epsilon}^{\rm MMSE},
 \end{align*}
 where 
 \begin{align*} 
& \widetilde{h}_{\epsilon}^{\rm MMSE}=
 D_Y(D'D/S)^{-1}\,\partial \Delta+
 (\epsilon n)\Bigg[\left(\widetilde{G}_Y-\widetilde{\iota}_Y\iota^\dagger\right)\Delta\notag\\
 &\quad \quad \quad \quad \quad \quad-\widetilde{G}_YG'\left(\widetilde{G}G'+(\epsilon n)^{-1}I\right)^{-1}\left((\epsilon n)^{-1}D(D'D/S)^{-1}\, \partial \Delta+\left(\widetilde{G}-\widetilde{\iota}\iota^\dagger\right)\Delta\right)\Bigg],\end{align*}
and $(\beta,\gamma)$ are replaced by the preliminary $(\widehat{\beta},\widehat{\gamma})$ in all the quantities above, including when producing the simulated draws. $\widetilde{\delta}_{\epsilon}^{\rm MMSE}$ is consistent for $\widehat{\delta}_{\epsilon}^{\rm MMSE}$ as $S$ tends to infinity {for fixed $n$, under suitable regularity conditions (see Bonhomme, 2012, for a closely related setup). Note that matrix inverses remain well-defined as $S$ tends to infinity, due to the presence of the Tikhonov-penalization term $(\epsilon n)^{-1}I$.}  

 \paragraph{Confidence intervals.}
 
 From Subsection \ref{subsec_CI}, computing confidence intervals only requires, in addition to computing critical values under correct specification, to compute an estimate of the bias of the estimator $b_{\epsilon}(h,\widehat{\beta},\widehat{\gamma})$. In semi-parametric mixture models we have, for an asymptotically linear estimator based on $h$ satisfying \eqref{Con:Unbiased} and \eqref{Con:EtaGradient},
 $$ b_{\epsilon}(h,\beta_0,\gamma_*)=\epsilon^{\frac{1}{2}} \left\{{\limfunc{Var}}_{\beta_0,\pi(\gamma_*)}[\Delta_{\beta_0}(A)-\mathbb{E}_{\beta_0,\pi(\gamma_*)}(h(Y)\,|\, A)]\right\}^{\frac{1}{2}}.$$
A numerical approximation to the bias of $\widehat{\delta}_{\epsilon}^{\rm MMSE}$ is then
\begin{align*} 
& \widetilde{b}_{\epsilon}(h_{\epsilon}^{\rm MMSE},\beta_0,\gamma_*)=\epsilon^{\frac{1}{2}}\left\|\Delta-\iota^{\dagger}\Delta-G'\widetilde{h}_{\epsilon}^{\rm MMSE}\right\|.
\end{align*}

 \paragraph{Values of $\epsilon$.}
 
 In turn, $\epsilon_k$ in (\ref{seq_eps_gen}) can be approximated as $\mu(\alpha,p)^2/(n\lambda_{k})$, where $\lambda_{k}$ is the $k$-th largest eigenvalue of $G'QG=\widetilde{G}'\widetilde{G}$ (removing the eigenvalue equal to one since it corresponds to a constant eigenfunction).

\section{Models defined by moment restrictions\label{App_GMM}}

  In this section, we consider settings where a finite-dimensional parameter $(\beta_0',\pi_0')'$ does not fully determine the distribution $f_0$ of $Y$, but satisfies a finite-dimensional system of moment conditions
  \begin{equation}
  \mathbb{E}_{f_0}\Psi(Y,\beta_0,\pi_0)=0,\label{moment_cond_gmm}
  \end{equation}
which may be just-identified, over-identified or under-identified. We focus on asymptotically linear generalized method-of-moments (GMM) estimators of $\delta_{\beta_0,\pi_0}$ that satisfy 
  \begin{equation}\widehat{\delta}=\delta_{\beta_0,\pi(\gamma_*)}+a(\beta_0,\gamma_*)'\frac{1}{n}\sum_{i=1}^n\Psi(Y_i,\beta_0,\pi(\gamma_*))+o_{P_0}(\epsilon^{\frac{1}{2}}+n^{-\frac{1}{2}}),\label{est_delta_hat_GMM}\end{equation}
  for a parameter vector $a(\beta_0,\gamma_*)$. We will characterize the form of $a(\beta_0,\gamma_*)$ leading to minimum worst-case MSE in $\Gamma_{\epsilon}(\gamma_*)$.

 We assume that the remainder in (\ref{est_delta_hat_GMM}) is uniformly bounded similarly as in (\ref{DefRemainder1}). In this case local robustness with respect to $(\beta_0',\gamma_*')'$ takes the form
\begin{equation}\label{localrobust_GMM}
\nabla_{\beta\gamma}\delta_{\beta_0,\pi(\gamma_*)}+\mathbb{E}_{f_0}\nabla_{\beta\gamma}\Psi(Y,\beta_0,\pi(\gamma_*))\, a(\beta_0,\gamma_*)=0.
\end{equation}
It is natural to focus on asymptotically linear GMM estimators here, since $f_0$ is unrestricted except for the moment condition (\ref{moment_cond_gmm}).

To derive the worst-case bias of $\widehat{\delta}$ note that, by (\ref{moment_cond_gmm}), for any $\pi_0\in \Gamma_{\epsilon}(\gamma_*)$ we have
$$\mathbb{E}_{f_0}\Psi(Y,\beta_0,\pi(\gamma_*))=-\left[\mathbb{E}_{f_0}\nabla_{\pi}\Psi(Y,\beta_0,\pi(\gamma_*))\right]'\,(\pi_0-\pi(\gamma_*))+o(\epsilon^{\frac{1}{2}}),$$
so, under appropriate regularity conditions,
\begin{align*} 
&\sup_{\pi_0 \in \Gamma_\epsilon(\gamma_*)}
\left|   \mathbb{E}_{f_0} \widehat{\delta}  -\delta_{\beta_0,\pi_0}\right|=\epsilon^{\frac{1}{2}} \, \left\|  \nabla_\pi \delta_{\beta_0,\pi(\gamma_*)} + \mathbb{E}_{f_0}\nabla_{\pi}\Psi(Y,\beta_0,\pi(\gamma_*))\, a(\beta_0,\gamma_*)\right\|_{\gamma_*} + o(\epsilon^{\frac{1}{2}} +n^{-\frac{1}{2}}).\end{align*} 
The worst-case MSE of 
$$\widehat{\delta}_{a,{\beta_0},\gamma_*}:=\delta_{\beta_0,\pi(\gamma_*)}+a(\beta_0,\gamma_*)'\frac{1}{n}\sum_{i=1}^n\Psi(Y_i,\beta_0,\pi(\gamma_*))$$
is thus
\begin{align*} 
&\epsilon \, \left\|  \nabla_\pi \delta_{\beta_0,\pi(\gamma_*)} + \mathbb{E}_{f_0}\nabla_{\pi}\Psi(Y,\beta_0,\pi(\gamma_*))\, a(\beta_0,\gamma_*)\right\|^2_{\gamma_*}\\
&\quad \quad \quad +a(\beta_0,\gamma_*)' \frac{\mathbb{E}_{f_0}\Psi(Y,\beta_0,\pi(\gamma_*))\Psi(Y,\beta_0,\pi(\gamma_*))'}{n}a(\beta_0,\gamma_*)+o(\epsilon +n^{-1}).\end{align*}

To obtain an explicit expression for the minimum-MSE estimator, let us focus on the case where $\pi_0$ is finite-dimensional and $\|\cdot\|_{\gamma_*}=\|\cdot\|_{\Omega^{-1}}$. Let us define 
\begin{align*} 
&V_{\beta_0,\pi(\gamma_*)}=\mathbb{E}_{f_0}\Psi(Y,\beta_0,\pi(\gamma_*))\Psi(Y,\beta_0,\pi(\gamma_*))',\,\,\,\,\, K_{\beta_0,\pi(\gamma_*)}=\mathbb{E}_{f_0}\nabla_{\pi}\Psi(Y,\beta_0,\pi(\gamma_*)),\end{align*}
and $$K_{\beta_0,\gamma_*}=\mathbb{E}_{f_0}\nabla_{\beta\gamma}\Psi(Y,\beta_0,\pi(\gamma_*)).$$
For all $\beta_0,\gamma_*$ we aim to minimize
\begin{align*} 
&\epsilon \, \left\|  \nabla_\pi \delta_{\beta_0,\pi(\gamma_*)} + K_{\beta_0,\pi(\gamma_*)}a(\beta_0,\gamma_*)\right\|^2_{\Omega^{-1}} +a(\beta_0,\gamma_*)' \frac{V_{\beta_0,\pi(\gamma_*)}}{n}a(\beta_0,\gamma_*),\\&\quad\quad\quad \text{ subject to }\nabla_{\beta\gamma}\delta_{\beta_0,\pi(\gamma_*)}+K_{\beta_0,\gamma_*}a(\beta_0,\gamma_*)=0.
\end{align*} 
A solution is given by\footnote{Here we assume that $K_{\beta_0,\gamma_*}V_{\beta_0,\pi(\gamma_*)}^{\dagger} K_{\beta_0,\gamma_*}'$ is non-singular, requiring that $\beta_0,\gamma_*$ be identified from the moment conditions. Existence follows from the fact that, by the generalized information identity, $V_{\beta_0,\pi(\gamma_*)}a=0$ implies that $K_{\beta_0,\pi(\gamma_*)}a=0$. Moreover, although $a_{\epsilon}^{\rm MMSE}(\beta_0,\gamma_*)$ may not be unique, $a_{\epsilon}^{\rm MMSE}(\beta_0,\gamma_*)'\Psi(Y,\beta_0,\pi(\gamma_*))$ is unique almost surely.}
\begin{align}
&a_{\epsilon}^{\rm MMSE}(\beta_0,\gamma_*)=- B_{\beta_0,\pi(\gamma_*),\epsilon }^{\dagger}K_{\beta_0,\gamma_*}'\left(K_{\beta_0,\gamma_*}B_{\beta_0,\pi(\gamma_*),\epsilon }^{\dagger} K_{\beta_0,\gamma_*}'\right)^{-1}\nabla_{\beta\gamma}\delta_{\beta_0,\pi(\gamma_*)}\notag\\
&-B_{\beta_0,\pi(\gamma_*),\epsilon }^{\dagger}\left(I-K_{\beta_0,\gamma_*}'\left(K_{\beta_0,\gamma_*}B_{\beta_0,\pi(\gamma_*),\epsilon }^{\dagger} K_{\beta_0,\gamma_*}'\right)^{-1}K_{\beta_0,\gamma_*}B_{\beta_0,\pi(\gamma_*),\epsilon }^{\dagger}\right) K_{\beta_0,\pi(\gamma_*)}'\Omega^{-1}\nabla_{\pi}\delta_{\beta_0,\pi(\gamma_*)},\label{aMMSE}
\end{align}
where $B_{\beta_0,\pi(\gamma_*),\epsilon }= K_{\beta_0,\pi(\gamma_*)}'\Omega^{-1}K_{\beta_0,\pi(\gamma_*)}+(\epsilon n)^{-1}V_{\beta_0,\pi(\gamma_*)}$, and $B_{\beta_0,\pi(\gamma_*),\epsilon }^{\dagger}$ is its Moore-Penrose generalized inverse. Note that, in the likelihood case and taking $\Psi(y,\beta,\pi)=\nabla_{\pi}\log f_{\beta,\pi}(y)$, the function $h(y,\beta_0,\gamma_*)=a_{\epsilon}^{\rm MMSE}(\beta_0,\gamma_*)'\Psi(y,\beta_0,\pi(\gamma_*))$ simplifies to (\ref{SolutionMMSE}).

As a special case, when $\epsilon=0$ we have
\begin{align*}
&a_{0}^{\rm MMSE}(\beta_0,\gamma_*)=-V_{\beta_0,\pi(\gamma_*)}^{\dagger}K_{\beta_0,\gamma_*}'\left(K_{\beta_0,\gamma_*}V_{\beta_0,\pi(\gamma_*)}^{\dagger} K_{\beta_0,\gamma_*}'\right)^{-1}\nabla_{\beta\gamma}\delta_{\beta_0,\pi(\gamma_*)}.
\end{align*}
In this case, given preliminary estimators $\widehat{\beta}$ and $\widehat{\gamma}$, the minimum-MSE estimator $$\widehat{\delta}_{\epsilon}^{\rm MMSE}=\delta_{\widehat{\beta},\pi(\widehat{\gamma})}+a_{0}^{\rm MMSE}(\widehat{\beta},\widehat{\gamma})'\frac{1}{n}\sum_{i=1}^n\Psi(Y_i,\widehat{\beta},\pi(\widehat{\gamma}))$$ is the one-step approximation to the optimal GMM estimator based on the reference model. To obtain a feasible estimator one simply replaces the expectations in $V_{\beta_0,\pi(\gamma_*)}$ and $K_{\beta_0,\gamma_*}$ by sample analogs.

As a second special case, consider $\epsilon $ tending to infinity. Focusing on the known-$(\beta_0,\gamma_*)$ case for simplicity, $a_{\epsilon}^{\rm MMSE}(\beta_0,\gamma_*)$ tends to $-K^{\rm ginv}_{\beta_0,\pi(\gamma_*)}\nabla_{\pi}\delta_{\beta_0,\pi(\gamma_*)}$,
where
\begin{align*}&K^{\rm ginv}_{\beta_0,\pi(\gamma_*)}:=\\
&	\left( V_{\beta_0,\pi(\gamma_*)}^\dagger \right)^{1/2} \left[\left( V_{\beta_0,\pi(\gamma_*)}^\dagger \right)^{1/2}   K_{\beta_0,\pi(\gamma_*)}'\Omega^{-1}K_{\beta_0,\pi(\gamma_*)} \left( V_{\beta_0,\pi(\gamma_*)}^\dagger \right)^{1/2} \right]^{\dagger}
\left( V_{\beta_0,\pi(\gamma_*)}^\dagger \right)^{1/2}
K_{\beta_0,\pi(\gamma_*)}'\Omega^{-1}\end{align*}
is a generalized inverse of $K_{\beta_0,\pi(\gamma_*)}$, and the choice of $\Omega$
corresponds to choosing one specific such generalized inverse.
In this case, the minimum-MSE estimator is the one-step approximation to a particular GMM estimator based on the ``large'' model.

Lastly, given a parameter vector $a$, confidence intervals can be constructed as explained in Subsection \ref{subsec_CI}, taking $$b_{\epsilon}(a,\widehat{\beta},\widehat{\gamma})=\epsilon^{\frac{1}{2}} \, \left\|  \nabla_\pi \delta_{\widehat{\beta},\pi(\widehat{\gamma})} + \frac{1}{n}\sum_{i=1}^n\nabla_{\pi}\Psi(Y_i,\widehat{\beta},\pi(\widehat{\gamma}))\, a(\widehat{\beta},\widehat{\gamma})\right\|_{\Omega^{-1}}.$$ 

\paragraph{Example.}

Consider again the OLS/IV example of Subsection \ref{subsec_par}, but now drop the Gaussian assumptions on the distributions. For known $C$, the set of moment conditions corresponds to the moment functions
$$\Psi(y,x,z,\beta,\pi)=\left(\begin{array}{c}x(y-x'\beta-\pi'(x-C z))\\z(y-x'\beta)\end{array}\right).$$ 
In this case, letting $W=(X',Z')'$ we have 
$$K_{\beta_0,\gamma_*}=-\mathbb{E}_{f_0}\left(XW'\right),\,\,\,\,\, K_{\beta_0,\pi(\gamma_*)}=-\mathbb{E}_{f_0}\left(\begin{array}{cc}XX'&XZ'\\(X-C Z)X' &0\end{array}\right),$$
and
$$ V_{\beta_0,\pi(\gamma_*)}=\mathbb{E}_{f_0}\left((Y-X'\beta_0)^2WW'\right).$$
Given a preliminary estimator $\widetilde{\beta}$, $V_{\beta_0,\pi(\gamma_*)}$ can be estimated as
$\frac{1}{n}\sum_{i=1}^n(Y_i-X_i'\widetilde{\beta})^2W_iW_i'$,
whereas $K_{\beta_0,\gamma_*}$ and $K_{\beta_0,\pi(\gamma_*)}$ can be estimated as sample means. The estimator based on (\ref{aMMSE}) then interpolates nonlinearly between the OLS and IV estimators, similarly as in the likelihood case.

\section{Numerical illustrations\label{App_Fig}}

\subsection{Interpretation of $\epsilon$ in the cross-sectional binary choice model}

Here we use the binary choice model of Subsection \ref{subsec_bin} to provide additional intuition about the interpretation of $\epsilon$ based on statistical testing. 

Let ${\cal{U}}_k$ denote the span of the first $k$ non-constant eigenfunctions of the operator $\widetilde{H}_{\pi}$. By construction, any density $\pi_0\notin \Gamma_{\epsilon_k}(\gamma_*)$ such that $(\pi_0-\pi(\gamma_*))/\pi(\gamma_*)\in {\cal{U}}_k$ can be ``detected'' easily, in the sense that the local power of a $5\%$-likelihood ratio test exceeds 80\%.\footnote{${\cal{U}}_k$ consists of cotangent elements that have zero mean under the reference model. Any such $u\in{\cal{T}}$ can be mapped to a direction $v=u \cdot \pi(\gamma_*)\in \overline {\cal T}$ in the tangent space.} In the upper panel of Figure \ref{Fig_CS_eps_interpret}, we plot the eigenfunctions in ${\cal{U}}_k$. Plotting those allows one to visualize the directions along which setting $\epsilon$ to either of the $\epsilon_k$'s provides power guarantees outside the neighborhood. We see that the eigenfunctions do not vary outside the $[-1,1]$ interval, where the support of $X'\beta_0$ lies. Within the $[-1,1]$ interval, the eigenfunctions oscillate and belong to orthogonal bases of functions. 

To see how well the true $\pi_0$ can be approximated using the directions in ${\cal{U}}_k$, in the bottom panel of Figure \ref{Fig_CS_eps_interpret}, we report the projection of $\pi_0$ onto ${\cal{U}}_k$. We see that, outside the $[-1,1]$ interval, the projection is only governed by the reference normal density, reflecting the limited support of $X$. Within the interval, the approximation to the true bimodal density improves as $k$ increases. At the same time, note that, consistently with our local approach, the approximating functions are not necessarily non-negative.\footnote{In addition, since we know $\pi_0$ in this exercise, we can compute the local power of a $5\%$-likelihood ratio test in direction $\pi_0-\pi(\gamma_*)$, for any value of $\epsilon$. We find a power of 0.51 at $\epsilon_1$ and 0.71 at $\epsilon_2$ when $X$ has 4 points of support, and 0.67 at  $\epsilon_1$, 0.92 at $\epsilon_2$, and 0.99 at $\epsilon_3$ when $X$ has 20 points of support.}

\begin{figure}[tbp]	\caption{Eigenfunctions of $\widetilde{H}_{\pi}$ in the cross-sectional binary choice model\label{Fig_CS_eps_interpret}}
	\begin{center}
		\begin{tabular}{cc}
			\multicolumn{2}{c}{A. Eigenfunctions}\\
			(a) $n_X=4$ & (b) $n_X=20$	\\
			\includegraphics[width=60mm, height=40mm]{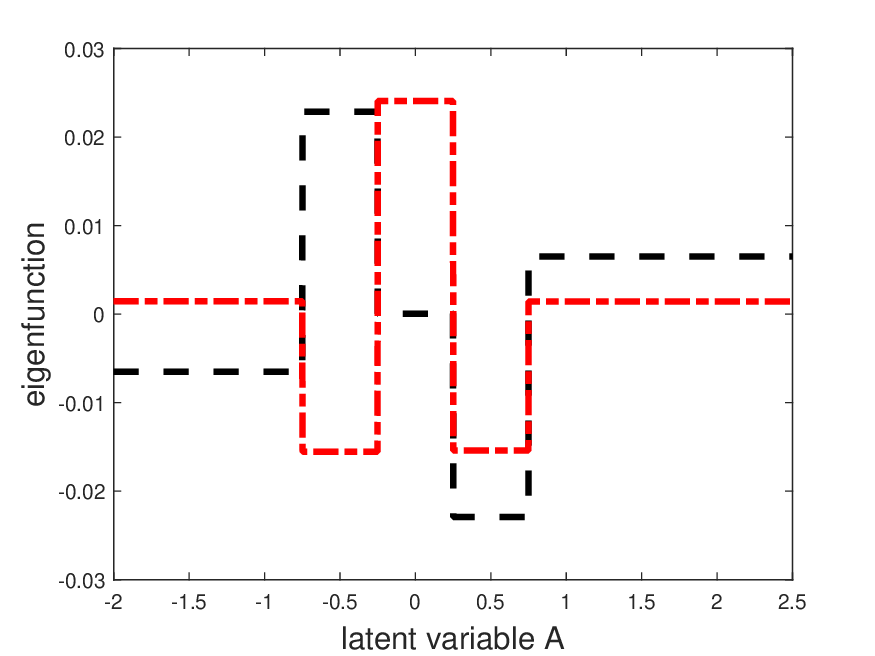}&	\includegraphics[width=60mm, height=40mm]{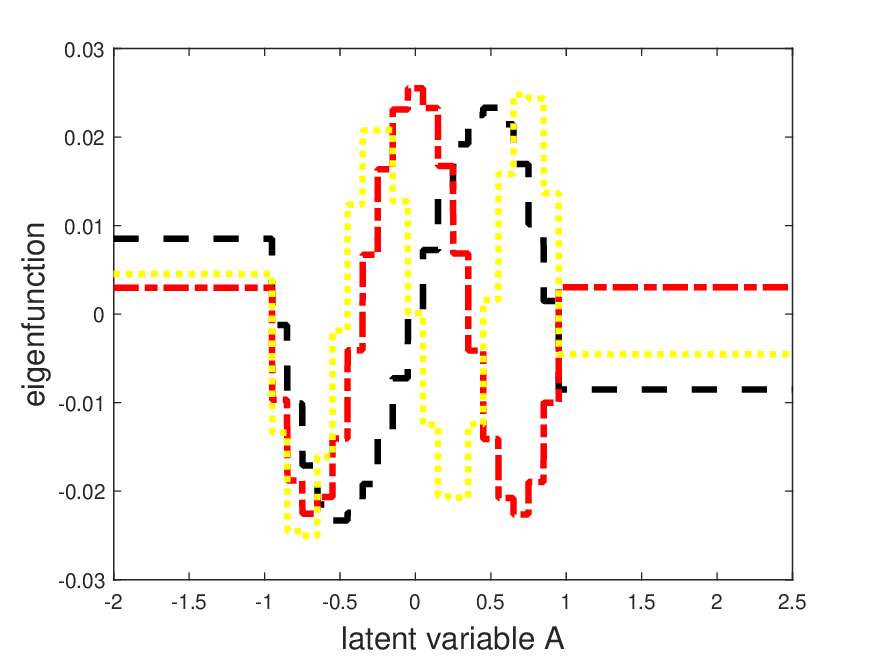}\\
			\multicolumn{2}{c}{B. Projections}\\
			(c) $n_X=4$ & (d) $n_X=20$ 	\\
			\includegraphics[width=60mm, height=40mm]{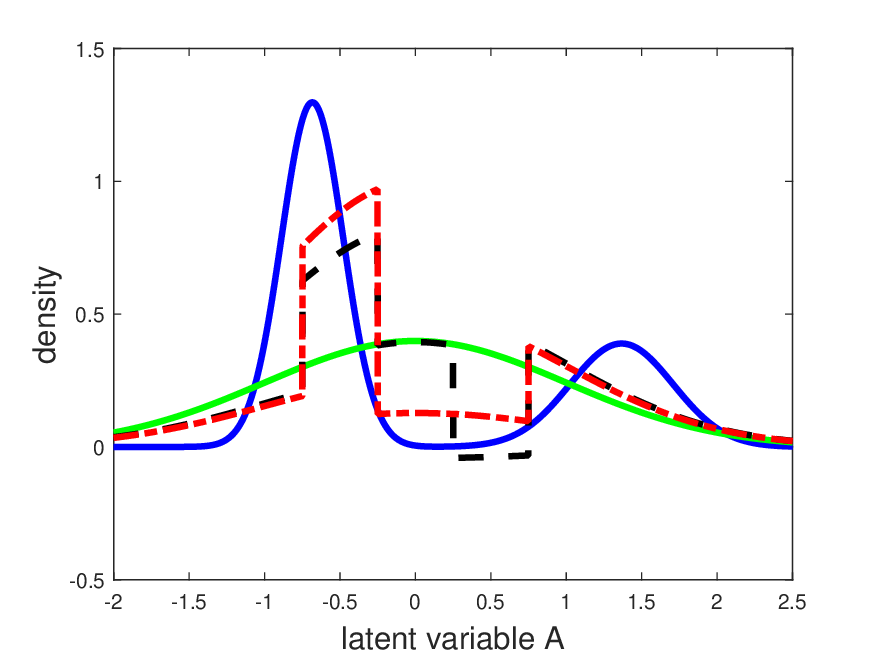}&	\includegraphics[width=60mm, height=40mm]{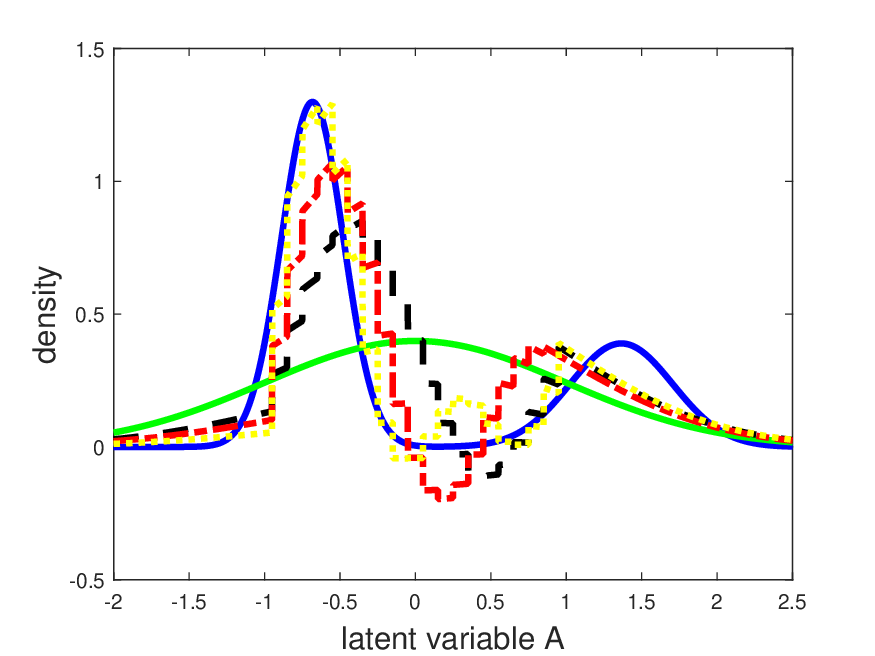}\\			\end{tabular}
	\end{center}
	\par
	\textit{{\small Notes: In the top panel we report the first 2 (respectively, first 3) non-constant eigenfunctions of $\widetilde{H}_{\pi}$. The first eigenfunction is shown in dashed, the second one in dashed-dotted, and the third one in dotted. In the bottom panel we plot the true and reference densities in solid, as well as the successive approximations using the first, the first two, or the first three eigenfunctions. }}
\end{figure}

\clearpage

\subsection{Additional tables}

\begin{table}[h!]
	\caption{Monte Carlo simulation of the average effect in the cross-sectional binary choice model, \textit{interpolation} ($x_0=(0.5,1)'$) \label{TabProbit_CS_1}}
	\begin{center}
		\begin{tabular}{l||ccccccc}
			Minimum-MSE, for $\epsilon=$  & {0.0001} & 0.20 & 0.40  & 0.60 & 0.80 & 1.00 \\\hline\hline 
			&\multicolumn{6}{c}{A. $n_X=4$} \\\hline
			Worst-case bias &  0.0021  &  0.0783 &   0.1104 &   0.1351  &  0.1560  &  0.1744\\
			Asymptotic standard error&  0.0228  &  0.0288  &  0.0297  &  0.0300 &   0.0302  &  0.0303
			\\
			Monte Carlo bias  & 0.1026 &   0.0197 &  0.0134  &  0.0111   & 0.0099   & 0.0092			
			\\
			Monte Carlo standard deviation & 0.0253  & 0.0281  &  0.0288   & 0.0291   & 0.0292  &  0.0293		
			\\
			Monte Carlo root MSE &0.1057  &  0.0343  &  0.0317  &  0.0311  &  0.0308  &  0.0307		
			\\
			CI length &0.0936  &  0.2697   & 0.3372   & 0.3878   & 0.4302  &  0.4674			
			\\
			CI coverage & 0.0180   & 0.9990   & 1.0000 &  1.0000  &  1.0000  &  1.0000				
			\\\hline	&\multicolumn{6}{c}{B. $n_X=20$} \\\hline
			Worst-case bias &  0.0021  &  0.0480   & 0.0610  &  0.0714  &  0.0805   & 0.0887
			\\
			Asymptotic standard error& 0.0227  &  0.0394   & 0.0453   & 0.0487  &  0.0509  &  0.0526
			\\
			Monte Carlo bias  & 0.0976   & 0.0080 &   0.0037  &  0.0026  &  0.0022 &   0.0020
			\\
			Monte Carlo standard deviation &0.0239 &   0.0386  &  0.0446   & 0.0480  &  0.0502   & 0.0519			
			\\
			Monte Carlo root MSE & 0.1005   & 0.0394  &  0.0447   & 0.0480   & 0.0502   & 0.0519
			\\
			CI length & 
			0.0931  &  0.2503  &  0.2996  &  0.3337   & 0.3607  &  0.3835
			\\
			CI coverage &
			0.0190   & 0.9990   & 1.0000 &   1.0000   & 1.0000  &  1.0000
			\\\hline
		\end{tabular}
	\end{center}
	\par
	\textit{{\small Notes: Performance of the minimum-MSE estimator in the cross-sectional binary choice model, for different values of $\epsilon$. $n=500$, results for $1000$ simulations. The nominal level for confidence intervals (CI) is 95\%. $n_X$ denotes the number of points of support of the first component of $X$.}}
\end{table}

\begin{table}[h!]
	\caption{Monte Carlo simulation of the average effect in the cross-sectional binary choice model, \textit{extrapolation} ($x_0=(-0.5,1)'$) \label{TabProbit_CS_2}}
	\begin{center}
		\begin{tabular}{l||ccccccc}
			Minimum-MSE, for $\epsilon=$  & 0.0001 & 0.20 & 0.40  & 0.60 & 0.80 & 1.00 \\\hline\hline 
			&\multicolumn{6}{c}{A. $n_X=4$} \\\hline
			Worst-case bias & 0.0029 &   0.1269  &  0.1794 &   0.2197   & 0.2537  &  0.2837\\
			Asymptotic standard error& 0.0296  &  0.0312  &  0.0315   & 0.0316  &  0.0316  &  0.0317
			\\
			Monte Carlo bias  & -0.0987  & -0.0903  & -0.0901  & -0.0900 &  -0.0900  & -0.0900
			\\
			Monte Carlo standard deviation &0.0283 &   0.0330  &  0.0334   & 0.0335  &  0.0336  &  0.0336
			\\
			Monte Carlo root MSE &0.1027 &   0.0961 &   0.0961   & 0.0961  &  0.0961  &  0.0961
			\\
			CI length &0.1219  &  0.3762  &  0.4822  &  0.5632  &  0.6314 &   0.6914
			\\
			CI coverage &0.2000  &  0.9370 &   0.9850  &  0.9960   & 0.9990  &  1.0000
			\\\hline	&\multicolumn{6}{c}{B. $n_X=20$} \\\hline
			Worst-case bias & 0.0028   & 0.1172  &  0.1645 &   0.2008   & 0.2314  &  0.2584
			\\
			Asymptotic standard error&0.0313 &   0.0401  &  0.0443   & 0.0470  &  0.0489  &  0.0503
			\\
			Monte Carlo bias  & -0.0902  & -0.0961   &-0.0988 &  -0.0999 &  -0.1005 &   -0.1009
			\\
			Monte Carlo standard deviation &0.0287  &  0.0373   & 0.0412  &  0.0437   & 0.0456  &  0.0471
			\\
			Monte Carlo root MSE & 0.0947   & 0.1031 &   0.1070 &   0.1090  &  0.1104   & 0.1113
			\\
			CI length & 0.1284  &  0.3915  &  0.5026   & 0.5857   & 0.6544  &  0.7141
			\\
			CI coverage & 0.2530   & 0.9500 &    0.9910 &   0.9960  &  0.9970   & 0.9970
			\\\hline
		\end{tabular}
	\end{center}
	\par
	\textit{{\small Notes: See the notes to Table \ref{TabProbit_CS_1}. }}
\end{table}

\begin{table}[h!]
	\caption{Monte Carlo simulation results for the autoregressive parameter in the dynamic binary choice panel data model\label{TabProbit_PD_1}}
	\begin{center}
		\begin{tabular}{l||cccccc}
			Minimum-MSE, for $\epsilon=$  & 0.00 & 0.20   & 0.40 & 0.60 &  0.80 &1.00 \\\hline\hline 
			&\multicolumn{6}{c}{A. $T=5$} \\\hline
			Worst-case bias &   0.0001  &  0.0179  &  0.0227 &   0.0266  &  0.0299  &  0.0327\\
			Asymptotic standard error&  0.0952  &  0.0975 &   0.0979 &   0.0981  &  0.0983 &   0.0985\\
			Monte Carlo bias  &  -0.1729  & -0.0615&   -0.0555  & -0.0531 &  -0.0518  & -0.0509\\
			Monte Carlo standard deviation & 0.1252 &   0.1111   & 0.1129&    0.1136 &   0.1141   & 0.1145
			\\
			Monte Carlo root MSE &0.2135  &  0.1270  &  0.1258   & 0.1255   & 0.1254 &    0.1253
			\\
			CI length &  0.3734  &  0.4179   & 0.4292  &  0.4379  &  0.4452 &   0.4516
			\\
			CI coverage & 0.5470  &  0.8890  &  0.9080  &  0.9160 &   0.9220  &  0.9280
			\\\hline	&\multicolumn{6}{c}{B. $T=10$} \\\hline
			Worst-case bias &0.0001  &  0.0090 &   0.0118  &  0.0140  &  0.0158  &  0.0175
			\\
			Asymptotic standard error& 0.0607  &  0.0614   & 0.0615 &   0.0616  &  0.0616 &   0.0617	\\
			Monte Carlo bias  & -0.0780  & -0.0137 &   -0.0120  & -0.0114  & -0.0110  & -0.0107
			\\
			Monte Carlo standard deviation & 0.0676 &   0.0731   & 0.0736   & 0.0738 &   0.0739   & 0.0740
			\\
			Monte Carlo root MSE & 0.1032  &  0.0744    &0.0745   & 0.0746  &  0.0747 &   0.0748
			\\
			CI length & 0.2381  &  0.2587  &  0.2647   & 0.2694 &   0.2733 &   0.2768
		\\
			CI coverage & 0.7130  &  0.9210  &  0.9330   & 0.9360   & 0.9360 &   0.9380
		\\\hline
			&\multicolumn{6}{c}{C. $T=20$} \\\hline
			Worst-case bias &   0.0001  &  0.0058 &   0.0078 &   0.0093&    0.0106 &   0.0118
			  \\
			Asymptotic standard error& 0.0418 &   0.0421   & 0.0422 &   0.0422&    0.0422   & 0.0422
			\\
			Monte Carlo bias  & -0.0304  & -0.0023  & -0.0019  & -0.0017  & -0.0017  & -0.0016\\
			Monte Carlo standard deviation &0.0442 &   0.0488 &   0.0490  &  0.0490   & 0.0491  &  0.0491
			\\
			Monte Carlo root MSE & 0.0537  &  0.0488   & 0.0490  &  0.0491   & 0.0491  &  0.0491
			\\
			CI length &0.1638  &  0.1766 &   0.1808  &  0.1840 &   0.1867  &  0.1891 \\
			CI coverage &0.8780 &   0.9110  &  0.9180  &  0.9230   & 0.9260  &  0.9300
			\\\hline		\end{tabular}
	\end{center}
	\par
	\textit{{\small Notes: Performance of the minimum-MSE estimator of $\beta_0$ in the dynamic panel data binary choice model, for different values of $\epsilon$. $n=500$, results for $1000$ simulations. The nominal level for confidence intervals (CI) is 95\%. }}
\end{table}

\begin{table}[h!]
	\caption{Monte Carlo simulation results for the average state dependence parameter in the dynamic binary choice panel data model\label{TabProbit_PD_2}}
	\begin{center}
		\begin{tabular}{l||cccccc}
			Minimum-MSE, for $\epsilon=$  & 0.00 & 0.20   & 0.40 & 0.60 &  0.80  &1.00\\\hline\hline 
			&\multicolumn{6}{c}{A. $T=5$} \\\hline
			Worst-case bias &    0.0000  &  0.0099 &    0.0134 &   0.0162 &    0.0185 &    0.0205
			\\
			Asymptotic standard error& 0.0259  &  0.0268  &  0.0270 &   0.0272 &   0.0273 &   0.0274
			\\
			Monte Carlo bias  & -0.0538 &  -0.0218 &  -0.0202  & -0.0196   &-0.0193  & -0.0191
			\\
			Monte Carlo standard deviation &0.0439 &   0.0324  &  0.0331 &   0.0334 &   0.0336   & 0.0337
			\\
			Monte Carlo root MSE &0.0694  &  0.0391  &  0.0387   & 0.0387  &  0.0387   & 0.0388
			\\
			CI length &0.1017  &  0.1250  &  0.1329 &   0.1389   & 0.1439  &  0.1483\\
			CI coverage &  0.4450  &  0.8620  &  0.8850  &  0.9000  &  0.9190 &   0.9240
			\\\hline	&\multicolumn{6}{c}{B. $T=10$} \\\hline
			Worst-case bias & 0.0000 &   0.0121  &  0.0169 &   0.0207   & 0.0238 &    0.0266
			\\
			Asymptotic standard error& 0.0181  &  0.0184   & 0.0185  &  0.0186 &   0.0186  &  0.0187
			\\
			Monte Carlo bias  &  -0.0212  & -0.0047  & -0.0048  & -0.0050  & -0.0051  & -0.0052
			\\
			Monte Carlo standard deviation &  0.0257 &   0.0229 &   0.0230   & 0.0231   & 0.0231 &   0.0232
			\\
			Monte Carlo root MSE &  0.0333  &  0.0233   & 0.0235   & 0.0236  &  0.0237   & 0.0238
			\\
			CI length & 0.0710  &  0.0963 &   0.1063  &  0.1141  &  0.1206  &  0.1263
			
			\\
			CI coverage &0.6610 &   0.9630 &   0.9780 &   0.9830 &   0.9870  &  0.9880
		\\\hline
			&\multicolumn{6}{c}{C. $T=20$} \\\hline
			Worst-case bias & 0.0000  &  0.0163 &    0.0230   & 0.0281  &  0.0325 &   0.0363
			\\
			Asymptotic standard error& 0.0134 &   0.0135  &  0.0136&    0.0137 &   0.0137 &   0.0138
			\\
			Monte Carlo bias &  -0.0097 &  -0.0028 &  -0.0028  & -0.0028  & -0.0028 &  -0.0028	\\
			Monte Carlo standard deviation &  0.0187 &   0.0153  &  0.0153  &  0.0154  &  0.0154    & 0.0155
			\\
			Monte Carlo root MSE &  0.0210   & 0.0155 &    0.0156   & 0.0156    &0.0157   & 0.0157	\\
			CI length &  0.0525 &   0.0857 &   0.0993  &  0.1098 &   0.1187  &  0.1265
				\\
			CI coverage & 0.7840 &   0.9890  &  0.9930  &  0.9960   & 0.9970  &  0.9970 	\\\hline		\end{tabular}
	\end{center}
	\par
	\textit{{\small Notes: Performance of the minimum-MSE estimator of $\delta_{\beta_0,\pi_0}$ in the dynamic panel data binary choice model, for different values of $\epsilon$. $n=500$, results for $1000$ simulations. The nominal level for confidence intervals (CI) is 95\%. }}
\end{table}

\end{document}